\documentclass[aps,prd,twocolumn,showpacs, reprint, nofootinbib,floatfix]{revtex4-2}

\usepackage[T1]{fontenc} 
\usepackage{microtype}
\usepackage{fontawesome5} 
\usepackage{soul}

\usepackage[dvipsnames]{xcolor}
\usepackage{graphicx}
\usepackage{caption}
\usepackage{subcaption}

\usepackage{amssymb}
\usepackage{amsthm, mathtools} 
\usepackage{amsfonts, bm, bbm}
\usepackage{tensor}
\DeclareMathOperator{\tr}{tr}

\usepackage{mathrsfs}

\usepackage[bottom]{footmisc}

\usepackage{enumitem}
\setlist[enumerate]{label=(\emph{\roman*})}

\numberwithin{equation}{section}

\usepackage{hyperref}
\hypersetup{
    colorlinks,
    citecolor=cyan,
    linkcolor=NavyBlue,
    filecolor=magenta,      
    urlcolor=NavyBlue,
    breaklinks=true,
    pdftitle={General formalism, classification and demystification of warp-drive spacetimes},
	pdfsubject={},
	pdfauthor={Hamed Barzegar}
}
\usepackage[capitalise]{cleveref}
\makeatletter
\AddToHook{cmd/appendix/before}{\def\cref@section@alias{appendix}\def\cref@subsection@alias{appendix}} 
\makeatother

\usepackage{tikz, tikz-network}
\usetikzlibrary{mindmap,trees,backgrounds, positioning}

\usepackage{thmtools}
\newtheorem{theorem}{Theorem}[section]
\newtheorem{lem}[theorem]{Lemma}

\newtheorem{Def}[theorem]{Definition}
\newtheorem{cor}[theorem]{Corollary}
\newtheorem{principle}[theorem]{Principle}

\newtheoremstyle{mystyle}       
  {}                            
  {}                            
  {\normalfont}                 
  {}                            
  {\bfseries}                   
  {.}                           
  { }                           
  {}                            

\theoremstyle{mystyle}   
\newtheorem{mytheorem}{Theorem}
\newtheorem{error}[mytheorem]{Error}
\newtheorem{rem}[theorem]{Remark}

\definecolor{mygreen}{rgb}{0,0.5,0}
\definecolor{mybloodyred}{rgb}{0.7,0,0}

\newcommand{\albe}{{\alpha\beta}}
\newcommand{\mnu}{{\mu\nu}}
\newcommand{\R}{\mathbb R}
\newcommand{\CA}{\mathcal A}
\newcommand{\CC}{\mathcal C}
\newcommand{\CD}{\mathcal D}
\newcommand{\CL}{\mathcal L}
\newcommand{\CK}{\mathcal K}
\newcommand{\CR}{\mathcal R}

\newcommand{\CO}{\mathcal O}
\newcommand{\p}{\partial}
\newcommand{\const}{\mbox{\rm const}} 
\newcommand{\mfd}{\mathcal{M}}

\newcommand{\metg}{\bm{g}}
\newcommand{\meth}{\bm{h}}
\newcommand{\metb}{\bm{b}}

\newcommand{\hnorm}[1]{\left|#1\right|_{\meth}}
\newcommand{\dnorm}[1]{\left|#1\right|_{\bm{\delta}}}
\newcommand{\dual}[1]{{\bm #1}^\flat}
\newcommand{\I}{{\rm I}}
\newcommand{\II}{{\rm II}}
\newcommand{\III}{{\rm III}}
\newcommand{\Ik}{{\rm I}[\bm \CK]}
\newcommand{\IIk}{{\rm II}[\bm \CK]}
\newcommand{\IIIk}{{\rm III}[\bm \CK]}
\newcommand{\diff}{{\mathrm d}}
\newcommand{\diffb}{{\mathbf d}}

\newcommand{\ddt}[1]{\frac{{\mathrm d} #1}{{\mathrm d} t}}

\newcommand{\LieN}{\CL_{\bm N}}
\newcommand{\LieV}{\CL_{\bm V}}

\newcommand{\Etot}{E_{\text{tot}}}
\newcommand{\cc}{\mathrm{\Lambda}}
\newcommand{\Eadm}{E_{\text{ADM}}}
\newcommand{\Padm}{P_{\text{ADM}}}
\newcommand{\Rwarp}{\textbf{R}-Warp}
\newcommand{\Twarp}{\textbf{T}-Warp}
\newcommand{\Swarp}{\textbf{S}-Warp}
\newcommand{\natario}{Nat\'ario}
\newcommand{\alcub}{Alcubierre}

\newcommand{\nn}{\nonumber}

\usepackage{orcidlink}

\begin{document}

\title{General formalism, classification, and demystification \\ of the current warp-drive spacetimes}

\author{Hamed Barzegar \orcidlink{0000-0002-6472-3115}} 

\email[Corresponding author: ]{hamed.barzegar@ens-lyon.fr}

\homepage[\\The author dedicates this work to the memory of those brave Iranian individuals who resisted the Ahriman and his forces with their bare hands in the fight for freedom.]{}

\author{Thomas Buchert \orcidlink{0000-0002-0828-3901}}
\email{thomas.buchert@ens-lyon.fr}

\author{Quentin Vigneron \orcidlink{0000-0002-7494-0455}}
\email{quentin.vigneron@ens-lyon.fr}

\affiliation{
ENS de Lyon, CRAL UMR5574, Universite Claude Bernard Lyon 1, CNRS, Lyon 69007, France
}

\date{February 18, 2026}

\begin{abstract}
We critically examine proposals for the so-called warp-drive spacetimes and classify these models based on their various restrictions within the framework of General Relativity. We then provide a summary of general formalism for each class, and in the process, we highlight some misconceptions, misunderstandings, and errors in the literature that have been used to support claims about the physicality and feasibility of these models. On the way, we prove several new no-go theorems. Our analysis shows that when the principles of General Relativity are applied correctly, most claims regarding physical warp drives must be reassessed,
and it becomes highly challenging to justify or support the viability of such models, not merely due to the violation of energy conditions.
\end{abstract}

\maketitle

\section{Introduction}
\label{sec: intro}

The human drive to explore and travel faster meets a fundamental barrier in Special Relativity (SR): the \emph{light-cone barrier}. However, this has not stopped physicists to explore possible workarounds for fast and eventually faster-than-light (FTL) travels. This goes back to efforts within SR, started, to our knowledge, first by \citet{1957_McMillan_ClockParadoxSpace},\footnote{McMillan was an American physicist who shared the 1951 Nobel Prize in Chemistry with Glenn Seaborg for the discovery of the neptunium, whereas his nephew, John Clauser, was awarded the 2022 Nobel Prize in Physics, jointly with Alain Aspect and Anton Zeilinger.}
using the relativistic effects of the so-called \emph{hyperbolic motion} resulted from a uniform acceleration. In his example, \citet{1957_McMillan_ClockParadoxSpace} shows that a uniformly accelerated observer (spaceship) can travel ``1.2 billion light-years in only 21 years'' as measured by an observer on, say, Earth (see also ruminations by \citet{Mumford_2021_AMS}).
This, later, motivated \citet{1960_Rindler_hyperbolic} to generalize the hyperbolic motion to the curved spacetime, after whom the \emph{Rindler coordinates} are coined. However, the apparent superluminal travel is a consequence of time dilation of the uniformly accelerated observers, and not the breaking of the light-cone barrier, and both conclude that such travels are impossible because insurmountable energy costs and cosmological horizon render the prospect of such travel ``slight'' and ``cosmologically quite insignificant'' \cite{1960_Rindler_hyperbolic,1957_McMillan_ClockParadoxSpace}.

The idea of FTL travel entered General Relativity (GR) through the concept of \emph{wormholes}, based on the works by Flamm \cite{1916_Flamm, *2015_Flamm_Republication}, who first described a ``tunnel'' geometry joining two asymptotically flat surfaces of spacetime, and \citet{1935_Einstein_ParticleProblemGR}, who proposed a ``bridge'' as a model for elementary particles. This early foundation led to the work of \citet{1955_Wheeler_Geons}, who explored multiply-connected spacetimes through hypothesized \emph{geons}, and the paper by \citet{1957_Misner_PhysicsGeometry} that first coined the term ``wormhole'' to describe mass and electric charge as properties of empty curved space (geometry).
These conceptual developments eventually led to the first specific models for traversability in 1973, such as Homer Ellis's ``drainhole'' \cite{1973_Ellis_Drainhole} and Bronnikov's tunnel-like solutions \cite{1973_Bronnikov_scalar}. A modern ``renaissance'' in the field was triggered by \citet{1988_Thorne_Morris_Wormholes}, who defined the formal requirements for traversable wormholes to be used in interstellar travel, specifically identifying the need for exotic matter that violates the null energy condition (NEC) to hold the throat open (see \cite{1996_Visser_LorentzianWormhole, 2017_book_Lobo, Lobo_2007_exotic} for a more detailed account of the historical development).

Following the proposal of wormholes, the next step in conception of FTL travels in GR was the invention of the \emph{warp drive} in 1994 by \citet{Alcubierre_1994_warp},\footnote{Alcubierre never pursued his own idea further (except from the review work \cite{Alcubierre_2017_warp_basics}), which may indicate that he was in agreement with the philosophical perspective outlined in \Cref{sec:conclusion}, a view we also share.} which allows for superluminal motion, supposedly, by contracting spacetime in front of a bubble and expanding it behind (see \Cref{sec: Alcubierre} below).
A motivation for Alcubierre's idea can be drawn from cosmology, specifically the inflationary phase of the early universe that causally connects two observers who should not have been in causal contact in a standard Big-Bang scenario without inflation \cite{Alcubierre_1994_warp, 2024_Barzegar_Buchert_Letter}. 
The main characteristic idea of Alcubierre is that, in contrast to the cosmological case, only a localized part of the universe would causing such effect.
Because this original model required physically impossible amounts of exotic matter, \citet{Broeck_1999_warp} proposed a modification that significantly reduced the (supposed) necessary negative energy density (see \Cref{sec: van den Broeck} below). \citet{Natario_2002_warp} advanced the theory by first generalizing the one-component velocity to three-component one, and introducing a zero-expansion warp drive, suggesting that the bubble could essentially ``slide'' through space by changing distances along the path of motion rather than relying on volume expansion (see \Cref{sec: Natario} below). 

Several theoretical analyses highlight fundamental pathologies that suggest warp drives are physically or technologically improbable: from the violation of energy conditions indicated first by \citet{Alcubierre_1994_warp} himself, to application of quantum inequalities by \citet{Pfenning_1997_unphysical} showing that a superluminal warp bubble would require an absurdly large amount of negative energy (which was significantly reduced in the subsequent works with further considerations)
while restricting bubble walls to nearly the Planck scale, and the ``horizon'' problem identified by \citet{1995_Krasnikov_Hyperfast} and \citet{Everett_1997_subway}, demonstrating that a crew cannot create or control a superluminal bubble on demand, and finally the semiclassical analysis done by \citet{Liberati_2009_semiclassical} revealing a fatal instability in which Hawking radiation accumulates at the bubble's white horizon, leading to an exponentially increasing energy density that would likely destroy the structure (and more problems which will be addressed in this work; see also \cite{2017_book_Lobo, Lobo_2007_exotic} for more details).

Despite all these pathologies, recent works starting with \citet{2021_Lentz_breakingbarrier} followed by \citet{2021_Bobrick_physicalwarp} and \citet{2021_FellHeisenberg_positive} (and others) claim that ``physically reasonable'' warp-drive configurations can be constructed with positive energy density, and even more with respecting the energy conditions. \citet{2021_Lentz_hyperfast} proposes using \emph{solitons} in an ``Einstein--Maxwell-plasma'' theory to achieve superluminal travel using only conventional matter sources like perfect fluid plasma. Fell and Heisenberg argue that ``hidden geometric structures'' in the spacetime curvature can provide a positive energy density for comoving Eulerian observers. Additionally, \citet{2021_Bobrick_physicalwarp} claim that specific classes of spherically symmetric or axisymmetric warp drives can satisfy the weak energy condition (WEC) or minimize its violation. 

All these (and other) claims have been refuted correctly (fully or partially) by \citet{2021_Visser_genericwarp}, which is one of the best work in this direction, and to which we refer the reader frequently whenever we will repeat a similar argument (we also refer the reader to this work for further references that are not cited in the present work).\footnote{Unfortunately, despite what has been said, this work  suffers from some (minor) errors too, e.g., Errors~\ref{err:ADMtotMass} and \ref{er:VisserShoshanyTotalMass}.} However, some of their correct critiques mentioned by \citet{2021_Visser_genericwarp} have been ignored by the supporters of ``physical'' warp drive, in particular in the successive works of \citet{2021_Bobrick_physicalwarp}, i.e., by \citet{2024_Fuchs_constant}, \citet{2024_Helmerich_Bobrick_WarpFactory}, and by \citet{Abellan_2023_anisotropic,Abellan_2023_spherical}, \cite{2024_Abellan_SphericalWarpbasedBubble}, and \citet{2025_Bolivar_piecewise}, and many more.

Therefore, the present work has its objective to systematically analyze those misconceptions, misunderstandings, and mistakes which have not been observed by \citet{2021_Visser_genericwarp} or any other works in the literature, and by doing so, we will also review the works mentioned above that are published after publication of \cite{2021_Visser_genericwarp}, hence our work can be regarded as a complementary work to theirs. This follows in three steps that has been reflected in the title: $(i)$ general formalism, \Cref{sec: framework}, where we lay down the general formalism needed to study such spacetimes and list the restrictions that are imposed on such spacetimes, $(ii)$ classification, \Cref{sec:classification}, where we give a classification based on the metric functions and some restrictions, and finally, $(iii)$ demystification, \Cref{sec:demyst}, where we critically review some severe errors that appeared in the literature. Hence, this section has both negative and positive approaches in that it not only discloses errors (cf.~\Cref{sec:demyst}), it also reveals some geometric structures behind such spacetimes. This is done by clarifications in terms of several Remarks and new theorems, e.g., Theorems \ref{thm:AlcOmegaZero}, \ref{thm:globhyp}, \ref{thm: adm mass of natario}, \ref{thm: positive mass theorem}, \ref{thm: natario WEC}, \ref{thm: zero expansion DEC}, \ref{thm:shearfree}, \ref{thm:harmonicwarp}, but also through Errors~\ref{er:alcgh}--\ref{er:fuchsgeodesic}. The demystification shows that while most of the literature (in particular the recent works which claim to propose physical warp-drive spacetimes) considers the need for negative energy (i.e., violation of the WEC) and its needed amount as the main problem of the warp-drive spacetimes, there are even more fundamental issues that are not necessarily related to the violation of energy conditions. These include in particular the notions of global hyperbolicity, the asymptotic flatness, the notion of mass and energy especially the ADM energy, the Synge-G method  and their consequences.

Thus, in a sense, \Cref{sec:demyst} on demystification of the current warp-drive spacetimes is the heart of the present paper. Here, we take an \emph{agonistic} approach (see \cite[Chapter~IV, \S 6]{1971_Synge_GR} and \Cref{sec: Synge} below), i.e., we ``wrestle'' with mathematical problems and consistency---most of the time even prior to solving the field equations.
In fact, we follow a basic scientific principle (see, e.g, \citet{1946_Planck_Scheinproblem}), which is roughly reminiscent of the logical positivism (see, e.g., \citet{1934_Carnap_Syntax}), that we frame as follows.
\begin{principle}\label{principle}
For any proposal presented within a theory, its formal (syntactic) structure must be precisely defined, and it must demonstrate consistency within that theory to prevent ungrounded speculations that could otherwise lead to the emergence of pseudo-problems \textup{(Scheinprobleme)}.
\end{principle}
The theory mentioned in \Cref{principle} in this work is taken to be GR.
In the literature, there are many speculations around warp-drive spacetimes, although many of them do not satisfy \Cref{principle}. Before revealing those inconsistencies, it may not make sense to go beyond and ponder properties of those spacetimes.
Hence, as a consequence of \Cref{principle}, we do not consider: \emph{(i)} quantum theoretical aspects of such spacetimes, such as quantum inequalities, \emph{(ii)} the horizon problem (except from the discussion around the global hyperbolicity in \Cref{sec: global hyperbolicity}), \emph{(iii)} proposals in theories beyond GR (with only exception being mentioned in \Cref{er:cont_eq_beyon_GR} below).\footnote{For investigations of warp-drive proposals beyond general relativity, we refer the reader to, e.g., \cite{Carneiro_2022_conservation, DeBenedictis_2018_Einstein_Cartan, Varieschi_2012_conformal}.} 
Moreover, in view of the public fascination with warp drive in popular culture and, occasionally, within the scientific literature and some space agencies, the present study restricts itself to peer-reviewed publications up to the date of submission, with emphasis on those contributions that assert the proposal of a physical warp-drive spacetime.
In fact, one can even claim that many of these works in this field violate another principle beyond \Cref{principle}, which we may call \emph{Thorne's principle for a good science-fiction} that was formulated as guidelines in \cite[Section~1]{2014_Thorne_interstellar} on the science behind the movie \emph{Interstellar} \cite{2024_Interstellar}. We do not list works that violate this principle fully or partially, and leave it to the reader.

This paper is organized as follows. Before proceeding with the \Cref{sec: framework} where we introduce the general formalism, we introduce the notations which we will follow throughout this work. Then, in \Cref{sec:classification} we successively apply the restrictions listed there to classify the current warp-drive spacetimes. In \Cref{sec:demyst}, we critically analyze misconceptions and mistakes in terms of Errors and prove some theorems. \Cref{sec:conclusion} finally concludes the paper with some discussions and conclusions.

\subsection{Notations}\label{sec: notations}
We denote by $\left( \mfd, \metg \right)$ a four-dimensional Lorentzian manifold where $\metg$ is a metric of signature $\left(-, +, +, +\right)$. 
$\mathfrak{X}\left( \mathcal{N} \right)$
and ${\bf T}^p_q \left( \mathcal{N} \right)$ denote the vector space of all smooth vector fields
and the space of all smooth $(p,q)$-tensor fields, respectively, on any given manifold $\mathcal{N}$.
Bold letters are used for coordinate-free notation of a given $(p,q)$-tensor ${\bf A} \in {\bf T}^p_q \left( \mathcal{N} \right)$ whose components are given in local coordinate bases $\{\bm \diff x^\alpha, \bm \p_\alpha \}$.  
Greek indices denote the spacetime coordinates, whereas Latin indices stand for spatial coordinates.
Accordingly, $\bm \nabla$ and $\bm D$ denote the associated covariant derivatives on $\mfd$ and $\Sigma_t$, respectively. 
The symbol $\boldsymbol{\partial}$ stands for the usual Euclidean three-dimensional nabla operator.
We occasionally use the notation $\mathbf{g} (\bm X, \bm Y)$ for the inner product of vector fields $\bm X$ and $\bm Y$ with respect to a given metric $\mathbf{g}$.
Therefore, $|\cdot|_{\mathbf{g}}$ denotes the norm of a given tensor with respect to the metric $\mathbf{g}$, i.e., for a tensor $\bm A$ on $\mfd$ we have $|\bm{A}|_{\mathbf{g}}^2 := \mathrm{g}^{\mu_1 \nu_1} \cdots \mathrm{g}^{\mu_n \nu_n} A_{\mu_1 \cdots \mu_n} A_{\nu_1 \cdots \nu_n}$.  
$\Delta_{\mathbf{g}}$ stands for the Laplace--Beltrami operator with respect to any metric $\mathbf{g}$
and $\CL_{\bm X}$ denotes the Lie derivative along a given vector field $\bm X$.
We use the musical symbol $\flat$ (for the musical isomorphism) to denote the metric-dual one-form of a given vector with respect to a given metric $\mathbf{g}$, i.e., for a generic vector field $\bm w \in \mathfrak{X}(\mathcal{N})$ we have $\dual{w} = {\mathbf{g}} \left(\bm w, \cdot\right)$ on a manifold $(\mathcal{N}, \mathbf{g})$.
For any given  2-tensor fields $\bm B$ on $(\Sigma_t, \meth)$ (the hypersurfaces), we denote by $\I \left[ \bm B \right]$, $\II \left[\bm B \right]$, $\III \left[ \bm B \right]$ the principal scalar invariants of $\bm B$ which can be defined as the principal scalar invariants of the matrix $\tensor{B}{^i_j} := \tensor{\left( \bm B \right)}{^i_j}$ in a local chart as follows:
\begin{align*}
    \I [{\bm B}]
    &:=
    \tr_{\meth} {\bm B} = \tensor{B}{^i_i}
\,,
\\
    \II [{\bm B}]
    &:=
    \frac{1}{2}
    \left(
        \tensor{B}{^i_i} \tensor{B}{^j_j}
        -
        \tensor{B}{^i_j} \tensor{B}{^j_i}
    \right)
    =
    \tensor{\II[\bm B]}{^i_i}
,
\\
    \III [{\bm B}]
    &:=
    \det[{\bm B}]
\\
    &=
    \frac{1}{6}
    \left(
        \tensor{B}{^i_i} \tensor{B}{^j_j}
        \tensor{B}{^k_k}
        +
        2
        \tensor{B}{^i_j} \tensor{B}{^j_k}
        \tensor{B}{^k_i}
        -
        3
        \tensor{B}{^i_i} \tensor{B}{^j_k}
        \tensor{B}{^k_j}
    \right)
\\
    &=
    \frac{1}{3}
    \left(
        \tensor{B}{^i_j} \tensor{B}{^j_k} \tensor{B}{^k_i}
        -
        \I^3 [{\bm B}]
        +
        3
        \I [{\bm B}] \, \II [{\bm B}]
    \right)
\,,
\end{align*}
where $\II[\bm B]_{ij} := \left( \I[{\bm B}] B_{ij} - \tensor{B}{_i^k} B_{kj} \right)/2$.
The round and square brackets imply symmetrization and anti-symmetrization over the indices enclosed, respectively, and the angle brackets denote the trace-free part of a given tensor.  We use units in which the speed of light is set to $c = 1$.

\section{General formalism and restrictions}
\label{sec: framework}

In this section we outline the general $3+1$ formulation, thoroughly discussed in \cite{2020_BMR_III, 2012_Gourgoulhon_formalism, 2008_Alcubierre_3+1, 1978_Smarr_kinematical, Baumgarte_Shapiro_2010_numerical} and \cite{2024_Barzegar_Buchert_Letter}; we essentially summarize the needed formalism and formulas from \cite{2020_BMR_III} (see also \cite{2024_Barzegar_Buchert_Letter}), and to maintain brevity, we refer the 
reader to \cite{ehlers,2012_Gourgoulhon_formalism, 2008_Alcubierre_3+1, 1978_Smarr_kinematical, 1984_Wald_GR} for further details. We refer to these references sparingly to avoid repetition throughout this section.

\subsection{Foliations, metric and the geometry of motion}

We consider a  four-dimensional Lorentzian manifold $\left( \mfd, \metg \right)$, where $\metg = g_{\mu \nu} \, \diffb x^\mu \otimes \diffb x^\nu$ in local coordinates $\left( x^\mu \right) = \left( t, x^i\right) \equiv \left( t, \bm x \right)$ with its Levi-Civita connection $\bm \nabla$, and $\mfd \cong \Sigma_t \times I$, where $\Sigma_t$, with a global time function $t \in I \subset \R$, is a Riemannian $3$-manifold. In other words, $\Sigma_t$ denotes a family of spacelike hypersurfaces that foliate the spacetime $\mfd$.

The matter model considered in this work is a single general fluid with the  $4$-velocity $\bm{u}$ that describes the unit, future-oriented, timelike tangent $4$-vector field of a timelike congruence, i.e., the fluid flow. 
The unit, timelike, future-oriented normal vector field of $\Sigma_t$ is denoted by $\bm{n}$ that satisfies ${\metg} \left( \bm{n},  \bm n\right) = -1$. Observers with the $4$-velocity $\bm n$ are called \textit{Eulerian}.

Using $\bm n$, we construct the bilinear form $\meth = h_{\alpha \beta} \, \mathbf{d} x^\alpha \otimes \mathbf{d} x^\beta$, and the projector $\tensor{h}{^\alpha_\beta}$, 
which projects spacetime tensors onto the hypersurfaces of the foliation with the following form and properties:
\begin{equation}\label{eq:proj_h}
\begin{aligned}
	h_{\mu \nu} &:= g_\mnu + n_\mu n_\nu
\,;
\\
	h_{\alpha \mu} n^\alpha &= 0 \, , \quad 
	\tensor{h}{^\mu_\alpha} \tensor{h}{^\alpha_\nu} = \tensor{h}{^\mu_\nu} \, , \quad 
	h^{\alpha\beta} h_{\alpha\beta} = 3 \, .
\end{aligned}
\end{equation}
$\meth$ defines the induced Riemannian metric (the first fundamental form) $h_{ij}$ on the hypersurface $\Sigma_t$ with the inverse denoted by $h^{ij}$, and its associated covariant derivative $D_\mu$ can be expressed in terms of the covariant derivative associated to $\metg$ as follows:
\begin{align*}
    D_\mu \tensor{F}{^{\alpha_1 \ldots \alpha_p} _{\beta_1 \ldots \beta_q}}
    &=
    \tensor{h}{^{\alpha_1} _{\rho_1}} \cdots \tensor{h}{^{\alpha_p} _{\rho_p}} \tensor{h}{_{\beta_1} ^{\sigma_1}} \cdots
\\
&\quad
    \cdots \tensor{h}{_{\beta_q} ^{\sigma_q}}
    \tensor{h}{_\nu^\mu} 
    \nabla_\nu \tensor{F}{^{\rho_1 \ldots \rho_p} _{\sigma_1 \ldots \sigma_q}}
\,,
\end{align*}
for a given tensor field $ \bm F \in {\bf T}^p_q \left( \mfd \right)$.

Another important vector field is the so-called \textit{time vector} $\bm \p_t$, whose integral curve keeps the same spatial coordinates that it attains on the initial hypersurface satisfying, ${\metg} (\bm \p_t, \bm \nabla t) = 1$. The time vector represents the flow of time throughout the spacetime and can be decomposed into vectors normal and tangent to $\Sigma_t$ as
\begin{equation}\label{eq: time vector}
    \bm{\partial}_t 
    = 
    N \bm n + \bm N
\,,
\end{equation}
where $N$ and $\bm N$ are the \textit{lapse function} and the \textit{shift vector}, respectively defined by $N := - {\metg}(\bm n, \bm \p_t)$ and $N_i := h_{ij} (\bm \p_t)^j$ with ${\metg}(\bm n, \bm N) = 0$. The time vector can, \emph{a priori}, be null, space- or time-like, and in general it is not normalized. A natural class of observers are the ones associated to the time vector $\bm \p_t$ whenever it is timelike, which we call the \emph{coordinate observers} \cite{2024_Barzegar_Buchert_Letter},\footnote{Coordinate observers are sometimes referred to as \emph{static observers} in the static (or stationary) context.} by normalizing the time vector $\bm \p_t$ having the $4$-velocity\footnote{An ambiguity appeared in \cite[Equation~(11)]{2024_Barzegar_Buchert_Letter}: the timelike character of such observers is assumed there without being explicit, in which case the absolute value enclosing the prefactor of $\bm \p_t$ is not necessary.}
\begin{equation}
    \hat{\bm \p}_t 
    := 
    \left( N^2 - N^k N_k \right)^{-1/2} \bm{\partial}_t
\,.
\end{equation}

Now, we can express the unit normal vector $\bm n$ and its metric dual form $\dual{n}$ with respect to coordinate basis $\{ \bm \p_\mu \}$ using \eqref{eq: time vector} as
\begin{equation}\label{eq:n_vec}
	\bm{n} = \frac{1}{N} \left( 1, - \bm{N} \right)
\,,
\quad
	\dual{n} = - N \, ( 1, \bm 0 ) 
\,,
\end{equation}
whose $4$-acceleration is given by (cf., e.g., \cite{2012_Gourgoulhon_formalism, 1978_Smarr_kinematical, 2020_BMR_III})
\begin{equation}\label{eq:naccel}
    \mathfrak{a}^\mu
    :=
    n^\alpha \nabla_\alpha n^\mu = 
    h^{\mu \nu} D_\nu \ln N
\,.
\end{equation}

Finally, using all relations above, one can decompose the four-dimensional line element into the Arnowitt--Deser--Misner (ADM) form (cf.~\cite{2008_ADM}),
\begin{equation}\label{eq:line_elem}
\begin{aligned}
	{\mathrm d}s^2 
	&= 
    g_{\alpha\beta} \, {\mathrm d}x^\alpha {\mathrm d}x^\beta
\\
	&= 
    - \left( N^2 - N^k N_k \right) {\mathrm d}t^2 
    + 
    2 N_i \, {\mathrm d}x^i \, {\mathrm d}t \,
	+ 
    h_{ij} \, {\mathrm d}x^i {\mathrm d}x^j \, ,
\end{aligned}
\end{equation}
where $N^i = h^{i j} N_j$. 

Moreover, the second fundamental form (or the extrinsic curvature) ${\bm \CK} = \mathcal{K}_{\mnu} \, \diffb x^\mu \otimes \diffb x^\nu$ is  defined by
$ {\bm \CK} \left( \bm W, \bm Z \right) := \metg \left( \bm n, \bm \nabla_{\bm W} \bm Z \right) = - \metg \left( \bm \nabla_{\bm W} \bm n, \bm Z \right)$, for all $\bm W, \bm Z \in \mathfrak{X}\left( \Sigma_t \right)$
which in components reads
\begin{equation}\label{eq: extrinsic curvature}
    \CK_{\mnu} 
    = 
    - \tensor{h}{^\alpha_\mu} \tensor{h}{^\beta_\nu} \nabla_\alpha n_\beta
    =
    - \nabla_\mu n_\nu 
    -
    n_\mu D_\nu \ln N
\,,
\end{equation}
whose trace is the mean curvature $\CK := \tr_{\meth} {\bm \CK} = h^{i j} \CK_{ij} = - \nabla_\mu n^\mu$.

The fluid can, in general, have vorticity. Hence, the spacetime \textit{cannot}, by the Frobenius theorem,  be foliated in such a way that the unit normal vector is the $4$-velocity of the fluid, which is the general assumption for what we call ``restricted warp drives'' (see \cite{2024_Barzegar_Buchert_Letter}). 

Now, we consider a general fluid with the $4$-velocity $\bm{u}$ that is, in general, tilted with respect to the unit normal vector field $\bm n$, which can be decomposed as follows (cf., e.g., \cite{2020_BMR_III,2012_Gourgoulhon_formalism,Wainwright_Ellis_1997, 2024_Barzegar_Buchert_Letter}):
\begin{equation}\label{eq:four_vel}
{\bm u} = \gamma ({\bm v}) \left( \bm n + \bm v \right) 
\, , 
\end{equation}
with $n^\mu v_\mu = 0$ and
\begin{equation}\label{eq:lorentz}
    \gamma ({\bm v}) 
    := 
    - n^\mu u_\mu  = \left( 1 - v^\mu v_\mu \right)^{-1/2} 
\,,
\end{equation}
where $\bm v = v^i \bm \p_i$ 
is the \textit{covariant spatial velocity} field. It is a measure of tilt between the unit normal and fluid frame that is equal to the \textit{covariant Lorentz factor} $\gamma ({\bm v})$ (cf.~\cite{2024_Barzegar_Buchert_Letter}).
In the \emph{flow-orthogonal} case, where $\bm u = \bm n$, we have $\bm v = \bm{0}$ and $\gamma (\bm 0) = 1$, hence no tilt.

Another important vector is the spatial fluid's coordinate velocity field  (referred to as \textit{coordinate velocity}) defined by
\begin{equation}
	\bm V := \frac{{\mathrm d} \bm x}{{\mathrm d}t} \, , \quad 
	\mathrm{with} \quad 
    n_\mu V^\mu = 0
\,,
	\label{eq:coord_vel}
\end{equation}
where $\diff \bm x$ is the  displacement of the spatial position of a fluid element in the coordinate system $(t, \bm x)$, hence it is (spatially) coordinate dependent, in contrast to the covariant  spatial velocity. The coordinate velocity $\bm V$ should not be confused with the covariant  velocity $\bm v$ introduced in \eqref{eq:four_vel}  which can be written  as (cf.~\cite{2024_Barzegar_Buchert_Letter})\footnote{The reader must take care about the notation as $v^i$ or $v$ is usually used in the literature for the coordinate velocity in this context.}
\begin{equation}\label{eq:relat_vel}
	\bm v = \frac{1}{N} \left( \bm N + \bm V \right) \, .
\end{equation}
Hence,
\begin{equation}\label{eq:four_vel_II}
	\bm u = \frac{\gamma ({\bm v})}{N} \left( N \bm n + \bm N + \bm V \right) \, , 
\end{equation}
where
\begin{equation}
	\frac{\gamma({\bm v})}{N} = \left( N^2 - \hnorm{\bm N + \bm V}^2 \right)^{-1/2}
\,.
\end{equation}
From \eqref{eq:four_vel_II} we obtain the component expressions for $\bm u$ and $\dual{u}$:
\begin{equation}\label{eq:comp_u} 
\begin{aligned}
    \bm u 
    &= 
    \frac{\gamma ({\bm v})}{N} \left( 1, \bm V \right) 
\,, 
\\
	\dual{u} 
    &= 
    \frac{\gamma({\bm v})}{N} \left( - N^2 + N^i \left( N_i + V_i \right), \, \dual{N} + \dual{V} \right) 
\,,
\end{aligned}
\end{equation}
where, as mentioned in \Cref{sec: notations}, we used the notation $ \dual{X} := {\meth} \left( \bm X, \cdot \right)$ for all vector fields $\bm X \in \mathfrak{X}\left( \Sigma_t \right)$. It is easy to show that
\begin{equation}\label{eq:Vtoui}
    V^i
    =
    \frac{u^i}{u^0}
\,.
\end{equation}
One can now specify different coordinate systems, e.g.,
\begin{enumerate}
    \item for a coordinate system comoving with the fluid, which corresponds to a specific shift, we have 
$\bm V = \bm 0$,
    \item for a vanishing tilt we have 
    \begin{equation}\label{eq:condfloworth}
        \bm V 
        = 
        - \bm N    
    \end{equation}
    which results from \eqref{eq:relat_vel}, independent of the choice of the shift vector.
\end{enumerate}
Choosing $\bm n = \bm u$ (hence $\bm v = \bm 0$ and $\gamma({\bm 0}) = 1$) 
corresponds to a foliation orthogonal to the fluid flow which is only possible if the fluid flow has no vorticity.
Even for irrotational fluids, introducing a tilt allows us 
to keep the freedom in the construction of the spatial hypersurfaces.
\begin{rem}
In the adapted coordinates introduced above, for any vector field $\bm X \in \mathfrak{X}(\Sigma_t)$, we have $X^0 = 0$ which follows from $n_\mu X^\mu = 0$ and $n_i=0$. However, from $n^\mu X_\mu = 0$ we have $X_0 = N^i X_i = N_i X^i$ (cf., e.g., \cite{2020_BMR_III}). The same applies to other spatial tensors that are orthogonal to $\bm n$.
\end{rem}

\subsection{Kinematics and dynamics of the fluid}

In the tilted situation, we introduce the bilinear form $\metb = b_\albe \, \mathbf{d}x^\alpha \otimes \mathbf{d}x^\beta$ with
\begin{equation}
\begin{aligned}
	b_\mnu &:= g_\mnu + u_\mu u_\nu 
\,;
\\
	b_{\alpha \mu} u^\alpha &= 0 
\, , \;\;\quad 
	\tensor{b}{^\mu_\alpha} \tensor{b}{^\alpha_\nu} = \tensor{b}{^\mu_\nu} 
\, , \;\;\quad 
	b^{\alpha \beta} b_{\alpha \beta} = 3 
\, , \label{eq:proj_f}
\end{aligned}
\end{equation}
where $\tensor{b}{^\mu_\nu}$ allows us to project tensors onto the local rest frames of the fluid (orthogonal to $\bm u$).
Note that the presence of a tilt between $\bm  u$ and $\bm n$ results in two different projectors $\metb$ and $\meth$. Using \eqref{eq:proj_f}, we decompose the $4$-covariant derivative of the $1$-form $\dual{u}$ (cf.~\cite{ehlers,2020_BMR_III}):\footnote{See \Cref{foot:derconvention} for a convention regarding the indices.}
\begin{equation}\label{eq:cov_u}
	\nabla_\nu u_\mu 
	=
	- a_\mu \, u_\nu + \theta_\mnu + \omega_\mnu 
\, ,
\end{equation}
with
\begin{subequations}\label{eq:kin_fluid}
\begin{align}
	a_\mu 
	&:= 
	u^\alpha \nabla_\alpha u_\mu 
\, ,
\\
	\theta_\mnu
	&:=
	\tensor{b}{^\alpha_\mu} \tensor{b}{^\beta_\nu} \nabla_{( \alpha} u_{\beta )}
	=:
	\sigma_\mnu 
	+
	\frac{1}{3} \theta b_\mnu 
\, ,
\\
	\theta 
	&:=
    \tr_{\metg} \bm \theta
    =
	\nabla_\alpha u^\alpha 
\, ,
\\
	\omega_\mnu 
	&:= 
	\tensor{b}{^\alpha_\mu} \tensor{b}{^\beta_\nu} \nabla_{[\beta } u_{\alpha ]} \, ,
\end{align}
\end{subequations}
where $a^\mu$ are the components of the 4-acceleration $\bm a$ of the fluid, $\theta_\mnu$ and $\theta$ the components of the fluid's expansion tensor and expansion rate of $\bm \theta$, $\sigma_{\mu \nu}$ the components of the fluid's shear tensor $\bm \sigma$, and $\omega_{\mu \nu}$ the components of the fluid's vorticity tensor $\bm \omega$. All these quantities are orthogonal to $\bm u$:  
\begin{equation}
\label{orthogonality}
    u^\mu a_\mu = 0 
\,,
\quad
    u^\mu \theta_{\mu \nu}=0 
\,,
\quad
    u^\mu \omega_{\mu \nu} = 0 
\,.
\end{equation}
Note that the rest frames of the fluid are not hyper\-surface-forming if $\bm \omega$ does not vanish due to Frobenius's theorem. Finally, one can define the squared \textit{rate of shear} and \textit{rate of vorticity}, 
\begin{equation}\label{eq: def of sig2 and omeg2}
    \sigma^2 
    := 
    \frac{1}{2} \sigma_{\mu \nu} \, \sigma^{\mu \nu}
\,,
\quad
    \omega^2 
    :=
    \frac{1}{2} \omega_{\mu \nu} \, \omega^{\mu \nu}
\,,
\end{equation}
which are both nonnegative. Moreover, from the anti-symmetrizing \eqref{eq:cov_u} it follows
\begin{equation}\label{eq:vorticity_components}
	\omega_{\mu \nu} 
	= 
    a_{[\mu} u_{\nu]} - \nabla_{[\mu} u_{\nu]}
	=
    a_{[\mu} u_{\nu]} - \p_{[\mu} u_{\nu]}
\,.
\end{equation}

\subsection{The matter model and energy-momentum conservation equations}
\label{sec: energy_momentum_tensor}

The matter model we consider in this work is a single general fluid with the energy-momentum tensor $\bm T = T_\mnu \, \diffb x^\mu \otimes \diffb x^\nu$ that can be written with respect to the fluid rest frames with the induced metric $\metb$ as
\begin{equation}\label{eq:se_fluid_u}
	T_\mnu  
	= 
	\epsilon \, u_\mu u_\nu + 2 \, q_{(\mu} u_{\nu)} + p_\mnu  \, , 
\end{equation}
where
\begin{subequations}
\begin{align}
	\epsilon 
	&:=
	u^\alpha u^\beta T_\albe 
\,,
\\
	q_\mu 
	&:= 
	- \tensor{b}{^\alpha_\mu} u^\beta T_\albe 
\,,
\\
	p_{\mnu} 
	&:= 
    \tensor{b}{^\alpha_\mu} \tensor{b}{^\beta_\nu} T_\albe
    =:
    \pi_\mnu + p \, b_\mnu 
\,,
\end{align}
\end{subequations}
with the properties
\begin{equation}\label{orthogonality2}
    u^\mu q_\mu = 0 \quad , \quad	u^\mu \pi_{\mu \nu} = 0 \quad , \quad b^{\mu \nu} \pi_{\mu \nu} = 0 \ , 
\end{equation} 
and where $\epsilon$ denotes the energy density of the fluid in its rest frames, $q^\mu$ the components of the momentum density $\bm q$ (energy flux relative to $\bm u$), $p_\mnu$ its stress density,
$p := b^{\mnu} p_\mnu / 3$ the isotropic pressure, and $\pi_\mnu$ the components of the  traceless anisotropic stress $\bm \pi$.

Alternatively, the energy-momentum tensor can be decomposed with respect to the normal frames with the induced metric $\meth$ as
\begin{equation} \label{eq:se_fluid_n}
	T_\mnu = E \, n_\mu n_\nu + 2 \, J_{(\mu} n_{\nu)} + S_\mnu \, ,
\end{equation}
with
\begin{subequations}
\begin{align}
    E 
    &:= n^\alpha n^\beta T_\albe 
\,, \label{eq:eulerianenergy}
\\
    J_\mu 
    &:= 
    - h^\alpha_{\phantom{\alpha} \mu} n^\beta T_\albe 
\, ,
\\
    S_\mnu 
    &:= 
    h^\alpha_{\phantom{\alpha} \mu} h^\beta_{\phantom{\beta} \nu} T_\albe
    =:
    \Pi_{\mnu} 
    +
    P h_{\mnu}
\, , 
\end{align}
\end{subequations}
where $E$ is the energy density of the fluid, $J_\mu$ its momentum density, $S_\mnu$ its stress density, $P := h^{\mnu} S_\mnu / 3$ the isotropic pressure, and $\Pi_\mnu$ the components of the spatial and traceless anisotropic stress, all as measured in the normal frames (by Eulerian observers).
Moreover, the following relations hold:
\begin{align}
    \tensor{T}{^i_j}
    &=
    \tensor{S}{^i_j}
    +
    n^i J_j
    =
    \tensor{S}{^i_j}
    -
    N^{-1} N^i J_j
\,,
\\
    T 
    &:=
    g^{\albe} T_\albe
    =
    -E
    +
    3 P
\,.
\end{align}

Using expression \eqref{eq:four_vel}, after some lengthy but straightforward calculations, we can relate Eulerian rest frame and fluid rest frame variables as follows:\footnote{For the application to warp-drive spacetimes, these general extrinsic relations will become relevant for an Eulerian observer looking at the moving spaceship in the tilted situation.}
\begin{subequations}\label{eq: Euler_frestframe_rels}
\begin{align}
    h_\mnu - n_\mu n_\nu
    &=
    b_\mnu  -  u_\mu u_\nu    
\,,
\\
\begin{split}\label{eq:rel_E}
    E 
    &= 
    \gamma^2 ( \epsilon + p ) - p
	+ 
    2 \, \gamma v^\alpha q_\alpha
\\
&\quad
    + 
    \pi_\albe  v^\alpha v^\beta 
\, ,  
\end{split}
\\ 
    E - 3 P 
    &= 
   \epsilon - 3 p
\, , \label{eq:rel_S}
\\
\begin{split}
    J_\mu 
    &=
    \gamma q_\mu
    +
    \gamma^2 \left( \epsilon + p + \gamma^{-1} q^\alpha v_\alpha \right) v_\mu 
\\
&\quad
    +
    v^\alpha \pi_{\mu \alpha}
\\
&\quad
    -
    \left(
        \gamma v^\alpha q_\alpha
        +
        \pi_\albe  v^\alpha v^\beta
    \right)
    n_\mu
\,,
\end{split}
\\
\begin{split}
    S_\mnu
    &=
    p h_\mnu + \pi_\mnu
    +
    \gamma^2 \left( \epsilon + p   \right) v_\mu v_\nu 
\\
&\quad
    +
    \pi_\albe v^\alpha v^\beta  n_\mu n_\nu 
    -
    2 v^\alpha \pi_{\alpha (\mu} n_{\nu )}
\\
&\quad
    +
    2 \gamma \left[ q_{(\mu} v_{\nu)}  - v^\alpha q_\alpha n_{(\mu} v_{\nu)} \right] 
\,. 
\end{split}
\end{align}
\end{subequations}
\begin{rem}
Using the fact that $n_i = 0$, we expressed many terms on the right-hand sides of \eqref{eq: Euler_frestframe_rels} in terms of $n_\mu$ so that the spatial components of the Eulerian quantities will simplify.
\end{rem}
\begin{rem}
    For a perfect fluid, the relations in \eqref{eq: Euler_frestframe_rels} reduce to
\begin{subequations}
\begin{align}
    h_\mnu - n_\mu n_\nu
    &=
    b_\mnu -  u_\mu u_\nu    
\,,
\\
    E
    &= 
    \gamma^2 ( \epsilon + p ) 
    -
    p
\,,
\\ 
	E - 3 P 
    &= 
   \epsilon - 3 p
\, ,
\\
    J_\mu 
    &=
    \gamma^2 \left( \epsilon + p \right) v_\mu 
\,,
\\
    S_\mnu
    &=
    p h_\mnu
    +
    \gamma^2 \left( \epsilon + p   \right) v_\mu v_\nu 
\,.
\end{align}
\end{subequations}
\end{rem}

Evidently, if $\bm u = \bm n$, then all quantities introduced above coincide. In this work, we will write the field equations in terms of the Eulerian quantities, since this is the \emph{creative} approach taken in the literature; otherwise, a \emph{realistic} approach should be adopted where the matter quantities should be used (see \cite[Chapter~IV, \S 6]{1971_Synge_GR} and \Cref{sec: Synge} below).
We finish this part by a lemma which relates $\theta_\mnu$ to $\CK_\mnu$.
\begin{lem}
The expansion tensor and the second fundamental form in the tilted case are related to each other by the following equation:
\begin{equation}
\begin{aligned}
    \theta_{\mu \nu}
    &=
    - \gamma \CK_{\mu \nu}
    +
    \gamma \nabla_{( \mu} v_{ \nu )}
    +
    \gamma^{-1} u_{( \mu} \p_{\nu )} \gamma
\\
&\quad
    +
    a_{( \mu} u_{\nu )}
    -
    \gamma \mathfrak{a}_{( \mu} n_{\nu )}
\,,
\end{aligned}
\end{equation}
with $\mathfrak{a}_\mu$ from \eqref{eq:naccel}.
\end{lem}

\subsubsection{Conservation of energy and momentum in fluid rest frames}

From ${\nabla}_\beta T^\albe = 0$ and using the notation $\dot{\bm F} := \nabla_{\bm u} \bm F$ for any given tensor field $\bm F$, we deduce the energy-momentum conservation laws as measured in fluid rest frames  (cf.~\cite{2020_BMR_III}), by contracting along $u_\alpha$:
\begin{equation}\label{eq: energyconservation}
	\dot\epsilon + \theta \left( \epsilon + p \right) 
		+ a_\alpha q^\alpha + \nabla_\alpha q^\alpha + \pi^\albe \sigma_\albe = 0 \, , 
\end{equation}
and by contracting with $b_{\mu \alpha}$:
\begin{equation}\label{eq: momentumconservation}
\begin{aligned}
&	(\epsilon + p) a_\mu 
    +
    \tensor{b}{^\alpha_\mu} \nabla_\alpha p 
    + 
    b_{\mu \alpha} \dot{q}^\alpha 
\\
&\quad
    + 
    \frac{4}{3} \theta q_\mu + q^\alpha (\sigma_{\mu \alpha} + \omega_{\mu \alpha}) 
    + 
    b_{\mu \alpha} \nabla_\beta \pi^\albe 
    =
    0
\,.
\end{aligned}
\end{equation}

\subsubsection{Conservation of energy and momentum in Eulerian rest frames}

We now derive the energy-momentum conservation laws as measured in Eulerian rest frames. From \eqref{eq:se_fluid_n} and ${\nabla}_\beta T^\albe = 0$, we get (cf.~\cite[Equation~(6.2)]{2012_Gourgoulhon_formalism}):
\begin{equation}
\begin{aligned}
&    \left(
        n^\mu \nabla_\mu E 
        -
        E K
    \right)
    n_\alpha 
    +
    E D_\alpha \ln N
    -
    J^\mu \CK_{\mu \alpha}
\\
&\quad
    +
    n_\alpha \nabla_\mu J^\mu  
    +
    n^\mu \nabla_\mu J_\alpha
    -
    \CK J_\alpha
    +
    \nabla_\mu \tensor{S}{^\mu_\alpha}
    =
    0
\,.
\end{aligned}
\end{equation}
Similar to the previous subsection, contracting the above with $n^\alpha$ and $\tensor{h}{^\alpha_\nu}$, we arrive at the energy-momentum conservation laws as measured in Eulerian rest frames (cf.~\cite[Equations~(6.10) and (6.20)]{2012_Gourgoulhon_formalism}):
\begin{align}
\begin{split}\label{eq:energycons}
    \frac{1}{N} (\p_t - \CL_{\bm N}) E
&    +
    D_i J^i
    +
    2 J^i D_i \ln N
\\
&\quad
    -
    \CK_{ij} S^{ij}
    -
    E \CK
    =
    0
\,,
\end{split}
\\
\begin{split}\label{eq:momentumcons}
    \frac{1}{N} (\p_t - \CL_{\bm N}) J_\alpha
&    +
    D_\mu \tensor{S}{^\mu_\alpha}
    +
    \tensor{S}{^\mu_\alpha} D_\mu \ln N
\\
&\quad
    -
    \CK J_\alpha
    +
    E D_\alpha \ln N
    =
    0
\,,
\end{split}
\end{align}
where we used the fact that $\tensor{h}{^\mu_\alpha} n^\nu \nabla_\nu J_\mu - J^\mu \CK_{\mu \alpha} = N^{-1} (\p_t - \CL_{\bm N}) J_\alpha$ (cf.~\cite[Section 6.2.4]{2012_Gourgoulhon_formalism}).  
Note that 
\begin{equation}\label{eq:LienofPhi}
    \CL_{\bm n} \Phi 
    = 
    N^{-1} (\p_t - \CL_{\bm N}) \Phi
\,;
\quad
    \forall \, \Phi \in C^\infty (\mfd)
\,.
\end{equation} 

\subsubsection{Rest mass, its density and conservation}

Finally, we define another important local quantity that satisfies a conservation law which is motivated by the conservation of (baryonic) particle number (see, e.g., \cite{2012_Gourgoulhon_formalism,Baumgarte_Shapiro_2010_numerical, 2020_BMR_III}). To an (isolated) system that contains matter we can attribute a conserved \emph{rest mass}; indeed, since we presumed the existence of a fundamental $4$-velocity of a fluid, we can define its \textit{mass current density} $\bm m$ as 
\begin{equation}
    \bm m := \varrho \bm u
\,,
\end{equation}
which satisfies the continuity equation
\begin{equation}\label{eq:cons_restmass_density}
    \nabla_\alpha m^\alpha 
    =
	\nabla_\alpha (\varrho u^\alpha) 
    = 0
\,,
\end{equation}
which is equivalent to
\begin{equation}
    u^\alpha \nabla_\alpha \varrho + \theta \varrho 
    = 
    0 
\,.
\end{equation}
The \textit{rest mass density} $\varrho$ of the fluid in its rest frames is therefore defined as such, out of which one can determine the rest mass.
To this end, we consider the rest mass density as measured by the Eulerian observer given by
\begin{equation}
    \mathscr{N}
    := 
    - m^\alpha n_\alpha   
\,,
\end{equation}
which is related to the (baryon) number density. 
The above combined with \eqref{eq:lorentz} yields 
\begin{equation}
   \mathscr{N}
   =
   \gamma \varrho
\,.
\end{equation}
Now, integrating \eqref{eq:cons_restmass_density} over a ``pilled-boxed-shaped'' spacetime region  bounded by a spatial domain $\CD_1$ and $\CD_2$ connected by a timelike mantle surface (see \cite[Section~3.5]{Baumgarte_Shapiro_2010_numerical}), and using the Gau{\ss} theorem, one can show that the rest mass $M_\CD$ over a spatial region $\CD$ containing the matter source, i.e.,
\begin{equation}\label{eq:restmass}
    M_\CD
    :=
    \int_\CD \gamma \varrho \sqrt{\det \meth} \, \diff^3 x
    =
    \int_\CD N u^0 \sqrt{\det \meth} \, \diff^3 x
\end{equation}
is conserved.

Analogously, the mass current density measured by the Eulerian observer can be defined by projecting $\bm m$ onto the spatial hypersurfaces as
\begin{equation}
    M^\alpha
    :=
    \tensor{h}{^\alpha_\beta} m^\beta
    =
    \gamma \varrho v^\alpha
    =
    \mathcal{P} v^\alpha
\,.
\end{equation}
Using the relations above, the continuity equation \eqref{eq:cons_restmass_density} can be written in terms of quantities measured by the Eulerian observers as (cf.~\cite[Section~6.3.2]{2012_Gourgoulhon_formalism})
\begin{equation}
    \left(
        \p_t - \CL_{\bm N}
    \right)
    \mathcal{P}
    -
    N \CK \mathcal{P}
    +
    D_i \left( N \mathcal{P} v^i \right)
    =
    0
\,,
\end{equation}
where we used \eqref{eq:four_vel} and 
\begin{equation}
    \nabla_\alpha X^\alpha 
    =
    D_i X^i
    +
    X^i D_i \ln N
\,;
\,\,
\forall \,\, \bm X \in \mathfrak{X}(\Sigma_t)
\,.
\end{equation}

\subsection{The 3+1 Einstein equations}

In this subsection, we briefly introduce the $3+1$ form of the Einstein equations
\begin{equation}
    \bm G
    +
    \cc \, \metg
    =
    8 \pi G \, \bm T
\,,
\end{equation}
where  
\begin{equation}\label{eq:EinsteinT}
    \bm G
    :=
    \textbf{Ric}
    -
    \frac{1}{2} {\rm R} \, \metg
\end{equation}
is the Einstein tensor, with $\textbf{Ric}$ being the Ricci curvature tensor of the spacetime, ${\rm R} := \tr_{\metg} \textbf{Ric}$, and $\cc \in \R$ is the cosmological constant. The evolution equations are (cf., e.g., \cite{2008_ADM, 2012_Gourgoulhon_formalism, 2008_Alcubierre_3+1, 1978_Smarr_kinematical, 2020_BMR_III})
\begin{align}
	\partial_t  h_{ij} 
    &=
	- 2 
    \left[
        N \CK_{ij} 
        - 
        D_{(i} N_{j)} 
    \right]
\, , \label{eq:evol_h} 
\\
	\partial_t \tensor{\CK}{^i_j} 
    &= 
	N 
    \Big\{ 
        \tensor{\CR}{^i_j}
        + 
        \CK \tensor{\CK}{^i_j} 
		+ 
        4 \pi G 
        \big[ 
            \left( 3 P - E \right) \tensor{\delta}{^i_j} 
            - 
            2 \tensor{S}{^i_j} 
        \big]
\nonumber 
\\ 
&\qquad\quad
		- 
        \cc \tensor{\delta}{^i_j}
    \Big\} 
    - D_j D^i N 
\nonumber 
\\ 
&\quad
    +
    N^k D_k \tensor{\CK}{^i_j} 
    + 
    \tensor{\CK}{^i_k} \, D_j N^k
	- 
    \tensor{\CK}{^k_j} \, D_k N^i
\, , \label{eq:evol_K}
\end{align}
or equivalently using the Lie derivative,
\begin{align}
    (\p_t - \CL_{\bm N})  h_{ij} 
    &=
	- 2 N \CK_{ij} 
\, , \label{eq:evol_h_Lie} 
\\
	(\p_t - \CL_{\bm N}) \tensor{\CK}{^i_j} 
    &= 
	N 
    \Big\{ 
        \tensor{\CR}{^i_j}
        + 
        \CK \tensor{\CK}{^i_j} 
\nonumber 
\\ 
&\quad
		+ 
        4 \pi G 
        \big[ 
            \left( 3 P - E \right) \tensor{\delta}{^i_j} 
            - 
            2 \tensor{S}{^i_j} 
        \big]
    \Big\} 
\nonumber 
\\ 
&\quad
    - 
    D_j D^i N 
    - 
    N \cc \tensor{\delta}{^i_j}
\, , \label{eq:evol_K_Lie}
\end{align}
together with the energy and momentum constraints
\begin{align}
    \CR 
    + 
    \CK^2 
    - 
    \CK_{ij} \, \CK^{ij} 
	&= 
    16 \pi G E 
    + 
    2 \cc 
\,, \label{eq:hamilt_const}
\\
    D_k \tensor{\CK}{^k_{i}} 
    -
    D_i \mathcal{K}
    &=
    8 \pi G J_i 
\, . \label{eq:moment_const} 
\end{align}
Here, $D_i$ is the induced covariant derivative associated with the metric $\meth$, hence we introduced $D^i := h^{ij} D_j$, and $\CR_{ij}$ are the components of the $3$-Ricci tensor $\bm{\mathcal{R}ic}$ associated to it, with $\CR = \tr_{\meth} \bm{\mathcal{R}ic}$.

Another form of \eqref{eq:evol_K} or \eqref{eq:evol_K_Lie} turns out to be useful. We introduce the trace-free part of the second fundamental form, i.e.,
\begin{equation}
    \CA_{ij}
    :=
    \CK_{\langle ij \rangle}
    =
    \CK_{ij}
    -
    \frac{1}{3} \CK \, h_{ij}
\,,
\end{equation}
and write down the evolution equation for the mean curvature $\CK$ and $\CA_{ij}$.
First, taking the trace of \eqref{eq:evol_K_Lie} and using the energy constraint \eqref{eq:hamilt_const}, we find the evolution equation for the mean curvature $\CK$ (cf., e.g., \cite{2008_Alcubierre_3+1, Baumgarte_Shapiro_2010_numerical, Rendall_2008_PDE, 1978_Smarr_kinematical}):
\begin{equation}\label{eq:evolmeancurvatureLie}
\begin{aligned}
    (\partial_t -  \LieN) \CK
    &=
	N 
    \left[ 
        \CK_{ij} \, \CK^{ij} 
        +
        4 \pi G  
        \left( E + 3 P \right) 
        - 
        \cc
    \right] 
\\
&\quad
    - 
    \Delta_{\meth} N 
\,,
\end{aligned}
\end{equation}
And for $\tensor{\CA}{^i_j}$, we find (cf., e.g., \cite[Equation~(4.13)]{2025_VigneronBarzegar_BianchiThurston})
\begin{equation}
\begin{aligned}
	(\partial_t -  \LieN) \tensor{\CA}{^i_j} 
    &= 
    N
    \left(
        \CK \tensor{\CA}{^i_j}
        +
        h^{i k} \CR_{\langle k j \rangle}
        -
        8 \pi G \, \tensor{\Pi}{^i_j}
    \right)
\\
&\quad
    -
    h^{ik} D_{\langle k} D_{j \rangle} N
\,.
\end{aligned}
\end{equation}

\subsection{Principal scalar invariants of the second fundamental form}
\label{sec:PIK}

We add a useful rewriting of the scalar parts of the Einstein equations in the following which is particularly useful for analyzing the current warp-drive spacetimes.

Using the principal scalar invariants for $\bm \CK$, from \eqref{eq:evolmeancurvatureLie} we obtain
\begin{align}
    (\partial_t -  \LieN) \Ik
    &=
    N
    \left\{
        \I^2 [\bm \CK]
        -
        2 \IIk
        +
        4 \pi G  
        \left( E + 3 P \right)
    \right\}
\nn
\\
&\quad
    - 
    \Delta_{\meth} N
    - 
    N \cc
\,,\label{eq: evol I}
\end{align}
which is the Raychaudhuri equation for the case when $\bm n = \bm u$ and $\bm \omega = \bm 0$. Using \eqref{eq:evol_K} we find
\begin{equation}\label{eq: evol II}
\begin{aligned}
    (\partial_t -  \LieN) \IIk
    &=
    8 \pi G N E \, \Ik
\\
&\quad
    -
    N
    \left(
        \tensor{\CR}{^i_j}
        -
        8 \pi G \tensor{S}{^i_j}
    \right)
    \tensor{\CK}{^j_i}
\\
&\quad
    -
    \Ik \, \Delta_{\meth} N
    +
    \tensor{\CK}{_i^j} D_j D^i N
\,,
\end{aligned}
\end{equation}
and
\begin{equation}\label{eq: evol III}
\begin{aligned}
&    (\partial_t -  \LieN)  \IIIk   
\\
    &=
    3 N \, \Ik \, \IIIk
\\
&\quad
    + 
    N \, \IIk
    \left\{
        12 \pi G
        \left( E - P \right)
        +
        \cc
        -
        2 \IIk
    \right\}
\\
&\quad    
    -
    2 N
    \left(
        \tensor{\CR}{^i_j} 
        -
        8 \pi G \tensor{S}{^i_j} 
    \right)
    \tensor{\IIk}{^j_i}
\\
&\quad
    - 
    \IIk \Delta_{\meth} N
    +
    2 \tensor{\IIk}{_i^j} D_j D^i N
\\
    &=
    3 N \, \Ik \, \IIIk
\\
&\quad
    -
    N \IIk
    \left[
        4 \pi G
        \left( E + 3P \right)
        +
        \cc
        -
        \CR
    \right]
\\
&\quad    
    -
    2 N
    \left(
        \tensor{\CR}{^i_j} 
        -
        8 \pi G \tensor{S}{^i_j} 
    \right)
    \tensor{\IIk}{^j_i}
\\
&\quad
    - 
    \IIk \Delta_{\meth} N
    +
    2 \tensor{\IIk}{_i^j} D_j D^i N
\,,
\end{aligned}
\end{equation}
where $\IIk_{ij}$ is defined as in \Cref{sec: notations}, and we used the energy constraint \eqref{eq:hamilt_const} in the new form
\begin{equation}\label{eq: Hamilton constraint II}
    \frac{1}{2} \CR  
	= 
    8 \pi G E 
    + 
    \cc 
    -
    \IIk
\,.
\end{equation}
\begin{rem}
The system of evolution equations \eqref{eq: evol I}--\eqref{eq: evol III}, has been derived before for the case where $N = 1$, $N^i = 0$, $\tensor{S}{^i_j} = 0$, and $\tensor{\Theta}{^i_j} = - \tensor{\CK}{^i_j}$ in \cite[Equations~(5)-(7)]{Brunswic_Buchert_2020_Gauss}. Of course, it is equivalent to the usual (3+1)-decomposition of the Einstein equations. However, there are some advantages in this form compared to the usual form for our study of warp-drive spacetimes (cf.~\eqref{eq:evolEqRwarp} and \Cref{rem:invariants} below). (These equations are valid in three dimensions; one can easily generalize them to the $(n+1)$-dimensional case, as it is done in \cite{Brunswic_Buchert_2020_Gauss} for the special case studied there.)
\end{rem}
%

\section{Classification  of warp-drive spacetimes}
\label{sec:classification}

In this section, we give a clear classification of existing warp-drive proposals and their generalizations. But, before the classification, we should make precise what we mean by a warp bubble. First, we assume that the topology of the spacetime is trivial, i.e., $\mfd \cong  \R^4$, as it is always assumed in the warp-drive literature.\footnote{This constitutes the main difference from the alternative proposal for FTL travels, i.e., traversable wormholes, where FTL travels are achieved via nontrivial topology.} And then, in line with \citet[Definition~1.9]{Natario_2002_warp}, we define warp bubble as follows.
\begin{Def}[Warp bubble]\label{def:WD}
A \emph{warp bubble} is said to be generated by a vector field $\bm W \in \mathfrak{X}\left( \mfd \right)$ if, for a given foliation $\Sigma_t$ with the induced metric $\meth$ and for a constant $R > 0$, the following holds
\begin{enumerate}
    \item $\bm W = \bm 0$ for any point $P_i \in \mfd$ at a small {spatial} distance from a point $P_s(t) \in \Sigma_t$, called the \emph{center of the bubble}, i.e., ${\rm d}_{\bm h}(P_s, P_i) \ll R$. This region represents the interior of the bubble.
    \item $\bm W = - \bm V_s(t)$ for any point $P_e \in \mfd$ far away from the center, i.e., ${\rm d}_{\bm h}(P_s, P_e) \gg R$, where $\bm D \bm V_s = 0$. This region represents the exterior of the bubble.
\end{enumerate}
We say that the warp bubble has a characteristic radius $R$ and a velocity $\bm V_s(t) \in \mathfrak{X}\left( \mfd \right)$.
\end{Def}
We call a spacetime that admits a warp bubble by \Cref{def:WD} a \emph{warp-drive} spacetime.
This definition is very generic. In particular, it does not say, \emph{a priori}, how the vector $\bm W$ is related to the metric or the matter content. This relation will depend on the precise type of warp-drive model which will be discussed in \Cref{sec: restrictions}.
However, despite being very general, \Cref{def:WD} is the core idea of such spacetimes that distinguishes them from other spacetimes, especially the black hole and cosmological ones. We will give a more concrete definition for the warp bubble in \Cref{sec:Rwarp} for the special case of \Rwarp{} models. 
\begin{rem}
    The condition $\bm D \bm V_s = 0$ is the manifestly covariant formulation of spatially constant vector field. In particular, if the spatial metric is flat, this condition ensures that in Cartesian coordinates (where $h_{ij} = \delta_{ij}$, see blow) the components of $\bm V_s$ are spatial constants.
\end{rem}

Now, we are in a position to give a classification of the existing (consistent) warp-drive spacetimes.
For this we give a list of restrictions that comply with the assumptions underlying them. We recognize four main restrictions that cover almost all warp-drive proposals, which are either spatially conformally flat, or spatially flat. The latter is occasionally referred to as \natario{} warp drive that is dubbed ``generic warp drives'' by \citet{2021_Visser_genericwarp}, covering most if not all proposals in the literature. 
We shall not use the word ``generic'' here, and we rather use ``restricted warp drives'' (\emph{\Rwarp{}} for short) introduced first in \cite{2024_Barzegar_Buchert_Letter} through the following list (see \Cref{def:Rwarp} below). 

\subsection{Restrictions imposed on current warp-drive spacetimes}
\label{sec: restrictions}

Based on the $3+1$-formalism introduced in \Cref{sec: framework}, we list the restrictions imposed in the current literature on warp-drive spacetimes. 

\begin{enumerate}[label=\textbf{R\arabic*}]\label{R}

\item \label{R_tilt}
\textbf{Flow-orthogonality:}
Assuming a \emph{flow-orthogonal} foliation of spacetime, i.e., where the $4$-velocity $\bm u$ is aligned with the normal to the hypersurfaces of a spacetime $3+1$ foliation ($\bm u = \bm n$). As a consequence, the covariant vorticity field has to vanish. 

\item \label{R_gauge}
\textbf{Fixing lapse and shift:}
Choosing \textit{a priori} diffeomorphism degrees of freedom in terms of specific choices of lapse and shift functions. The mostly used assumption is a constant lapse function to avoid the time dilatation of two distant Eulerian observers. Then, the (fluid) covariant acceleration $\bm a$ has to vanish as a consequence of \ref{R_tilt}, hence we work in a \textit{geodesic slicing} of spacetime (cf., e.g., \cite[Section~10.2.1]{2012_Gourgoulhon_formalism}). Additionally, some restrictions on the shift can be considered, e.g., by taking $\bm N = - \bm W$, i.e., the shift generates the warp bubble.

\item \label{R_curv}
\textbf{Special curvature of spatial slices:} Assuming either spatially conformally flat or flat slices:

\begin{enumerate}[label=\textbf{R3\alph*}]

\item \label{R_conf}
\textbf{Spatial conformal flatness:} Assuming that the spatial slices are conformally flat.

\item \label{R_flat}
\textbf{Spatial flatness:} Assuming that the spatial Ricci curvature  $\bm \CR$ vanishes.
\end{enumerate}

\item
\label{R_asym}
\textbf{Asymptotic flatness:}
Assuming that the spacetime is asymptotically flat.
\end{enumerate}
In addition to the restrictions listed above, we specify the following restrictions which may appear in warp-drive proposals:
\begin{enumerate}[resume, label=\textbf{R\arabic*}]\label{R2}

\item \label{R_kin}
\textbf{Kinematical restrictions:} 
Restrictions on kinematic variables like the rate of expansion, the shear tensor and the (coordinate) vorticity tensor.

\item 
\label{R_cc}
\textbf{Cosmological constant:}
Assuming that the cosmological constant $\cc$ vanishes.
\end{enumerate}
Imposing different subsets of these restrictions on a warp-drive spacetime will lead to different spacetimes.
Previously, we distinguished two general classes of possible warp-drive spacetimes in \cite{2024_Barzegar_Buchert_Letter}, i.e., the \Rwarp{} and \Twarp{} models. Here, we narrow down the definition of \Rwarp{} models given in \cite{2024_Barzegar_Buchert_Letter}, based on the list above.
\begin{Def}[\Rwarp{} model]\label{def:Rwarp}
    A warp-drive spacetime that satisfies \ref{R_tilt}+\ref{R_gauge}+\ref{R_flat}+\ref{R_asym}(+\ref{R_cc}) is an \Rwarp{} model.
\end{Def}
\begin{rem}
One might wonder why we do not call this class of warp-drive spacetimes the \natario{} warp drive as it is sometimes done in the literature. The reason is the absence of the important condition \eqref{cond: Natario} (see \Cref{sec: Natario}) in \Cref{def:Rwarp}.
\end{rem}
\begin{rem}\label{rem:R_asym}
    We included the asymptotic flatness as a restriction to the list, although it is not mentioned explicitly in the main works by, e.g., \citet{Alcubierre_1994_warp}, \citet{Natario_2002_warp}, and \citet{Broeck_1999_warp}, but it is naturally the case for the Alcubierre warp drive (see \Cref{rem:falloff_alc}), and its generalizations based on the same idea. This, of course, satisfies \ref{R_cc} automatically. This is consistent with the idea of a warp bubble that is well localized. This latter property together with the asymptotic flatness, i.e., definition of a warp bubble and \ref{R_asym}, play a fundamental role in the warp-drive spacetimes and their properties, such as the violation of energy conditions (see, e.g., \Cref{sec: restrictions}). However, a precise definition for \ref{R_asym} is needed, which we will give in \Cref{def: asympt. flatness}, to avoid possible errors.
\end{rem}
\Rwarp{} models include many existing proposals, whereas there is no proposed model of type \Twarp{}, except from the one introduced in \cite{2024_Barzegar_Buchert_Letter}. This class of warp-drive spacetime essentially violates \ref{R_tilt}, hence it is tilted. They might potentially violate other restrictions as well, but the tilt is their main character.
Here, we distinguish yet another class which lies in a sense in between of the two classes mentioned above, by dropping two restrictions from \Rwarp{}, i.e., \ref{R_gauge} and \ref{R_curv}, keeping only \ref{R_tilt}. The reason for this is that the current warp-drive spacetimes are based on the idea that the motion of the warp bubble is associated with the shift vector and not the tilt, hence allowing for keeping \ref{R_tilt} even when dropping \ref{R_gauge} and \ref{R_curv}.
We call this new class of models \emph{\Swarp{}} (see \Cref{sec:Swarp} below). In a sense, one has the hierarchy \Rwarp{} $\subset$ \Swarp{} $\subset$ \Twarp{}. Hence, in order to classify existing proposals in the literature, we commence by considering \Swarp{} models as their broadest possible form  with generic spatial metric, lapse function and shift vector, and we then recover other proposals by applying \ref{R_tilt} to \ref{R_kin} consecutively.
\begin{rem}
The classification done by \citet{2021_Bobrick_physicalwarp} is at best vague, suffers from some errors (see, e.g., \Cref{err:timelikeKilling}), and does not satisfy \Cref{principle}.
\end{rem}

Since most known warp-drive spacetimes do not include the cosmological constant, we impose \ref{R_cc} (i.e., setting $\cc = 0$) for the spacetimes considered in this paper from now on. As remarked earlier, \ref{R_cc} is already included in \Rwarp{} models.

Throughout we use the (spatial) coordinate system given in \eqref{eq:line_elem}, and call coordinate systems \emph{Cartesian} in which $h_{ij} = e^{2\psi} \delta_{ij}$ for both $\psi \neq 0$ and $\psi = 0$ cases. We perform all calculations in such coordinates. Moreover, by abuse of notation, we define $\dnorm{\bm x} := (\delta_{ij} x^i x^j)^{1/2}$.

Before advancing further, we note that since we take an agonistic approach, we analyze the spacetimes classified below purely from mathematical and geometric point of view, and do not, \emph{a priori}, attribute any physical significance to them.

\subsection{Flow-orthogonal spacetimes (\Swarp{} models)}
\label{sec:Swarp}

The most general metric which can be generalized based on the existing proposals for warp-drive spacetimes is the one which merely assumes \ref{R_tilt}, i.e.,
\begin{equation}
\begin{aligned}
    h_{i j}
    &=
    h_{ij}(t, \bm x)
\,,
\\
    N
    &=
    N(t, \bm x)
\,,
\\
    N^i(t, \bm x)
    &=
    -V^i (t, \bm x)
    =
    -W^i (t, \bm x)
\,.
\end{aligned}
\end{equation}

We call such a model, which merely satisfies \ref{R_tilt}, an \Swarp{} to distinguish it from \Rwarp{} introduced in \cite{2024_Barzegar_Buchert_Letter}, i.e., these are models where the warp-bubble generator is the (fluid) coordinate velocity, i.e., $\bm V = \bm W$. Since \Swarp{} models encompass a wide range of possible spacetimes, it does not make sense to write down the Einstein equations for them. Instead, we will write them for more restricted models. Note that we do not, \emph{a priori}, impose \ref{R_asym} in the following, unless for \Rwarp{} models, by definition. 

Additionally, since the existing proposals for warp-drive spacetimes are all special cases of \Swarp{} models with flat spatial slices, the principal scalar invariants of the extrinsic curvature $\bm \CK$ (introduced in \Cref{sec:PIK}) turn out to be very useful, not only because they encode useful information about the second fundamental form due to its special form it takes for these models, but also because they reveal some valuable relations.
\begin{rem}\label{rem:geometricstructure}
To our knowledge, perhaps except from \cite{Visser_2004_fundamental, 2021_Visser_genericwarp} where it is discussed in an indirect way, the relations involving the principal invariants are not discussed in the literature. Although we agree with \cite[Section~III.B]{2021_Visser_genericwarp} that these relations are not useful the way they are advertised in \cite{2021_FellHeisenberg_positive}, they reveal however some important structure for warp-drive spacetimes. Our intuition for usefulness of these quantities comes mainly from works in relativistic inhomogeneous cosmology (see, e.g., \cite{Buchert_instability_1994, 1977_EhlersBuchert_Newtonian, 2020_BMR_III}), although they are introduced and used in continuum mechanics, especially elasticity (see, e.g., \cite{1960_Truesdell_ClassicalFieldTheories, 2012_Marsden_MathematicalElasticity}).
\end{rem}

As mentioned before, the warp-bubble generator for more general warp-drive spacetimes could be different than the coordinate velocity of the matter (see \cite{2024_Barzegar_Buchert_Letter} for a proposition). However, from now on, we will assume that the coordinate velocity always is the warp-bubble generator, i.e., $\bm V = \bm W$.

\subsection{Spatially conformally flat slices with generic shift vector and lapse function}
\label{sec: conformally flat with lapse}

We assume \ref{R_tilt} and \ref{R_conf}. More precisely, we assume
\begin{equation}\label{eq:conf_flat_with_N}
\begin{aligned}
    h_{i j}
    &=
    e^{2 \psi} \delta_{i j}
\,,
\\
    N
    &=
    N(t, \bm x)
\,,
\\
    N^i (t, \bm x)
    &=
    - V^i (t, \bm x)
\,,
\end{aligned}
\end{equation}
where $\psi : \Sigma_t \rightarrow \R$
is a smooth function at each time $t$, and we expressed the metric components in Cartesian coordinates. Therefore, the second fundamental form in this general case can be determined using the relations from \Cref{sec: conform} and \eqref{eq:evol_h} as follows
\begin{equation}
    \CK_{i j}
    =
    - \frac{1}{N} \p_{(i} V_{j)}
    +
    \Psi_{i j}
\,,
\end{equation}
with
\begin{equation}
    \Psi_{i j}
    :=
    - \frac{1}{N} 
    \left[
        N \CL_{\bm n} \psi \, h_{i j}
        -
        2 \p_{(i} \psi V_{j)}
    \right]
\,,
\end{equation}
where we used \eqref{eq:LienofPhi}.
Hence,
\begin{align}
    \Ik
    &=
    - \frac{1}{N} h^{i j} \p_{i} V_{j}
    +
    \I [{\bm \Psi}]
\nn
\\
    &=
    - \frac{e^{-2 \psi}}{N} \delta^{i j} \p_{i} V_{j}
    +
    \I [{\bm \Psi}]
\,,
\\
    \IIk
    &=
    \frac{1}{2 N^2} 
    \left[
        \left( h^{i j} \p_{i} V_{j} \right)^2
        -
        h^{i j} h^{k \ell} \p_{(i} V_{k)} \p_{(j} V_{\ell)}
    \right]
\nn
\\
&\quad
    +
    \II [{\bm \Psi}]
\nn
\\
    &=
    \frac{e^{- 4 \psi}}{2 N^2} 
    \left[
        \left( \delta^{i j} \p_{i} V_{j} \right)^2
        -
        \delta^{i j} \delta^{k \ell} \p_{(i} V_{k)} \p_{(j} V_{\ell)}
    \right]
\nn
\\
&\quad
    +
    \II [{\bm \Psi}]
\,,
\end{align}
where we introduced the principal invariants of the symmetric tensor $\Psi_{i j}$ in a manner similar to those for $\bm \CK$.
From \eqref{eq: Hamilton constraint II} using \eqref{eq: con_Ricci_sc} it follows
\begin{equation}\label{eq:EconfN}
\begin{aligned}
    8 \pi G E
    &=
    \frac{e^{- 4 \psi}}{2 N^2} 
    \left[
        \left( \delta^{i j} \p_{i} V_{j} \right)^2
        -
        \delta^{i j} \delta^{k \ell} \p_{(i} V_{k)} \p_{(j} V_{\ell)}
    \right]
\\
&\quad
    +
    \II [{\bm \Psi}]
    - 
    e^{-2 \psi}
    \left(
        2 \Delta_{{\bm \delta}} \psi
        +
        \dnorm{\bm \p \psi}^2
    \right)
\,.
\end{aligned}
\end{equation}
We see that it is in general possible to have positive energy density by choosing $\psi$ and the shift vector properly. However, this does not imply that such models are free from detriments and do not violate the WEC.

The metric given in this section covers most, if not all warp-drive spacetimes in the literature, although this general metric has not been analyzed so far in this context.
In the next sections we recover all existing warp-drive proposals by imposing more restrictions.

\subsection{Spatially flat metric with generic lapse function and shift vector}
\label{sec: Natario with N}

Assuming \ref{R_tilt} and \ref{R_flat} (i.e., setting $\psi = 0$ hence $\bm \Psi = \bm 0$ in the previous subsection), \eqref{eq:conf_flat_with_N} reduces to
\begin{equation}
\begin{aligned}
    h_{i j}
    &=
    \delta_{i j}
\,,
\\
    N
    &=
    N(t, \bm x)
\,,
\\
    N^i (t, \bm x)
    &=
    - V^i (t, \bm x)
\,.
\end{aligned}
\end{equation}
Since $\bm \Psi = \bm 0$, the second fundamental form and its first and second invariants reduce to 
\begin{subequations}\label{eq: K_I_II generic flat slices}
\begin{align}
    \CK_{i j}
    &=
    - \frac{1}{N} \p_{(i} V_{j)}
\,,
\\
    \Ik
    &=
    - \frac{1}{N} \p_{i} V^i
\,,
\\
    \IIk
    &=
    \frac{1}{2 N^2} 
    \left[
        \left( \p_{i} V^i \right)^2
        -
        \p^{i} V^{j} \p_{(i} V_{j)}
    \right]
\,,
\end{align}
\end{subequations}
where
\begin{equation}\label{eq:partialup}
    \p^i := \delta^{i j} \p_j
\,.
\end{equation}
From \eqref{eq: Hamilton constraint II} we find 
\begin{equation}\label{eq: E for R1+R3}
\begin{aligned}
    E
    &=
    \frac{1}{8 \pi G} \IIk
\\
    &=
    \frac{1}{16 \pi G N^2} 
    \left[
        \left( \p_{i} V^i \right)^2
        -
        \p^{i} V^{j} \p_{(i} V_{j)}
    \right]
\,,
\end{aligned}
\end{equation}
which can be rewritten as
\begin{equation}\label{eq: energy density Natario}
    E
    =
    \frac{1}{16 \pi G N^2} 
    \Big[
        \p_i \left( V^i \p_{j} V^j - V^j \p_{j} V^i \right)
        -
        \Omega^{ij} \Omega_{ij}
    \Big]
\,,
\end{equation}
as it is shown in \cite[Equation~(7.16)]{2021_Visser_genericwarp} for the $N=1$ case and \cite[Equation~(4.4)]{2023_Shoshany_WarpCTC} for a generic lapse function. Here, we introduced the \emph{coordinate vorticity}\footnote{We here follow a traditional convention that is motivated by the comma notation for derivatives, i.e., $\Omega_{ij} := V_{[i,j]}$.\label{foot:derconvention}}
\begin{equation}\label{eq:CoordVor}
    \Omega_{ij} 
    := 
    \p_{[j} V_{i]}
    =
    - \p_{[i} V_{j]}
\,,
\end{equation}
which, as will be clear later, plays a central role in \Rwarp{} models. We also introduce the \emph{coordinate angular velocity vector} as\footnote{Sometimes, instead of the angular velocity vector $\Omega^i$, which is the dual to the vorticity tensor, the vorticity vector $2 \Omega^i$ is used interchangeably.}
\begin{equation}\label{eq:coord_ang_vel}
    \Omega{}^i 
    := 
    - \frac{1}{2} \epsilon^{ijk} \Omega_{jk}
\,,
\end{equation}
which yields
\begin{equation}\label{eq: coor_Omeg_to_vecOmeg}
    \Omega_{ij} = - \epsilon_{ijk} \Omega{}^k
\,.
\end{equation}
We further define the squared magnitude of coordinate vorticity, similar to \eqref{eq: def of sig2 and omeg2}, as
\begin{equation}\label{eq:Omegasquared}
    \Omega^2
    := 
    \frac{1}{2} \Omega_{ij} \Omega^{ij}
    =
    \Omega{}^i \Omega_i
\,.
\end{equation}
%

\subsection{Spatially conformally flat slices with generic shift vector and unit lapse}
\label{sec:confflatunitlapse}

Assuming \ref{R_tilt}+\ref{R_gauge}+\ref{R_conf}, \eqref{eq:conf_flat_with_N} reduces to
\begin{equation}\label{eq: conf flat}
\begin{aligned}
    h_{i j}
    &=
    e^{2 \psi} \delta_{i j}
\,,
\\
    N
    &=
    1
\,,
\\
    N^i (t, \bm x)
    &=
    - V^i (t, \bm x)
\,.
\end{aligned}
\end{equation}
We then have
\begin{equation}
    \CK_{i j}
    =
    - \p_{(i} V_{j)}
    +
    \Psi_{i j}
\,,
\end{equation}
where
\begin{equation}
    \Psi_{i j}
    =
    -  
    \CL_{\bm n} \psi \, h_{i j}
    +
    2 \p_{(i} \psi V_{j)}
\,.
\end{equation}
Hence,
\begin{align}
    \Ik
    &=
    - e^{-2 \psi} \p^i V_i
    +
    \I [{\bm \Psi}]
\,,
\\
    \IIk
    &=
    \frac{e^{- 4 \psi}}{2} 
    \left[
        \left(\p^{i} V_{i} \right)^2
        -
        \delta^{i j} \delta^{k \ell} \p_{(i} V_{k)} \p_{(j} V_{\ell)}
    \right]
\nn
\\
&\quad
    +
    \II [{\bm \Psi}]
\,,
\end{align}
with
\begin{align}
    \I [{\bm \Psi}]
    &=
    - 3 \p_t \psi
    -
    V^i \p_i \psi
\,,
\\
    \II [{\bm \Psi}]
    &=
    3 (\p_t \psi)^2
    +
    2 V^i \p_i \psi \p_t \psi
    -
    V^2 e^{-2\psi} \dnorm{\bm \p \psi}^2
\,,
\end{align}
where we used \eqref{eq:LienofPhi}. Then, from \eqref{eq:EconfN} the Eulerian energy density takes the form
\begin{equation}\label{eq:EforVanDen}
\begin{aligned}
    8 \pi G E
    &=
    \frac{e^{- 4 \psi}}{2} 
    \left[
        \left( \delta^{i j} \p_{i} V_{j} \right)^2
        -
        \delta^{i j} \delta^{k \ell} \p_{(i} V_{k)} \p_{(j} V_{\ell)}
    \right]
\\
&\quad
    +
    3 (\p_t \psi)^2
    +
    2 V^i \p_i \psi \p_t \psi
\\
&\quad
    - 
    e^{-2 \psi}
    \left[
        2 \Delta_{{\bm \delta}} \psi
        +
        ( 1 + V^2 ) \dnorm{\bm \p \psi}^2
    \right]
\,.
\end{aligned}
\end{equation}

\subsubsection{Van Den Broeck metric with 3-dimensional shift vector}
\label{sec: van den Broeck}

If one adds \ref{R_asym} to the list above, and assumes special form for the coordinate velocity \`a la Alcubierre (based on \eqref{eq:AlcShape})  we arrive at the Van Den Broeck metric given in \cite{Broeck_1999_warp} (although he considered only a one-dimensional shift vector).

\subsection{Restricted warp drive (\Rwarp)}
\label{sec:Rwarp}

Now, we assume \ref{R_tilt}+\ref{R_gauge}+\ref{R_flat}. This class of warp drive is achieved by just setting $N = 1$ in \Cref{sec: Natario with N} (equivalently, by inserting $\psi = 0$, hence $\bm \Psi = \bm 0$, in \Cref{sec:confflatunitlapse}), i.e.,
\begin{equation}
\begin{aligned}
    h_{i j}
    &=
    \delta_{i j}
\,,
\\
    N
    &=
    1
\,,
\\
    N^i (t, \bm x)
    &=
    - V^i (t, \bm x)
\,,
\end{aligned}
\end{equation}
or, equivalently,
\begin{equation}\label{eq:Rwarp_metric}
    \bm g
    =
    - \diffb t^2
    +
    \delta_{ij}
    \left(
        \diffb x^i - V^i \diffb t
    \right)
    \left(
        \diffb x^j - V^j \diffb t
    \right)
.
\end{equation}
Assuming further \ref{R_asym}, we arrive at what is called the ``restricted warp drive'' (\Rwarp) in \cite{2024_Barzegar_Buchert_Letter}, since it involves the main restrictions, i.e., \Cref{def:WD}, which we concretize in the following: first, we note that the distance ${\rm d}_{\bm h}(P_s, P_i)$ between the center and points in the interior of the warp bubble is translated to $\dnorm{\bm x} \ll R$ for two Eulerian observers in the given coordinate system, and similarly, the distance  ${\rm d}_{\bm h}(P_s, P_e)$ between the center and points in the exterior is transformed to $\dnorm{\bm x} \gg R$. Second, the Eulerian observer in the interior can ``move'' with arbitrary speed $\dnorm{\bm{V}_s(t)}$ with respect to Eulerian observers in the exterior, hence their speed can be superluminal. Now, we define what we mean by superluminality of \Rwarp{} models.
\begin{Def}[Superluminal \Rwarp{}]\label{def:BubbleGenerator}
An \Rwarp{} is said to be \emph{superluminal} if
\begin{equation}\label{eq:superluminal}
    \dnorm{\bm{V}_s(t)} > 1
\end{equation}
for some time interval.
\end{Def}
Similar to previous parts, we compute the extrinsic curvature and its principal invariants
\begin{subequations}\label{eq: K_I_II natario}
\begin{align}
    \CK_{i j}
    &=
    - \p_{(i} V_{j)}
\,,
\\
    \Ik
    &=
    - \p_{i} V^i
\,,
\\
    \IIk
    &=
    \frac{1}{2} 
    \left[
        \left( \p_{i} V^i \right)^2
        -
        \p^{i} V^{j} \p_{(i} V_{j)}
    \right]
. \label{eq:Rwarp_II}
\end{align}
\end{subequations}

Setting $N=1$ in \eqref{eq: energy density Natario} (or $\psi = 0$ in \eqref{eq:EforVanDen}) with \eqref{eq:Omegasquared}  we find
\begin{equation}\label{eq:Euler_energy_density_Rwarp}
    E
    =
    \frac{1}{16 \pi G} 
    \Big[
        \p_i \left( V^i \p_{j} V^j - V^j \p_{j} V^i \right)
        -
        2 \Omega^2
    \Big]
.
\end{equation}
The energy and momentum constraints \eqref{eq: Hamilton constraint II} read, respectively,
\begin{align}
    \IIk
    &=
    8 \pi G E 
\,,\label{eq:Rwarp_Ham_const}
\\
    \p_k \tensor{\Omega}{^k_i}
    &=
    16 \pi G J_i
\,,\label{eq:Rwarp_mom_const}
\end{align}
where the second equation follows  immediately from \eqref{eq:moment_const}  by a straightforward calculation with the following  consequence
\begin{equation}\label{eq:divfreeJ}
    \p_i J^i = 0
\,.
\end{equation}
Using 
\begin{align*}
    S^{ij} \CK_{ij}
    &=
    \Pi^{ij} \CA_{ij}
    +
    P \, \Ik
\,,
\\
\begin{split}
    S^{ij} \IIk_{ij}
    &=
    \frac{1}{6} \Pi^{ij}
    \left(
        \Ik \CA_{ij}
        -
        3 \tensor{\CA}{_i^k} \CA_{k j}
    \right)
    +
    P \, \IIk
\,,
\end{split}
\end{align*}
and inserting \eqref{eq:Rwarp_Ham_const} into \eqref{eq: evol I}--\eqref{eq: evol III}, we obtain
\begin{subequations}\label{eq:evolEqRwarp}
\begin{align}
    (\partial_t +  \LieV) \Ik
    &=
    \I^2 [\bm \CK]
    -
    12 \pi G  
    \left( E - P \right) 
,\label{eq:evolIRwarp}
\\
\begin{split}\label{eq:evolIIRwarp}
    (\partial_t + \LieV) \IIk
    &=
    2 \, \Ik \, \IIk
\\
&\quad
    -
    8 \pi G 
    \left(
        E - P
    \right)
    \Ik
\\
&\quad
    +
    8 \pi G \, \Pi^{ij} \CA_{ij}
\,,
\end{split}
\\
    (\partial_t + \LieV)  \IIIk
    &=
    3 \, \Ik \, \IIIk
\nn
\\
&\quad
    -
    4 \pi G
    \left( E - P \right) \IIk
\\
&\quad 
    +
    \frac{8}{3} \pi G \, \Pi^{ij}
    \left(
        \Ik \CA_{ij}
        -
        3 \tensor{\CA}{_i^k} \CA_{k j}
    \right)
,
\nn
\end{align}
\end{subequations}
and
\begin{equation}\label{eq:evolShearRwarp}
	(\partial_t + \LieV) \tensor{\CA}{^i_j} 
    = 
    \Ik \tensor{\CA}{^i_j}
    -
    8 \pi G \, \tensor{\Pi}{^i_j}
\,.
\end{equation}
The energy conservation equation \eqref{eq:energycons} is equivalent to \eqref{eq:evolIIRwarp} using \eqref{eq:Rwarp_mom_const}, and the momentum conservation equation \eqref{eq:momentumcons} reduces to
\begin{equation}\label{eq:Rwarp_momcons}
    (\p_t + \LieV) J_i
    +
    \p_j \tensor{\Pi}{^j_i}
    +
    \p_i P
    -
    \Ik J_i
    =
    0
\,.
\end{equation}
Therefore, equations \eqref{eq:Rwarp_Ham_const}, \eqref{eq:Rwarp_mom_const}, \eqref{eq:evolEqRwarp}--\eqref{eq:Rwarp_momcons}, are the Einstein  and the conservation equations for \Rwarp{} models.
\begin{rem}\label{rem:invariants}
    We see a pattern in \eqref{eq:evolEqRwarp} which naturally comes out by a straightforward rewriting: the prefactors of the first terms are equal to the number of the corresponding principal invariant (PI), say $\#$PI, multiplied by $\Ik \rm{PI}[\bm \CK]$, and the matter terms are $(E-P)$ times $(\rm{P\I - \I})[\bm \CK]$ (with $(\I - \I)[\bm \CK] \equiv 1$), multiplied with factors of $4 \pi G$ times $(4-\#\rm{PI})$, e.g., if PI=$\II$, then, $\#\rm{PI} = 2$, and the first term is $2 \Ik \, \IIk$ and $(\rm{PI - \I})[\bm \CK]=\Ik$. This structure is easily seen when one considers shear-free \Rwarp{} models (see also \Cref{sec:shearfree} and \Cref{thm:shearfree}).
\end{rem}
\begin{rem}
Equations (7) and (5) in \cite{2021_Lentz_breakingbarrier} and \cite{2021_FellHeisenberg_positive}, respectively, are the result of
inserting \eqref{eq:Rwarp_II} into \eqref{eq:Rwarp_Ham_const}, written in components, where the connection to $\IIk$ and its role in the formalism is not as clear as in our formulation.
\end{rem}
\begin{rem}\label{rem:Lienoperator}
    The operator $\p_t + \LieV$ is nothing but the \emph{material} (or \emph{convective}, or \emph{Lagrangian}) derivative (usually denoted by $\diff/\diff t$), when applied on scalar fields, that is well-known in continuum mechanics (see, e.g., \cite{1960_Truesdell_ClassicalFieldTheories}) and other fields (e.g., \cite{2020_BMR_III}).
\end{rem}
\begin{rem}
    In this case, the spatial covariant derivative reduces to the partial derivative. However, one could use a flat metric with its associated covariant derivative $D_i$, and get (almost) the same equations since $[D_i, D_j] = 0$ holds.
\end{rem}
\begin{rem}
\label{remvelocities}
Under assumptions of \ref{R_tilt}+\ref{R_gauge}+\ref{R_flat}, we find ${\bm u}=(1,{\bm V})$ and $\dual{u} = (-1,\bm 0)$, hence ${\theta}_{ij} = \Theta_{ij} := \p_{(i} V_{j)} = - {\CK}_{ij}$, i.e., the principal scalar invariants (as introduced in \Cref{sec: notations}) of the extrinsic curvature are identical (up to sign changes) to the principal scalar invariants of the expansion tensor.

\end{rem}
%

\subsubsection{Zero-expansion \Rwarp{} models}
\label{sec:zeroexpnat}

The \natario{} zero-expansion warp drive is a subclass of \Rwarp{} spacetimes.\footnote{To be more precise, one has to assume \ref{R_asym} and add the condition \eqref{cond: Natario}.}
In this case, we have $\CK = \Ik = 0$, hence $\bm \CK = \bm \CA$ and $\IIk = - \hnorm{\bm \CA}^2/2$, which together with \eqref{eq:Rwarp_Ham_const} implies the violation of the WEC (see \Cref{thm: natario WEC}). Moreover, from \eqref{eq:evolIRwarp} we obtain $E = P$.

\subsubsection{Coordinate vorticity-free \Rwarp{} models}
\label{sec:freevortwarp}

In this case, since $\bm \Omega = \bm 0$ and the spatial hypersurfaces are $\R^3$, by the Poincar\'e lemma we have
\begin{equation}\label{eq:gradientshift}
    V_i 
    =
    \p_i \phi
\,,
\end{equation}
for sufficiently differentiable function $\phi$ on $\mfd$ (cf.~\Cref{er:Helmholtz} below). Moreover, from \eqref{eq:Rwarp_mom_const} it follows that $\bm J = \bm 0$, which is observed in \cite[Section~3.1]{2021_FellHeisenberg_positive} (cf., also, \Cref{err:hypshiftzeroJ} below). It turns out that the Alcubierre warp drive in this case reduces to Minkowski spacetime (see \Cref{thm:AlcOmegaZero}).

\subsubsection{Shear-free \Rwarp{} models}
\label{sec:shearfree}

In this case, $\CA_{ij} = 0$ which implies  $\CK_{ij} = \CK \delta_{ij}/3$, and hence $\IIk = \I^2[\bm \CK]/3$ and $\IIIk=\I^3[\bm \CK]/27$. This means that all three equations in \eqref{eq:evolEqRwarp} are equivalent. Moreover, by \eqref{eq:evolShearRwarp} we have $\Pi_{ij} = 0$. Finally, from \eqref{eq:divfreeJ} and the momentum constraint \eqref{eq:moment_const} it follows that $\Ik = \CK$ is a harmonic function (see \Cref{thm:shearfree} for a consequence).

\subsection{\natario{} metric}
\label{sec: Natario}

As mentioned in \Cref{sec:Rwarp}, the \natario{} metric is a subclass of \Rwarp{} models with a further condition.
In fact, \citet{Natario_2002_warp}  slightly generalized the Alcubierre metric (see \Cref{sec: Alcubierre} below). 
The shift vector $\bm N$ is assumed to have the following property: components of $\bm N$ should be bounded smooth functions, i.e.,
\begin{equation}\label{cond: Natario}
    N^i \in C^\infty_b (\mfd)
\,.
\end{equation}
Hence, the only differences between the Alcubierre and the \natario{} proposals are that $\bm N$ has three non-vanishing components in general, and \natario{} does not prescribe a specific form of the shift vector (e.g., a form function $f$ which plays the role of a window function, cf.~\Cref{sec: Alcubierre} below). The condition \eqref{cond: Natario} assures that the spacetime is globally hyperbolic (see \Cref{sec: global hyperbolicity}, in particular \Cref{thm: Choquet}) which forbids the existence of closed timelike curves (CTCs) whose existence and role we do not analyze in this paper, and refer the reader to, e.g., \cite{2023_Shoshany_WarpCTC, Shoshany_2019_lecture, Everett_1996_warp}.

\subsubsection{\natario{} with 1D shift vector}
\label{sec: 1d natario}
For this case, we have
\begin{equation}
\begin{aligned}
    h_{i j}
    &=
    \delta_{i j}
\,,
\\
    N
    &=
    1
\,,
\\
    N^i(t, \bm x)
    &= 
    - V^i (t, \bm x)
    =
    - V(t, \bm x) \delta^i_1
\,,
\end{aligned}
\end{equation}
under which \eqref{eq: K_I_II natario} becomes
\begin{subequations}
\begin{align}
    \CK_{i j}
    &=
    - \p_{(i} V \delta^1_{j)}
\,,
\\
    \Ik
    &=
    - \p_1 V
\,,
\\
    \IIk
    &=
    - \frac{1}{4} 
    \left[
        \left( \p_2 V \right)^2
        +
        \left( \p_3 V \right)^2
    \right]
\,.
\end{align}
\end{subequations}
The only nonvanishing components of the extrinsic curvature are
\begin{equation}
\begin{aligned}
    \CK_{1 1} 
    &= 
    - \p_1 V
    =
    \Ik
    =
    \CK
\,,
\\
    \CK_{1 i} 
    &= 
    - \frac{1}{2} \p_i V  
\, ; \quad i = 2,3
\,.
\end{aligned}    
\end{equation}
The coordinate vorticity becomes $\Omega_{ij} = \delta^1_{[i} \p_{j]} V$ with the only nonvanishing components being
\begin{equation}
    \Omega_{1 i} 
    = 
    \frac{1}{2} \p_i V
    = 
    - \CK_{1 i}
\, ; \quad i = 2,3
\,,
\end{equation}
which results in
\begin{equation}
    \Omega^2 
    = 
    \frac{1}{4} 
    \left[ 
        \left( \p_2 V \right)^2 
        + 
        \left( \p_3 V \right)^2 
    \right]
    =
    - \IIk
\,,
\end{equation}
since the first term in bracket in \eqref{eq:Euler_energy_density_Rwarp} vanishes.
This has the immediate implication coming from the energy constraint equation \eqref{eq: Hamilton constraint II}:
\begin{equation}\label{eq:energy_one_d_nat}
    E
    =
    - \frac{1}{32 \pi G }
    \left[
        \left( \p_2 V \right)^2
        +
        \left( \p_3 V \right)^2
    \right]
    =
    - \frac{\Omega^2}{8 \pi G }
    \leq
    0
\,.
\end{equation}
This shows that the energy density is non-positive, hence the WEC is violated (cf.~\Cref{sec:WEC}).
\begin{rem}
   The relation \eqref{eq:energy_one_d_nat} shows that the coordinate vorticity is the central kinematic quantity for the \natario{} warp drive with $1$-dimensional shift vector, including the Alcubierre metric. Nonetheless, it does not mean that if $\bm \Omega = \bm 0$, then the spacetime is trivial, although this is certainly the case for the specific form function in Alcubierre warp drive (see \Cref{thm:AlcOmegaZero}). The correlation between the coordinate vorticity and the Eulerian energy density, to our knowledge, has been first observed in \cite{2024_Barzegar_Buchert_Letter} (cf.~\Cref{rem:geometricstructure}).
\end{rem}
%

\subsection{Alcubierre metric}
\label{sec: Alcubierre}

Assuming all restrictions, i.e., \ref{R_tilt}+\ref{R_gauge}+\ref{R_flat}+\ref{R_asym}+\ref{R_kin}, 
we retrieve the Alcubierre metric, proposed first in \cite{Alcubierre_1994_warp}, by taking the same metric as in \Cref{sec: 1d natario} with a special form for the coordinate velocity, i.e., 
\begin{equation}\label{eq:AlcADMform}
\begin{aligned}
    h_{i j}
    &=
    \delta_{i j}
\,,
\\
    N
    &=
    1
\,,
\\
    N^i(t, \bm x)
    &= 
    - V^i (t, \bm x)
    =
    - V_s(t) [f(r_s) - 1] \delta^i_1
\,,
\end{aligned}
\end{equation}
with $r_s := (x^2 + y^2 +z^2)^{1/2}$. 

The reader might wonder why the metric given by \eqref{eq:AlcADMform} is not the same as the original metric given by \citet{Alcubierre_1994_warp}. The different representation of the coordinate velocity in \eqref{eq:AlcADMform} is due to the fact that we defined the metric of the \Rwarp{} models in \eqref{eq:Rwarp_metric} in terms of the comoving coordinate with $x \equiv x_{{\rm Alc.}} - x_s(t)$, where $x_s(t)$ is the trajectory of the spaceship and
\begin{equation*}
    V_s(t)
    =
    \ddt{x_s(t)}
\,,
\end{equation*}
is the velocity of the spaceship, located at the center of the warp bubble. The reason for this representation is to cope with \Cref{def:WD} motivated by \citet[Definition~1.9]{Natario_2002_warp}.

The warp bubble in this special case is described by a function $f: [0, \infty) \rightarrow [0,1]$, which is a window function (called ``form function'' in \cite{Alcubierre_1994_warp}), assumed to be spherically symmetric and has the following form
\begin{equation}\label{eq:AlcShape}
    f(r_s)
    =
    \frac{\tanh(\sigma (r_s + R)) - \tanh(\sigma (r_s - R))}{2 \tanh(\sigma R)}
\,,
\end{equation}
where $R >0$ and $\sigma > 0$ are two arbitrary parameters interpreted as the radius of the warp bubble and inversely proportional to the bubble wall thickness, respectively. The function $f(r_s)$ has the general property that it, approximately, is equal to one in a neighborhood of the origin of the warp bubble and it vanishes outside it \cite{Alcubierre_1994_warp}, i.e.,
\begin{equation}
    \lim_{\sigma \rightarrow \infty}
    f(r_s)
    =
\begin{cases}    
    1
\,;
& -R \leq r_s \leq R
\,,
\\
    0
\,,
& \text{otherwise} 
\,.
\end{cases}
\end{equation}
The form function $f(r_s)$ guarantees that the shift vector satisfies \eqref{cond: Natario}, and the spacetime is asymptotically flat (see \Cref{rem:falloff_alc}).
Therefore, both the exterior and the interior of the warp bubble are flat.

Then, \eqref{eq: K_I_II natario} becomes
\begin{subequations}
\begin{align}
    \CK_{i j}
    &=
    - V_s \p_{(i} f \delta^1_{j)}
    =
    - V_s f' \delta^1_{(i} \p_{j)} r_s 
\,,\label{eq:KijAlc}
\\
    \Ik
    &=
    - \frac{x}{r_s} V_s f'
\,,
\\
    \IIk
    &=
    - \frac{y^2 + z^2}{4 r_s^2} V_s^2 \left( f' \right)^2
    =
    - \Omega^2
\,,\label{eq:IIalcOmeg}
\end{align}
\end{subequations}
where $f' \equiv \diff f/\diff r_s$.
Hence, 
\begin{equation}\label{eq: energy den in Alcubierre}
    E
    =
    - \frac{1}{32 \pi G }
    \frac{y^2 + z^2}{r_s^2} V_s^2 \left( f' \right)^2
    \leq
    0
\,,
\end{equation}
which indicates the violation of the WEC, as it is mentioned already by \citet{Alcubierre_1994_warp}.

We finish this section with the following theorem which was announced in \cite{2024_Barzegar_Buchert_Letter} and \Cref{sec:freevortwarp}.
\begin{theorem}\label{thm:AlcOmegaZero}
    Coordinate vorticity-free Alcubierre warp drive is Minkowski.
\end{theorem}
\begin{proof}
    If $\bm \Omega = \bm 0$, then from \eqref{eq:IIalcOmeg} and hence \eqref{eq: energy den in Alcubierre}, it follows that either $V_s = 0$ or $f' = 0$. In either case, we obtain $\bm \CK = \bm 0$ (or even stronger) which finishes the proof.
\end{proof}

\section{Demystification of warp-drive spacetimes}
\label{sec:demyst}

This section examines some profound misconceptions, misunderstandings, and mistakes together with their consequences in the warp-drive literature. To enhance readability and facilitate access, this section makes frequent use of Remarks and Errors.

\subsection{Global hyperbolicity of warp-drive spacetimes}
\label{sec: global hyperbolicity}

In this section we examine the claim that the Alcubierre warp-drive spacetime is \emph{globally hyperbolic}. While \citet{Natario_2002_warp}  defines a warp-drive spacetime to be globally hyperbolic (although, as we shall argue below, it puts a constraint on the coordinate velocity), this property is justified wrongly by \citet{Alcubierre_1994_warp}, and it is often repeated in the existing literature without critical examination. In fact, one needs extra assumptions to assure the global hyperbolicity of the Alcubierre warp drive. Global hyperbolicity has, in any case, a severe consequence, as we shall see below.

We start with the Error that  appeared in the original work by \citet{Alcubierre_1994_warp}.
\begin{error}\label{er:alcgh}
\citet{Alcubierre_1994_warp} claimed that every spacetime whose metric is locally described by its ADM form of the metric is globally hyperbolic as long as the spatial metric, $h_{ij}$ in our notation, is positive definite for all $t \in I \subset \R$. This is not true in general as one can construct counterexamples to this statement, e.g., a Schwarzschild spacetime with negative (ADM) mass (cf., e.g., \cite[Section~11.2]{1984_Wald_GR}) or the anti-de~Sitter spacetime (cf., e.g., \cite[Section~V.4.6]{2008_Choquet_GR}) fail to be globally hyperbolic. \Cref{thm: Choquet} below also shows that the positive-definiteness of $h_{ij}$ alone does not guarantee that the spacetime is globally hyperbolic.  This is often repeated or assumed in almost all subsequent works so that we do not list them here.\footnote{For example, \citet[Section~II]{2021_Visser_genericwarp} state: ``The warp-drive spacetime is by construction globally hyperbolic.''}
\end{error}
Determining the global hyperbolicity of a given spacetime is, in general, a difficult task. Nevertheless, there are some conditions under which one can prove it. For this purpose, we suggest using a theorem by \citet{Choquet-Bruhat_Cotsakis_2002_gh} (see also \cite[Section~11.4]{2008_Choquet_GR}). But, before proceeding, it is helpful to repeat the (informal) definition of the global hyperbolicity. 
\begin{Def}[Global hyperbolicity]
A spacetime $(\mfd, \metg)$ is globally hyperbolic if it admits a Cauchy surface $\Sigma$, i.e., a spacelike surface in $\mfd$ such that each causal curve without end points intersects it once and only once. 
\end{Def}
\begin{rem}
The more technical definitions of the global hyperbolicity require some notions from the causality theory on Lorentzian manifolds. We refer the reader to \cite[Section~4.5.4]{2019_Minguzzi_causal} and \cite[Section~2.8]{2020_Chrusciel_GeometryOfBH} for further details and recent developments on this definition.
\end{rem}
\begin{theorem}[Choquet-Bruhat--Cotsakis, 2002]\label{thm: Choquet}
    Any Lorentzian spacetime $(\mfd, \metg)$ with $\mfd  \cong \Sigma_t \times I$ and $I = (t_0, \infty)$, where $\Sigma_t$ is  an $3$-dimensional Riemannian manifold for each $t \in I$, and $\metg$ is a Lorentzian metric which has the ADM form as in \eqref{eq:line_elem}, is globally hyperbolic if
\begin{enumerate}[label=(\roman*)]
    \item the spacetime is time-oriented by increasing $t$,
    
    \item the lapse function $N$ is bounded below and above by positive numbers $N_m$ and $N_M$, i.e., $0 < N_m \leq N \leq N_M$, 

    \item the metric $h_{ij}$ is a complete Riemannian metric on $\Sigma_t$ uniformly bounded below for all $t \in I$ by a metric $  \mathfrak{g}_{ij} = h_{ij}(t_0)$, i.e., there exists a positive constant $A > 0$ such that for vectors $\bm w \in \mathfrak{X}\left( \Sigma_t \right)$ we have $A \mathfrak{g}_{ij} w^i w^j \leq h_{ij} w^i w^j$,

    \item the $\meth$-norm of the shift vector is uniformly bounded by some positive number, i.e., $\hnorm{ \bm N} = \left( h_{ij} N^i N^j \right)^{1/2} \leq \bar{N}$ for some $\bar{N} > 0$.
\end{enumerate}
\end{theorem}
\begin{rem}\label{rem:GHconditions}
Condition $(i)$ either prevents possible emergence of  horizons or pathologies in the given foliation, and hence puts a bound on the lapse function and the shift vector through $\hnorm{\bm N} < N$ (which, \emph{per se}, does not prohibit superluminality in the context of warp-drive spacetimes, see also \cite[Section~5.2.2]{2012_Gourgoulhon_formalism}) that in combination with condition $(ii)$ results in $\hnorm{\bm N} < N_M$, thereby implying condition $(iv)$, or it ensures that, in the presence of horizons, we restrict ourselves to the exterior region, or more precisely,  the \emph{domain of outer communication} (DOC): $I^+(\mathscr{I}^-) \cap I^-(\mathscr{I}^+)$ (see, e.g., \cite[Section~3.1]{2020_Chrusciel_GeometryOfBH} for a precise definition).
\end{rem}
\begin{rem}\label{rem:GHconditionsRwarp}
Conditions $(ii)$ and $(iii)$ are evidently satisfied for the \Rwarp{} models. 
One needs to check conditions $(i)$ and $(iv)$. Condition $(i)$ is satisfied globally only for the subluminal case where no horizon emerges (cf., e.g., \cite{Natario_2002_warp,Hiscock_1997_quantum, Liberati_2009_semiclassical}), or in the DOC for the superluminal case where two horizons exist or form (cf.~the Penrose diagrams in \cite[Figures~2 and 3]{Liberati_2009_semiclassical}). Conversely, if condition $(i)$ holds globally, then from \Cref{rem:GHconditions} it follows that  $\dnorm{\bm V} < 1$, i.e., the \Rwarp{} model is subluminal.
Condition $(iv)$, however, is not mentioned by \citet{Alcubierre_1994_warp}, but \citet{Natario_2002_warp} added this condition later explicitly, i.e., condition~\eqref{cond: Natario}.
\end{rem}
As it is somewhat clear from the above, one can show that it is impossible to ``construct'' a globally hyperbolic superluminal \Rwarp{} model. 
We adopt the definitions of an \emph{eternal} and a \emph{dynamic} warp drive introduced by \citet{Liberati_2009_semiclassical} to make precise what we mean by constructing a warp drive.
\begin{Def}[Eternal and dynamic warp-drive spacetimes]
An \emph{eternal} warp drive is the one with a time-independent $\bm V_s$ (hence $\bm V$), whereas a \emph{dynamic} warp drive has a time-dependent $\bm V_s(t)$ in such a way that it starts from zero and reaches some final value by introducing a ``switching factor,'' say, $S_w(t)$.
The latter is equivalent to saying that one constructs a warp drive.
\end{Def}
\begin{theorem}[Global hyperbolicity of \Rwarp{} models]
\label{thm:globhyp}
It is impossible to construct a superluminal \Rwarp{} model that is globally hyperbolic.
\end{theorem}
\begin{rem}
Although \citet{Liberati_2009_semiclassical} defined $S_w(t)$ to be $C^0$, another (perhaps more desired) possibility is to opt for a step function as a switching factor that obviously provides more problems. However, as we shall show in the proof below, even for a $C^0$ switching factor, superluminal \Rwarp{} models cannot be constructed if one requires them to be globally hyperbolic at the same time.
\end{rem}

\begin{proof}[Proof of \Cref{thm:globhyp}]
If the given \Rwarp{} model is subluminal, then, by \Cref{rem:GHconditionsRwarp}, it satisfies all conditions of \Cref{thm: Choquet}, hence it is globally hyperbolic. Conversely, if it is globally hyperbolic in the given foliation, then it cannot contain horizons which is only the case if it is subluminal.
If it is superluminal, as mentioned in \Cref{rem:GHconditionsRwarp}, it allows for two horizons. Now, if it is an eternal \Rwarp{} model, then no Cauchy horizon arises (cf.~\cite{Liberati_2009_semiclassical}), but only the exterior region is globally hyperbolic. In any case, an eternal \Rwarp{} model, as the name suggests, are not constructed, hence we are left with the possibility of construction of a superluminal dynamic \Rwarp{} model. In this case, Cauchy horizons appear (cf.~\cite{Liberati_2009_semiclassical}; see also \cite{1998_Low_limits}), so that not even the exterior (the DOC) is globally hyperbolic. This is independent of the regularity of the switching factor $S_w(t)$. We hence conclude that a globally hyperbolic \Rwarp{} model is necessarily subluminal.
\end{proof}
\begin{rem}\label{rem:GHdeter}
An implication of \Cref{thm:globhyp} is that one may omit the global hyperbolicity when assuming a reasonable hypothesis for the construction of warp drive as intended.
This, from a philosophical point of view, makes sense; indeed, if one could manipulate the spacetime however one wants, then the theory could not be fully deterministic (see, e.g., \cite{2021_Smeenk_DeterminismGR} for a discussion on determinism in GR).
\end{rem}
\begin{rem}
Our result is, in a sense, complementary to, but distinct from that by \citet{1998_Low_limits}, who in an interesting work  utilizes the initial-value problem to address the challenge of modeling \emph{agency} within a deterministic framework by defining a ``decision'' as a localized modification of initial data on a subset of a future Cauchy surface.
He then proves that the spacetime region outside the \emph{domain of dependence} of this decision remains isometric to the original, i.e., even if an agent has the ``freedom'' to choose (or change) the initial data at a point $p \in \mfd$, that choice (or change) is strictly bounded by the domain of dependence of that slice, hence, the ``decision'' cannot influence any part of the Universe faster than the speed of light. Consequently, Low concludes in \cite[Theorem 4.2]{1998_Low_limits} that it is impossible to reach a destination any sooner than light would allow if the matter involved is ``physically reasonable''---a matter field defined as satisfying a \emph{symmetric hyperbolic system} and the dominant energy condition (DEC). Therefore, even if an observer can ``choose'' initial conditions, the deterministic evolution of nonexotic matter ensures that a warp drive cannot bypass the universal speed limit (associated to the ray cone of the symmetric hyperbolic system, or the light cone; see \Cref{rem:Geroch} for further discussion). Although our result for \Rwarp{} may seem to be more general, since we did not assume any energy conditions, \citet{1998_Low_limits} proved his result for any notion of warp drive which does not have any restriction except from energy conditions.
\end{rem}
\begin{rem}
    \Cref{thm:globhyp} implies that either $\bm V$ (hence $\bm V_s(t)$) should be already (i.e., eternally) superluminal (tachyonic) before ``turning on'' the warp drive, or it can never be superluminal if one assumes the global hyperbolicity. This was first pointed out by \citet{1995_Krasnikov_Hyperfast} using causality theory, and also somehow highlighted, but only from a qualitative perspective, by \citet{1998_Coule_NoWarp}. In addition, following the work by \citet{1995_Krasnikov_Hyperfast}, \citet{Everett_1997_subway} show that a spaceship crew cannot create or control a warp bubble following a ``simple argument'': since the bubble's front edge is spacelike separated from the ship's center, hence any signal or action taken by the captain is restricted to their future light cone, because a forward-emitted photon eventually reaches a point where its speed equals the bubble's velocity, meaning it remains at rest relative to the bubble and never reaches the front edge. Consequently, the crew has no way to create or control it, meaning the bubble would have to be constructed in advance by an external observer whose own light cone already covers the entire intended path. See the following Error.
\end{rem}
\begin{error}\label{er:imposs_GH}
The facts mentioned in the previous Remark contradict the assumption of possibility of ``building'' a warp drive starting from an ``approximately Minkowski spacetime'' reflected, for example, in the works by \citet{2021_Lentz_hyperfast, 2021_Lentz_breakingbarrier}, and \citet{2024_Helmerich_Bobrick_WarpFactory}, if one insists on the global hyperbolicity of the given spacetime.
\end{error}
\begin{rem}
\Cref{er:imposs_GH} is an example of violation of \Cref{principle}; a concept is often used by supporters of physical warp drives without exploring why, and what exactly the consequences are.
\end{rem}
%

\subsection{Asymptotic flatness}
\label{sec:asymptflatADM}

In the literature, warp-drive spacetimes are assumed to fulfill \ref{R_asym}, i.e., to be asymptotically flat. However, this is occasionally assumed without precise investigation which led to some misunderstandings and wrong results.
Thus, we provide a simplified and specific definition of asymptotically flat spacetimes, sufficient for our purposes, and refer the reader to, e.g., \cite{York_1978_kinematics, Straumann_2013_GR, 2012_Gourgoulhon_formalism, Barzegar_2017_energy, 2013_Chrusciel_energy} for further details and subtleties. 

\begin{Def}[Asymptotically flat slices]\label{def: asympt. flatness}
Let $(\mfd, \metg)$ be a Lorentzian manifold and $\Sigma$ a complete oriented 3-dimensional hypersurface. Then, $\Sigma$ is said to be asymptotically flat iff the following holds
\begin{enumerate}
    \item there is a compact set $\CC \subset \Sigma$ such that $\Sigma \setminus \CC$ is diffeomorphic to $\R^3 \setminus B(\mathscr{R})$ where $B(\mathscr{R})$ is a coordinate ball of radius $\mathscr{R} > 0$,
    \item in local standard coordinates $(x^i)$ on $\Sigma \setminus \CC$ obtained from $\R^3 \setminus B(\mathscr{R})$, the components of the metric on $\Sigma \setminus \CC$ satisfy the following falloff conditions
    \begin{equation*}
        h_{ij}
        =
        \delta_{ij}
        +
        \CO(r^{-1})
        \,,
        \quad
        \p_k h_{ij}
        =
        \CO(r^{-2})
        \,,
    \end{equation*}
    \item in the same coordinates, the components of the second fundamental form (extrinsic curvature) of $\Sigma$ on $\Sigma \setminus \CC$ satisfy the following falloff conditions
    \begin{equation*}
        \CK_{ij}
        =
        \CO(r^{-2})
        \,,
        \quad
        \p_k \CK_{ij}
        =
        \CO(r^{-3})
        \,,
    \end{equation*}
\end{enumerate}
when $r \rightarrow + \infty$, where $r := \dnorm{\bm x}$.
\end{Def}
\begin{rem}\label{rem:asymflatRwarpSigma}
    Since the asymptotic flatness for \Rwarp{} models are always discussed with regard to the $\{t = \const.\}$-foliations in \eqref{eq:Rwarp_metric} which are also taken to be Cauchy surfaces, we can reformulate \Cref{def: asympt. flatness} by considering asymptotically flat \emph{spacetimes}, i.e., spacetimes with $\mfd = \Sigma_t \times I$ and $I = (t_0, \infty)$ where $(\Sigma_t, \meth)$ are complete 3-dimensional Riemannian manifolds for all $t \in I$, and replace $\Sigma$ by $\Sigma_t$. This definition, however, contains already global hyperbolicity, whereas \Cref{def: asympt. flatness} does not presume it. Note that, in a spacetime sense, we further have $\mathcal{L}_n K_{ij} = \CO(r^{-3})$ in the same coordinate system.
\end{rem}
Based on \Cref{rem:asymflatRwarpSigma}, hereafter whenever we refer to an \Rwarp{} model, we mean that the $\{t = \const.\}$-slices are asymptotically flat in the sense prescribed in \Cref{def: asympt. flatness}.
\begin{rem}\label{rem:falloff_alc}
The \alcub{} spacetime (and any other warp-drive spacetime with \eqref{eq:AlcShape}) is asymptotically flat; indeed, for the \alcub{} shape function \eqref{eq:AlcShape} and its derivative we have $f, f' \in \CO(e^{-2r})$, which in turn implies $\CK_{ij}, \p_k \CK_{ij}  \in \CO(e^{-2r})$ which certainly decay faster than required.
\end{rem}
The ignorance of the precise definition of asymptotic flatness  \Cref{def: asympt. flatness} has some serious consequences which we will mention in conjunction with the errors regarding the notion of the ADM energy (see Errors~\ref{err:Schustermain} and \ref{er:VisserShoshanyTotalMass}).
Asymptotic flatness puts some strong restrictions on a given \Rwarp{} model. For example, as we will show in \Cref{thm: adm mass of natario}, it dictates \Rwarp{} models to have vanishing ADM energy. Furthermore, we prove in the following that for the coordinate vorticity-free \Rwarp{} models and the shear-free \Rwarp{} models, introduced in \Cref{sec:freevortwarp} and \Cref{sec:shearfree}, respectively, the asymptotic flatness forces them to be Minkowski.
\begin{theorem}\label{thm:shearfree}
    Shear-free \Rwarp{} models are Minkowski.
\end{theorem}
\begin{proof}
    If $\CA_{ij} = 0$, we have $\CK_{ij} = \CK \delta_{ij}/3$, and from the momentum constraint equation \eqref{eq:moment_const} we find $J_i = - 12 \pi G \p_i \CK$, which in turn in combination with \eqref{eq:divfreeJ} yields  $\Delta_{\bm \delta} \CK = 0$. Now, since the \Rwarp{} models are asymptotically flat by definition, then $\CK$ decays as $\CO(r^{-2})$ on $\R^3$ (see \Cref{def: asympt. flatness} below), which means that it is bounded on $\R^3$ as $\CK$ is continuous. Now, by Liouville's Theorem (cf., e.g., \cite[Section~2.2.3]{2010_Evans_PDE}) we conclude that $\CK$ is (spatially) constant (hence identically zero because it decays to zero, cf.~also \Cref{thm:harmonicwarp} below). Therefore, $\bm \CK = \bm 0$, and all its principal invariants vanish, resulting in $E=0$, and by \eqref{eq:evolEqRwarp} in $P=0$. Finally, from \eqref{eq:evolShearRwarp} it follows $\Pi_{ij} = 0$, which finishes the proof.
\end{proof}
\begin{theorem}\label{thm:harmonicwarp}
Let $(\mfd, \metg)$ be an \Rwarp{} model. If the shift vector is the gradient of a harmonic function on the spatial slices, then $(\mfd, \metg)$ is Minkowski.
\end{theorem}
\begin{proof}
We start with \eqref{eq:gradientshift}, i.e., $V_i = \p_i \phi$, hence $\CK_{ij} = - \p_i \p_j \phi$. Since $\phi$ is harmonic on the spatial slices, the mean curvature vanishes, i.e., $\CK = \Delta_{\bm \delta} \phi = 0$. Since the slices are asymptotically flat $\phi$ decays fast enough, hence it is bounded on $\R^3$. By Liouville's Theorem we conclude that $\phi$ is zero (cf.~\Cref{thm:shearfree}), resulting in $V_i = 0$, hence finishing the proof. 
\end{proof}
\begin{error}\label{er:LentzElliptic}
\citet[Section~3]{2021_Lentz_breakingbarrier} argues that an \Rwarp{} model with an ``elliptic'' shift vector is considered in the literature, i.e., a warp-drive model with a shift vector of the form \eqref{eq:gradientshift} that satisfies $\Delta_{\bm \delta} \phi = 0$. But, \Cref{thm:harmonicwarp} shows that such a spacetime is identically Minkowski, and not a zero-expansion \natario{} warp drive.
\end{error}
We finish this section by demonstrating a misunderstanding which appeared several times in the warp-drive literature.
\begin{error}\label{er:truncation}
\citet{2021_Bobrick_physicalwarp} talk about the ``truncation'' of the gravitational field in the \alcub{} or \natario{} spacetimes,  which, as criticized correctly by \citet[Appendix~A]{2021_Visser_genericwarp}, is wrong since the shape function \eqref{eq:AlcShape} does not have a compact support, but it falls off very fast (cf.~\Cref{rem:falloff_alc}). \citet[Footnote~9]{2024_Helmerich_Bobrick_WarpFactory} point out this mistake, without acknowledging the critique by \citet{2021_Visser_genericwarp}, or referring to the mistake done by \citet{2021_Bobrick_physicalwarp}. The same misunderstanding entered the work by \citet{2021_FellHeisenberg_positive}. See also \Cref{er:truncationWEC} for a related mistake.
\end{error}
%

\subsection{Energy and mass in \Rwarp{} spacetimes}
\label{sec:notionofMassEnergy}

In this section, we shall clarify the misconception around the notion of (total) energy and ADM energy, and their role in \Rwarp{} models.

In most of the works on warp-drive spacetimes which are \Rwarp{} models, one attributes a total energy (or mass) to  these spacetimes, using a naive, Newtonian-like definition for the energy (in the given foliation)
\begin{equation}\label{eq:Etot}
    \Etot
    :=
    \int_{\Sigma_t} T_{\mnu} n^\mu n^\nu \diff^3 x
    =
    \int_{\Sigma_t} E \diff^3 x
\,,
\end{equation}
with $E$ being the Eulerian energy density defined in \eqref{eq:eulerianenergy}.
Here, we argue that this notion of energy (mass) is ambiguous, or at best is nonunique, for \Rwarp{} models and therefore \emph{cannot} fulfill the role it is meant to play in the literature, i.e., it does \emph{not} necessarily represent the ``total mass'' of the warp bubble. One of these ambiguities is related to the misunderstanding about the ADM energy. In the warp-drive literature there are works which speculate about the ADM energy for \Rwarp{} models, but, as it turns out easily by just applying the definitions properly, it is zero for such spacetimes, whereas $\Etot$  is not, in general. However, $\Etot$ in \eqref{eq:Etot} is occasionally taken to be the same as the ADM energy (cf.~\Cref{rem:ADMmassNewton}). Therefore, to present the results appropriately, we will start with the notion of the ADM quantities (hence the ADM energy).

We shall begin with a basic but important fact: the ADM energy is defined only for asymptotically flat spacetimes. Unfortunately, this seemingly trivial and basic fact has often been misunderstood in the literature on warp-drive spacetimes. 

For this purpose, we define the ADM energy, or more generally, the ADM $4$-momentum, thus, in particular, distinguishing between the ADM mass and the ADM energy (cf.~\cite{2008_ADM_republication}).\footnote{Sometimes in the literature the ADM energy and ADM mass are used interchangeably. We will stick to the precise definition repeated here.}
\begin{Def}\label{def:ADMquantities}
Let $(\mfd, \metg)$ be asymptotically flat in the sense it is defined in \Cref{def: asympt. flatness}. Then, the total energy $E_{\mathrm{ADM}}$ and momentum $\bm P_{\mathrm{ADM}}$ contained in $\Sigma_t$ is encoded in the ADM $4$-momentum 
\begin{equation}
    \bm p_{\mathrm{ADM}}
    =
    \left( E_{\mathrm{ADM}}, \bm P_{\mathrm{ADM}} \right)
\,,
\end{equation}
with
\begin{equation}
    E_{\mathrm{ADM}}
    :=
    \frac{1}{16 \pi G} \lim_{r \rightarrow \infty} \int_{{\rm S}_r} 
    \left(
        \p^i h_{ij}
        -
        \p_j h
    \right)
    \frac{x^j}{r} \diff {\rm S}
\,, \label{eq: ADM enegy}
\end{equation}
and
\begin{equation}
    \left( \bm P_{\mathrm{ADM}} \right)^i
    :=
    \frac{1}{8 \pi G} \lim_{r \rightarrow \infty} \int_{{\rm S}_r}
    \left(
        \tensor{\CK}{^i_j}
        -
        \CK \tensor{\delta}{^i_j}
    \right)
    \frac{x^j}{r} \diff {\rm S}
\,, \label{eq: ADM momentum}
\end{equation}
where $\p^i$ is defined as in \eqref{eq:partialup}, $h := \delta^{ij} h_{ij}$, and ${\rm S}_r$ denotes a sphere of radius $r$ on $\R^3 \setminus B(\mathscr{R})$ and $\diff {\rm S}$ is the volume element on ${\rm S}_r$ induced by the flat metric, all for the coordinates $(x^i)$ defined in \Cref{def: asympt. flatness}. Moreover, the ADM mass is
\begin{equation}
    m_{\text{ADM}}^2
    :=
    \Eadm^2
    -
    |\bm \Padm|^2_{\bm \delta}
\,,
\end{equation}
\end{Def}

We continue with a trivial but often overlooked result.
\begin{theorem}\label{thm: adm mass of natario}
Consider an \Rwarp{} model spacetime. Then,
\begin{enumerate}
    \item the ADM energy vanishes, i.e., $\Eadm = 0$,
    \item the ADM momentum does not vanish in general, i.e., $\bm \Padm \neq \bm 0$, hence in combination of the previous fact we have $\bm p_{\mathrm{ADM}} = \left( 0, \bm \Padm \right)$,
    \item if the \Rwarp{} model spacetime reduces to the \alcub{} one, then both the ADM energy and ADM momentum vanish, i.e., $\bm p_{\mathrm{ADM}}^{\mathrm{Alc.}} =  \bm 0$.
\end{enumerate}
\end{theorem}
\begin{proof}
The result for the ADM energy follows immediately by setting $h_{ij} = \delta_{ij}$ in \eqref{eq: ADM enegy}. For the ADM momentum the result depends on the falloff conditions of the second fundamental form. It is, however, zero for the special case of the Alcubierre metric due to the rapid decay of the second fundamental form (see \Cref{rem:falloff_alc}).
\end{proof}

\Cref{thm: adm mass of natario} has an important consequence which is related to the famous positive energy (mass) theorem \cite{Witten_1981_positive, Schoen_Yau_1981_positive}, and also a consequence of asymptotic flatness~(\ref{R_asym}).
\begin{theorem}[Positive energy theorem]\label{thm: positive mass theorem}
    Consider an \Rwarp{} model spacetime. Then, the positive energy theorem, which relies on the DEC being satisfied, dictates that the spacetime must be Minkowski.
\end{theorem}
\begin{rem}
If one drops the DEC assumption, some serious physical ambiguities arise, as shown in \Cref{rem:imaginaryADM} (see also \Cref{cor:asymflatDEC}).
\end{rem}
We do not wish to list all misuses of the notion of the ADM mass in the literature. In the following we list some of these occurrences in terms of Errors.

\begin{error}
The left-hand side of Equation~(11.26) in a work by \citet{Alcubierre_2017_warp_basics} is clearly misleading, if not wrong. Again, the authors did not pay attention to the criteria in \Cref{def: asympt. flatness} (cf.~\Cref{thm: adm mass of natario}). This appears to originate from the work of \citet{Visser_2004_fundamental} (see also \cite{Visser_2004_linearized}), where the authors introduce the concept of ``ADM mass due to stress-energy'' (see \cite[Equation~(67)]{Visser_2004_fundamental}). This notion is not well-defined, and serves as the basis for their speculation about the ``compensation of the ADM mass of the spaceship by the ADM mass due to the stress-energy'' and the violation of the energy conditions (\cite[Section~2.1]{Visser_2004_fundamental}). Moreover, \citet[Equation~(67)]{Visser_2004_fundamental} take the ADM energy to be equal to $\Etot$ given in \eqref{eq:Etot} (see \Cref{err:ADMtotMass}). See also \Cref{er:VisserShoshanyTotalMass} for a related mistake.
\end{error}
\begin{error}\label{err:Schustermain}
The work by \citet{Visser_2023_adm} reflects several misconceptions regarding the notion of the mass and energy in GR, especially the ADM energy;\footnote{Although the subtleties of this issue were expounded in \cite[Section~(2.3)]{Visser_2023_adm}.} indeed, \citet{Visser_2023_adm} mistakenly attribute a nonvanishing ADM energy (referred to as the ADM mass by the authors) to the \natario{} warp drive by considering the ``zero-vorticity warp bubble,'' since ``this is the form most appropriate to make the appearance of notions of mass in a warp drive explicit, as it can be easily compared to Schwarzschild in Painlev\'e--Gullstrand form,'' and so they construct what they call ``Schwarzschild-based warp drive.'' However, this strategy relies on several misunderstandings that lead to wrong consequences which we list below:
\begin{enumerate}
    \item By \Cref{thm: adm mass of natario}, the ADM energy of generic asymptotically flat \natario{} warp drives (i.e., \Rwarp{} models) vanishes, and this \emph{cannot} be avoided, and it has nothing to do with the vorticity as the authors claim. This misunderstanding occurred already by the same authors in \cite[Section~VII.C]{2021_Visser_genericwarp}, where they speculated that only Alcubierre and zero-expansion \natario{} warp drives have vanishing ADM energy. Then, the parameter $M$ in \cite[Equation~(2.11)]{Visser_2023_adm} is at best ambiguous. Unless, the authors mean that the $\{t = \const.\}$-slices there are not asymptotically flat, and only after a transformation to an asymptotically flat foliation the ADM mass of the spacetime can be computed. In any case, this shows a lack of precision (see $(iv)$ below).
    
    \item The foliation associated to the Painlev\'e--Gullstrand  coordinates is not asymptotically flat (since the third condition in \Cref{def: asympt. flatness} is not satisfied),
    making the concept of ADM energy irrelevant in this context. This is a classic pitfall against which one is warned (cf., e.g., \cite[Example 8.2]{2012_Gourgoulhon_formalism} and \cite[Section~3.5]{Baumgarte_Shapiro_2010_numerical}), and it is reflected 
    in the part where authors write ``[...] ﬁniteness of the ADM mass places mild constraints on the fall-off of the metric components.''

    \item Comparing the Schwarzschild metric in the Painlev\'e--Gullstrand coordinates with that of the \natario{} warp drive is wrong if one assumes that the \natario{} warp drive is asymptotically flat.

    \item Although the ADM mass of the Schwarzschild metric in the foliation associated to the Painlev\'e--Gullstrand coordinates has no meaning, other notions of mass or energy (such as the Komar or the Misner--Sharp mass, see, e.g., \cite{Barzegar_2017_energy, 1978_Beig_ADMEnergyG00, 1964_Misner_SphericalCollapse}) can be used to infer the ADM energy of the Schwarzschild spacetime in this foliation. However, in general, this \emph{cannot} be done for an asymptotically flat \natario{} warp drive with respect to foliations that are not asymptotically flat, because the metric is, in particular, not stationary nor spherically symmetric.
    
    \item \citet[Section~5]{Visser_2023_adm} consider the Misner--Sharp mass for the Schwarzschild-based warp drive and conclude $m(r) = m [1 + \CO(r^{-n})]$ for some $n>0$, which is again wrong, because, firstly, the Misner--Sharp mass is defined only for the spherically symmetric spacetimes (cf., e.g., \cite{1996_Hayward_GravitationalEnergySpherical}, see also \cite{2010_Abreu_KodamaTimeGeometrically} by one of the authors), whereas the Schwarzschild-based warp drive is not spherically symmetric,\footnote{Although the same authors point out a similar misunderstanding by others in \cite[Appendix~A]{2021_Visser_genericwarp}, they make the same mistake here.} secondly, even if one naively computes it, it would not be of the form claimed by the authors (it also would not be invariant as it depends on coordinates), and $m$ is not ``simply'' the ADM mass.
    
    \item Finally, the use of what is supposed to be the ADM energy to show the violation of the NEC is not valid, as it is not defined here (cf.~\Cref{rem:SchusterNEC}).
\end{enumerate}
Consequently, most of the speculations involving the ADM energy in the context of \Rwarp{} models should be re-evaluated. This includes not only the work by \citet{Visser_2023_adm}, but also the ones by \citet{2024_Clough_GW}, where the authors related the ``true ADM mass'' to ``causal disconnectedness'' of the Alcubierre warp drive,  and \citet{2024_Fuchs_constant} which suffers further from a vague formulation so that one cannot decide if the speculated metric is asymptotically flat or not (see also \Cref{err:ADMtotMass}).
\end{error}
\begin{error}\label{er:inadequateADMVisser}
\citet[Section~VII.C]{2021_Visser_genericwarp} give an inadequate falloff for ``zero-vorticity warp drives'' (cf.~\Cref{sec:freevortwarp}) while assuming that the given spatial slice is asymptotically flat, hence the ADM energy would be zero. Indeed, the falloff rate of the second fundamental form is $\CO(r^{-3/2})$ which is not enough and does not satisfy \Cref{def: asympt. flatness} (the same mistake is done by \citet[Section~4.1]{2023_Shoshany_WarpCTC}). Therefore, the ADM energy in this case in not defined at all, unless they mean the ADM mass that is well-defined in another foliation. Even in that case the given falloff rate is insufficient.
\end{error}
\begin{rem}\label{rem:imaginaryADM}
Although, by \Cref{thm: adm mass of natario}, the ADM energy vanishes for \Rwarp{} models, the ADM momentum can, in general, be nonzero which leads to $m_{\text{ADM}}^2 = - |\bm \Padm|^2_{\bm \delta}$, resulting in an imaginary ADM mass. This, in the quantum field theory, corresponds to off-shell particles (virtual particles), in particular the tachyonic ones (cf., e.g., \cite{1967_Feinberg_FTL}, \cite[Section~8]{2002_Fayngold_FTL}, \cite[Remarks~2.3 and 4.5]{Gourgoulhon_2013_SRT}, and references therein). However, instead of particles, in this case we are dealing with an isolated system, whose interpretation is not clear to us. It is, however, to some extent reminiscent of the work by \citet{1998_Coule_NoWarp}.
It is even worse if we consider the Alcubierre warp drive, for which we have $\bm p_{\mathrm{ADM}}^{\mathrm{Alc.}} =  \bm 0$, i.e., a spacetime which has no nontrivial global quantities at infinity, still carrying some content.
Therefore, it is not clear at all what the meaning of this configuration is, even if we allow for violation of some energy conditions.
\end{rem}
\begin{rem}
In a sense, \Rwarp{} spacetimes are dual to the Schwarzschild spacetime; indeed, for the Schwarzschild spacetime we have $\bm p_{\mathrm{ADM}} = \left( m, \bm 0 \right)$ where $m$ is the mass parameter, whereas for the \Rwarp{} spacetime we have $\bm p_{\mathrm{ADM}} = \left( 0, \bm \Padm \right)$ when $\bm \Padm \neq \bm 0$.
\end{rem}
We note that there is no direct relation between the violation of the WEC and the vanishing of the ADM mass, in general.
\begin{error}\label{er:ADMviolEC}
    \citet{2024_Fuchs_constant} speculate that the violation of the energy conditions for the Alcubierre metric (or \Rwarp{} models in general) is due to the vanishing of the ADM energy. However, this is not a consequence as there are direct counter examples (see, e.g., \cite{1990_Frauendiener_ShellBlackHole, 2025_Saito_RemovingNakedSingularities}).
\end{error}

In the following important Remark, we elaborate on the relation between the ADM energy, the Eulerian energy density, and quantities in the Newtonian limit.

\begin{rem}\label{rem:ADMmassNewton}
\newcommand{\Newt}{{\scriptscriptstyle{\textbf{N}}}}
The classical Newtonian limit can be derived from a weak-field metric in the Newtonian gauge, i.e., $\bm g = - (1 + 2 \Phi) \diffb t^2 + (1 - 2 \Psi) \bm \gamma$ (assuming no vector and tensor modes). If one assumes that the pressure and the anisotropic stress from the matter content do not source the gravitational field at the Newtonian level, then the two scalar fields $\Psi$ and $\Phi$ coincide and correspond to the Newtonian gravitational potential which is the solution to the Poisson equation $\Delta_{\bm \delta} \Psi = 4 \pi G \rho_\Newt$, where $\rho_\Newt = \lim_{c\rightarrow\infty} E$ is the Newtonian limit of the Eulerian energy density. A direct consequence is that the ADM energy is related to the Newtonian potential as follows (cf., e.g., \cite[Section~11.2]{1984_Wald_GR}, \cite[Section~8.3.3]{2012_Gourgoulhon_formalism}, and \cite[Section~1.1]{2013_Chrusciel_energy}):
\begin{equation}\label{eq:NewtADM}
\begin{aligned}
    E_{\mathrm{ADM}}
    &=
    \frac{1}{4 \pi G} \lim_{r \rightarrow \infty} \int_{{\rm S}_r} 
    \p_i \Psi \,
    \frac{x^i}{r} \diff {\rm S}
\\
    &=
    \frac{1}{4 \pi G}\int_{\R^3} \Delta_{\bm \delta} \Psi \, \diff^3 x\\
    &=
    \int_{\R^3} \rho_\Newt \, \diff^3 x
\,.
\end{aligned}
\end{equation}
This relation between the ADM energy, the Newtonian potential, and the  Newtonian limit of the Eulerian density is however much more subtle for warp-drive spacetimes for several reasons. First, because the \Rwarp{} metric \eqref{eq:Rwarp_metric} is not written in the Newtonian gauge, one cannot simply read off $\Phi$ and $\Psi$ from that metric. Second, the fact that $\int_{\mathbb{R}^3} E {\rm d}^3 x \not=0$ but $E_{\mathrm{ADM}} =0$ in general for \Rwarp{} shows that the Newtonian limit of the former is not equal to the latter. This highly suggests that, either ``$\Phi = \Psi$'' does not hold anymore in the limit,\footnote{Having $\Phi \neq \Psi$ in the Newtonian limit would imply that the pressure and the anisotropic stress of the matter content have a nonzero contribution to the Newtonian gravitational potential. For instance, this can be the case for radiative fluids.} or that the Newtonian potential is not sourced by the Newtonian limit of $E$. These two points make the identification and interpretation of the Newtonian gravitational potential, and more generally the Newtonian limit, for warp drives difficult, and a proper investigation of the Newtonian limit of such spacetimes is therefore required. To this end, the use of the covariant ``geometric'' limit involving Galilean structures seems the most appropriate formalism (see, e.g., \cite{1976_Kunzle, 2020_Hansen_et_al}). In an interesting work, using a similar formalism known as the  \emph{Ehlers frame theory} \cite{1981_Ehlers_Newtonsch, *2019_Ehlers_RepNewtLimit} (see also \cite{2019_Buchert_EditorialNoteNewtonian}), \citet{Natario_2004_newtonian_warp} provides a Newtonian potential to what he calls \emph{Newtonian spacetimes} whose  metric includes that of coordinate vorticity-free \Rwarp{} models. Nevertheless, the proposed potential there needs careful interpretation, and cannot easily be related to $\Psi$ from above.

\end{rem}
Now, we can assess better why the meaning of $\Etot$ is not clear, and why it is not a good measure of energy for such spacetimes:
\begin{enumerate}
\item  Computing $\Etot$ from \eqref{eq:Etot} and using \eqref{eq:Euler_energy_density_Rwarp}, we find (cf.~\cite[Equation~(7.17)]{2021_Visser_genericwarp})
\begin{equation}
    \Etot
    =
    - \frac{1}{8 \pi G}  
    \int_{\R^3} \Omega^2 \, \diff^3 x
    \leq
    0
\,,
\end{equation}
which is different from zero for $\bm \Omega \neq \bm 0$, whereas the ADM energy is always zero because of the asymptotic flatness of the $\{t = \const.\}$-slices in \eqref{eq:Rwarp_metric}. We stress again that if one means that the ADM energy is computed with respect to another foliation, while assuming that the ``standard foliation'' is not asymptotically flat (but the spacetime is, which is necessary to talk about the ADM energy), this should be clearly said, thereby satisfying \Cref{principle}.

\item Contrary to the ADM energy which is time-independent for asymptotically flat spacetimes (cf., e.g., \cite[Section~8.3.5]{2012_Gourgoulhon_formalism}), hence making it a reasonable notion of global energy, $\Etot$ is not, in general; one has to use the conservation equation \eqref{eq:energycons} for $E$ which is equivalent to \eqref{eq:evolIIRwarp}.

\item $\Etot$ depends on the foliation, whereas the ADM energy does not.

\item Finally, as we saw in \Cref{rem:ADMmassNewton}, it is not clear if or how $\Etot$ is related to a Newtonian potential.
\end{enumerate}

The core idea of this section is ignored in countless works in the literature.
\begin{error}\label{err:ADMtotMass}
In many works, $\Etot$ in \eqref{eq:Etot} is taken to be the well-defined definition of the energy for such spacetimes, and in some cases equivalent to the ADM mass, e.g., by \citet[Equation~(67)]{Visser_2004_fundamental}, \citet[Section~~VII.C]{2021_Visser_genericwarp} in particular Equation~(7.14) therein, \citet[Section~2.3]{Visser_2023_adm}, and the first version by \citet{2024_Clough_GW} (cf.~\Cref{rem:ADMmassNewton}). To avoid the exhaustive list, we only underscore this misunderstanding in the recent works (and the references therein) by: \citet{Alcubierre_2017_warp_basics},
\citet{2021_Lentz_breakingbarrier} \citet{2021_Bobrick_physicalwarp}, \citet{2021_FellHeisenberg_positive}, \citet{2021_Visser_genericwarp},  \citet{Visser_2023_adm}, \citet{2024_Clough_GW}. 
\end{error}

One possible source of this confusion in the literature would be that the local definition of the rest mass $M_\CD$ introduced in \eqref{eq:restmass} is taken to be $\Etot$ in \eqref{eq:Etot}. But, these are two different concepts and may not be related to each other directly; only in special cases the former merely constitutes a part of the latter (see, e.g., \cite[Exercise~3.23]{Baumgarte_Shapiro_2010_numerical}).

Nonetheless, if $\Etot$ in \eqref{eq:Etot} is not a meaningful one for the total mass (energy) of the warp bubble in R-Warp models, then one can ask what would constitute an appropriate (quasi-local) definition in this context when $\Eadm = 0$?
\citet{2024_Clough_GW} give some insight into the meaning of the naive definition \eqref{eq:Etot}, which clearly aligns with what is mentioned above as they report that the ``final mass of the spacetime volume [i.e., the naive definition of the total mass] is more positive (having started at zero).'' However, a rigorous investigation of a proper definition of the energy of \Rwarp{} models is beyond the scope of the present paper.

To summarize, we conclude with two important statements: $(i)$ the naive definition of total mass (energy) \eqref{eq:Etot} has no obvious meaning for \Rwarp{} models, $(ii)$ $\Etot \neq \Eadm = 0$ for \Rwarp{} models, in general.

\subsection{Symmetries and Killing vector fields}

Sometimes there are some misconceptions about the symmetry of the warp-drive spacetimes, e.g., that they might be spherically symmetric (which was pointed out by \citet{2021_Visser_genericwarp}). In this section, we gather other misunderstandings appearing in the warp-drive literature.

\begin{error}\label{er:sph_symm}
\citet{2021_Bobrick_physicalwarp} claim that both the Alcubierre and \natario{} metrics are spherically symmetric. This claim was correctly criticized by \citet[Appendix~A]{2021_Visser_genericwarp} to which  we refer the reader for further details. Then, \citet{2021_Bobrick_physicalwarp} state that superluminal warp drives, including the \alcub{} or \natario{} ones, cannot be spherically symmetric. This is again a misunderstanding since those metrics are not spherically symmetric anyway. This wrong statement is repeated by \citet[Section~4.4]{2024_Helmerich_Bobrick_WarpFactory}, despite being rightly refuted in  \cite[Appendix~A]{2021_Visser_genericwarp}.
\end{error}
\begin{error}\label{err:timelikeKilling}
\citet{2021_Bobrick_physicalwarp} use ``a global Killing vector field'' to classify their warp proposals. Moreover, their primary focus is put on the existence of a global timelike Killing vector field. However, this consideration is pron to some problems:
\begin{enumerate}
    \item It is by no means clear what this Killing vector is, how it manifests itself in such spacetimes in general, and why they are needed at all.
    \item As criticized correctly by \citet[Appendix~A]{2021_Visser_genericwarp} (see also the relevant critique in \cite[Appendix~C]{2021_Visser_genericwarp}), if a proposed warp model admits a global timelike Killing vector, then the corresponding warp bubble is at rest with respect to a static observer at infinity; in other words there is no motion, let alone a superluminal one, i.e., the ultimate goal of a warp drive (this misconception is repeated by \citet{Abellan_2023_spherical}). The same mistake is repeated by \citet{2021_FellHeisenberg_positive}, where the authors assume that the shift vector potential is time-independent (see \cite[Appendix~C]{2021_Visser_genericwarp} for more critiques).
    \item They claim that the Alcubierre and even the \natario{} metrics do admit a global timelike Killing vector, which is clearly wrong.
    \item As an example for a possible spacelike Killing vector field, one can take the shift vector to be a Killing vector for the Euclidean metric, which results in a  flat spacetime (cf.~\cite[Corollary~1.6]{Natario_2002_warp}).
    \item They claim this Killing vector allows for a $3+1$ decomposition without saying how exactly. All these (and more) demonstrate a misunderstanding and misuse of the notion of Killing vector fields (cf.~\Cref{principle}).
\end{enumerate}
\end{error}
\begin{error}\label{er:coordtrafo}
As pointed out correctly by \citet[Appendix~B]{2021_Visser_genericwarp}, \citet[Appendices~A.1 \& A.2]{2021_Bobrick_physicalwarp} were confused by a wrong perception of coordinate transformation, which was repeated uncritically by \citet[Section~2]{2021_FellHeisenberg_positive}.
Despite the right critique by \citet{2021_Visser_genericwarp}, \citet[Sections 4.1 \& 4.2]{2024_Helmerich_Bobrick_WarpFactory} made another wrong statement that the \alcub{} and Van Den Broeck metrics can be transformed to time-independent metrics, i.e., admitting global timelike Killing vector field, by transforming to what they call ``constant velocity comoving frame'' via a ``Galilean transformation,'' hence dropping some time-derivatives in the calculation (cf.~\Cref{sec:numerics} below). \citet{2024_Fuchs_constant}  make the same mistake around which they build their central concept. 
\end{error}
%

\subsection{Kinematics}

An in-depth treatment of kinematics of \Rwarp{} models lies beyond the scope of this work.
In this section, we will summarize some mistakes regarding the kinematics in an Error.
\begin{error}\label{er:ErrorsFuchsConstant}
    \citet[Section~1.2]{2024_Fuchs_constant} consider the $3+1$-form of geodesic equations. Apart from being formally unclear as to what, e.g., the operator $\diff / \diff t$ means exactly (cf.~\Cref{principle} and \Cref{err:syntax}), they make several mistakes upon which they build their proposal. First, they consider the Alcubierre metric (cf.~\Cref{sec: Alcubierre}), and claim that the spatial derivative of the shift vector vanishes ``inside the passenger volume region'' (because the metric should be flat, cf.~\Cref{err:syntax}), hence concluding from the geodesic equation that $\diff u_i / \diff t = 0$, where $u_i$ is the covariant component of the unit normal vector field (flow-orthogonal motion). However, both claims are wrong; the derivative of the shift vector does \emph{not} vanish inside the bubble, and although the latter equation holds, the reason is not the one the authors gave, but the fact that $u_i = 0$ always holds in this case (see \Cref{remvelocities}). 
    Second, the authors take $u^i$ and $u_i$ to be a vector and its covector, respectively, on the spatial slices, which is wrong, and might imply $u^i = V^i = \diff x^i / \diff t \equiv 0$ due to the previous mistake. This furthermore shows a confusion regarding the $3$-velocity (coordinate velocity); the coordinate velocity of the physical Eulerian observers for the Alcubierre warp drive is by definition $\bm V = - \bm N$ (or $\diff x^i / \diff t = - \beta^i$ in their notation) and in this case given by \eqref{eq:Vtoui} and not $u_i/u^0$ as written by \citet{2024_Fuchs_constant}, since the latter is identically zero, and has nothing to do with the spatial metric as they claim in and around \cite[Equation~(8)]{2024_Fuchs_constant}.\footnote{This consideration would make sense in the context of a \Twarp{} \cite{2024_Barzegar_Buchert_Letter}, which clearly is not what the authors aimed at.} Hence, the $3+1$ form of the geodesic equation does not say anything more than the very property of the flow-orthogonal motion. Hence, the whole conceptualization of a ``constant velocity warp drive'' by \citet{2024_Fuchs_constant} has no meaning as long as no distinction is made between $u^i$ and $u_i$. 
    In general, it seems that the authors have in mind a Newtonian background where the shift vector describes some motion on it (see also \Cref{err:syntax}), e.g., by writing: ``[...] a warp drive that can enable the transport of different observers can do so using a shift vector inside the passenger volume,'' which lacks clear meaning, especially since the shift vector is a choice of coordinate system.  
    Even if it was clear what \citet{2024_Fuchs_constant} mean by this, their model suffers from other serious problems, see, e.g. \Cref{err:TOV} and \Cref{app:TOV}.
\end{error}
%

\subsection{The Synge G-method and unphysical solutions}
\label{sec: Synge}

In general, there are several ways one can approach GR. In particular, for solving the Einstein field equations there are two different ways: from ``right to left,'' or ``backwards'' (see \cite{2024_Ellis_Garfinkle_Gmethod,2024_Barzegar_Buchert_Letter}).
However, \citet[Chapter~IV, \S 6]{1971_Synge_GR} identifies three distinct ``spirits'' of investigation regarding the Einstein field equations: the \emph{realistic}, who seeks to align the equations with their ``already extensive knowledge of the physical universe'' (from ``right to left''); the \emph{agonistic}, who focuses on mastering the inherent mathematical difficulties of the field equations; and the \emph{creative}, whose ``pleasure'' lies in the ``construction of universes, \emph{fantastic} or \emph{realistic}, satisfying the field equations'' (from ``left to right''). This creative approach relies on what is known as the G-method (cf.~\cite{2024_Ellis_Garfinkle_Gmethod,2024_Barzegar_Buchert_Letter}), a process where the metric tensor is given (hence, the left side of the Einstein equation) to mathematically determine the energy-momentum tensor (the right side). While this ``backwards'' methodology allows for the generation of exact solutions through differentiation rather than solving complex partial differential equations, its unrestricted use often leads to ``fantastic'' results that may not correspond to the physical spacetime, potentially violating not only the energy conditions but also some stability conditions (see \cite{2024_Ellis_Garfinkle_Gmethod}).

The creative Synge G-method is the basis of almost all works in the warp-drive literature, with a few exceptions (e.g., \cite{1998_Low_limits, 2024_Clough_GW}) that represent a correct  initial value formulation (IVF) (despite  Errors \ref{err:ADMtotMass} and \ref{er:cloughFallacy}). The ambiguity in taking a clear approach is often misunderstood, overlooked, or ignored, which led to some Errors that we mention in the following.
\begin{rem}
    The approach of the present work is to some extent agonistic; indeed, prior to attach meaning to a solution, even prior to solving Einstein's equations (and often merely trying to respect \Cref{principle}), we managed to identify the misconceptions and misunderstandings embracing the current proposals for warp drive. Although, all three ``spirits'' recognized by \citet[Chapter~IV, \S 6]{1971_Synge_GR} are essential to understand the Einstein field equations better, we clearly advocate more the realistic and agonistic approaches, rather than the creative one.
\end{rem}
\begin{rem}
    To some extent, there is a fundamental difference between two  FTL proposals, i.e., the wormhole and the warp drive, in that the former can exist eternally, whereas the latter is meant to be built. Therefore, an IVF for the latter is essential while the former does not necessarily need it.
\end{rem}
\begin{error}\label{err:contEq}
\citet{2021_Bobrick_physicalwarp} claim that the continuity equations are violated for the Alcubierre warp drive. This statement is manifestly wrong, and is based on a naive Newtonian viewpoint, as pointed out correctly by \citet[Appendix~A]{2021_Visser_genericwarp}, since by the Synge-G method, any given (regular enough) Lorentzian metric will, by the field equations, produce  an energy-momentum tensor that satisfies the continuity equation. This error likely influenced \citet[Section~5]{2021_Lentz_hyperfast}, as it was absent in the original paper \citet{2021_Lentz_breakingbarrier}. Moreover, \citet{2021_Bobrick_physicalwarp} attributes this fact to asymptotically flat spacetimes which is again wrong. This mistake was corrected by \citet[Section~2.2]{2024_Helmerich_Bobrick_WarpFactory} without referring to \cite{2021_Visser_genericwarp} which already had pointed at this mistake. Furthermore, \citet{2024_Helmerich_Bobrick_WarpFactory} reframe this fact as a requirement which is not correct since any diffeomorphism invariant theory enforces the energy-momentum tensor to be divergence free (cf., e.g., \cite[Appendix~E]{1984_Wald_GR}). 
\end{error}
\begin{error}\label{er:cont_eq_beyon_GR}
The misunderstanding and misconception explained in \Cref{err:contEq}, led to another publication by \citet{2022_Carneiro_energyconservation}, which is entirely based on the wrong statement with the hope to solve this pseudo-problem (cf.~\Cref{principle}) by going beyond GR.
\end{error}
\begin{error}\label{er:hyperbolicshitLentz}
    \citet{2021_Lentz_breakingbarrier} proposes a ``hyperbolic shift vector'' which is sourced by a scalar field that satisfies some ``spatial'' hyperbolic equation with a source, i.e., there is no time involved, and the ``propagation'' occurs not in time, but in one spatial direction (see also \Cref{err:hypshiftzeroJ}). Then, the author argues that such a construction can be done utilizing plasma matter. It is by no means clear why plasma should be a candidate as the author does not give support as to what phenomena lead to an equation of motion similar to the hyperbolic one given in \cite{2021_Lentz_breakingbarrier}. To our knowledge, a somewhat similar equation of motion shows up in the context of the supersonic flow (cf., e.g., \cite{1949_Tempest_plasma}) which clearly was not the motivation for \citet{2021_Lentz_breakingbarrier}. In any case, as pointed out in \cite[Appendix~C]{2021_Visser_genericwarp}, the proposal (``hyperbolic-shift-vector-plasma construction'') by \citet{2021_Lentz_breakingbarrier} is not a solution to Einstein's equations, hence violating even the Synge G-method (see also \Cref{sec:bubble_mech}).\footnote{There are even more ambiguities in \cite{2021_Lentz_breakingbarrier}. Looking at \cite[Equations~(26)--(29)]{2021_Lentz_breakingbarrier}, it is not clear if the $4$-velocity $u^\mu$ therein is aligned with the unit normal vector (see \ref{R_tilt}), suggesting that the model might be what we refer to as a \Twarp{} model (see \cite{2024_Barzegar_Buchert_Letter}), which is certainly not what the author intended to.}
\end{error}
\begin{error}\label{err:hypshiftzeroJ}
\citet[Section~3]{2021_Lentz_breakingbarrier} claims that $\bm J = \bm 0$ because of the particular choice of the shift vector (the ``hyperbolic shift vector''). This is, however, related to the form of the shift vector in \eqref{eq:gradientshift}, not the property of the potential function considered there, and it follows directly by the Synge G-method.
\end{error}

Obviously, independent of the approach, statements regarding the given spacetime should be about solutions to Einstein's equation. In \Cref{er:hyperbolicshitLentz} we pointed at one instance where this was not the case. In the following we identify another appearance of a similar error.
\begin{error}\label{err:TOV}
    The construction proposed by \citet{2024_Fuchs_constant} is not a solution to Einstein's equations; indeed, as we showed in \Cref{app:TOV}, the authors do not solve the Tolman--Oppenheimer--Volkoff (TOV) equations resulting from the Einstein equation in the given context. Hence, the algorithm in \citet[Figure~3]{2024_Fuchs_constant} for ``metric creation'' is flawed. This is yet another example of violation of \Cref{principle}.
\end{error}
%

\subsection{Numerics and simulations}
\label{sec:numerics}

There are some numerical works in the literature regarding warp-drive spacetimes. As before, we put an emphasis on the recent works which claim to propose a physical warp drive.

Based on the discussion in \Cref{sec: Synge}, the only clear and sound work in the warp-drive literature, to our knowledge, is the work by \citet{2024_Clough_GW}, in which the authors correctly show that if one starts with the Alcubierre metric as initial data (and with a stiff fluid), one ends up with instability. This can be seen even with a \emph{quasi-Newtonian} approach, since on the spatial slices everything reduce to vector analysis where the collapse is readily seen. And all this happens already for the subluminal case to which they restricted themselves. Nonetheless, this work exhibits some shortcomings, including Errors \ref{err:ADMtotMass}, \ref{er:cloughFallacy}.
\begin{error}\label{er:numericalmagic}
\citet{2024_Helmerich_Bobrick_WarpFactory} (among others) claim (which likely is motivated by \citet[Section~5]{2021_Lentz_breakingbarrier}, \cite[Section~5]{2021_Lentz_hyperfast}, \citet[Section~4]{2021_FellHeisenberg_positive}, and perhaps \citet{2024_Abellan_SphericalWarpbasedBubble}): ``Warp drive research has been mostly derived by analytic studies of various metrics. Such limitations have led to a narrow perspective when it comes to warp drive physicality, especially due to the non-linear nature of the gravitational field equations.'' This implies that the fundamental mistakes, misconceptions, misunderstandings (at some of which \citet{2021_Visser_genericwarp} and the present work pointed), and even more the lack of a clear proposal for a warp drive (cf.~\Cref{sec:bubble_mech} below), can be avoided if one exploits numerical analysis. This view can only lead to another pseudo-problem (cf.~\Cref{principle}), since if a model has formal (and fundamental) problems, one cannot circumvent those by numerics or simulations. They further claim: ``The key challenge in warp research is two-fold: The first is evaluating the complicated Einstein field equations to find stress-energy tensors and the second is in evaluating those stress-energy tensors for physicality,'' which, beside being not the primary issue of the ``warp research,'' shows the unfamiliarity of the authors with the numerical relativity (see \cite{2024_Clough_GW} for a relevant example, cf., also, \Cref{sec: Synge}). Moreover, \citet[Section~6]{2024_Helmerich_Bobrick_WarpFactory}, claim: ``Traditionally, the analysis of physicality in these metrics has been limited to considering the energy density observed by Eulerian observers. However, this approach alone is insufficient for a thorough evaluation of the energy conditions. To overcome these limitations and enable the exploration of more complex solutions, Warp Factory employs a numerical framework,'' which is again not true since it was mostly the ``tradition'' of the authors, as for example this point (that was, among others, claimed by \citet{2021_Bobrick_physicalwarp}) was correctly criticized  by \citet{2021_Visser_genericwarp}   without using numerics.
\end{error}
\begin{error}\label{er:IVF}
    Some works, including the ones by \citet{2024_Abellan_SphericalWarpbasedBubble}, \citet{2024_Helmerich_Bobrick_WarpFactory}, and to a great extent \citet{2024_Fuchs_constant},  do not perform a proper $3+1$-numerics as it should be done, i.e., they suffer from an unclear IVF (see the only correct work by \citet{2024_Clough_GW} in this field of research), which caused various errors and ambiguities (cf.~also \Cref{er:ErrorsFuchsConstant}). 
\end{error}
\begin{rem}\label{rem:numErrors}
In general, it seems that  \citet{2024_Fuchs_constant,2024_Helmerich_Bobrick_WarpFactory,2021_Lentz_breakingbarrier,2021_FellHeisenberg_positive,2024_Abellan_SphericalWarpbasedBubble, 2025_Bolivar_piecewise} rushed to plot some figures without an in-depth analysis, thereby making fundamental mistakes.
\end{rem}

We finish this part by mentioning that the system of equations in \Cref{sec:Rwarp} in terms of  principal invariants may turn out to be useful for further numerical experiments.

\subsection{Energy conditions and their consequences}

As mentioned in \Cref{sec: intro}, the main focus of the majority of the works on warp-drive spacetimes has been put on the validity or violation of energy conditions, without tackling all the other pathologies we mentioned in this work. This is why we bring this section quite in the end.

We try to gather the results on the energy conditions for \Rwarp{} models.
Although, it is somewhat easy to show the violation of the NEC, hence all the other energy conditions (see \Cref{fig:ECs}), for the Alcubierre metric (see, e.g., \cite{2021_Visser_genericwarp}), the corresponding task is not as easy for a more general model. 
For \Rwarp{} models, a proof was finally given by \citet{2021_Visser_genericwarp} applying \Cref{def:WD} and \ref{R_asym} properly.

Moreover, using \eqref{eq:Euler_energy_density_Rwarp} and \eqref{eq:Etot}, \citet[Section~VII.D.1]{2021_Visser_genericwarp} show that \Rwarp{} models with $\bm \Omega \neq \bm 0$ do violate the WEC (see, however, \Cref{er:VisserShoshanyTotalMass} for a mistake which does not change this result). But, while \citet[Section~4.1]{2023_Shoshany_WarpCTC} showed that asymptotically flat warp drives with a generic lapse function, introduced in \Cref{sec: Natario with N}, with $\bm \Omega \neq \bm 0 $ indeed violate the WEC, they demonstrated that one cannot conclude the violation of the WEC when $\bm \Omega = \bm 0$, in general. This, of course, applies for the unit lapse case, i.e., \Rwarp{} models. 

This is an interesting case; indeed, since we know that the NEC is violated by \Cref{thm:NEC} below, the WEC is also violated. But, this example shows that sometimes it is easier to confer the NEC than the WEC, although the latter is weaker than the former.

On the other hand, in many works it is suggested that the superluminality is a reason for the violation of (some) energy conditions. This is however not true, in general, and the violation of energy conditions occurs way before considering the speed of the warp bubble (see, e.g., \cite[Section~4]{Visser_2004_fundamental}).

We give a definition for the pointwise\footnote{Or ``pointilliste,'' contrasted with the averaged energy conditions, or ``impressionist'' one; both introduced by \citet{2017_Curiel_PrimerEC}, borrowing terminology from the modern art movement. The former implies the latter, but not vice versa.} energy conditions in the following and refer the reader to, e.g., \cite[Section~4]{1973_EllisHawking_LSS}, and especially \cite{2017_Curiel_PrimerEC} for an in-depth treatment of this topic (see also \cite{2020_Kontou_EC} for a recent account).

The pointwise (classical) energy conditions are, according to the hierarchy (cf.~\Cref{fig:ECs}):
\begin{enumerate}
    \item The null energy condition
    \begin{equation}\tag{NEC}\label{NEC}
        T_\mnu k^\mu k^\nu \geq 0    
        \,;
        \,\,
        \text{for all null vectors} \,\,  \bm k = k^\mu \bm \p_\mu  
        \,,
    \end{equation}
    with its geometric form $\textbf{Ric}(\bm k, \bm k) \geq 0$, by the Einstein equation.

    \item The weak energy condition
    \begin{equation}\tag{WEC}\label{WEC}
        T_\mnu t^\mu t^\nu \geq 0    
        \,;
        \,\,
        \text{for all timelike vectors} \,\,  \bm t = t^\mu \bm \p_\mu  
        \,,
    \end{equation}
    with its geometric form $\bm{G}(\bm k, \bm k) \geq 0$, where $\bm G$ is the Einstein tensor \eqref{eq:EinsteinT}.

    \item The dominant energy condition
    \begin{equation}\tag{DEC}\label{DEC}
    \begin{aligned}
        T_\mnu t^\mu t^\nu \geq 0    
        \,;
        \,\,
        &\text{for all timelike vectors} \,\,  \bm t
        \,,
        \\
        &\text{such that} \,\, \tensor{T}{^\mu_\nu} t^\nu \,\, \text{is causal},
    \end{aligned}
    \end{equation}
    with its geometric equivalent statement by replacing $\bm T$ above by $\bm G$.\footnote{See, e.g., \cite{2017_Curiel_PrimerEC} for an alternative formulation.}

    \item The strong energy condition
    \begin{equation}\tag{SEC}\label{SEC}
    \begin{aligned}
        (T_\mnu - \tfrac{1}{2} T_{\mnu}) t^\mu t^\nu \geq 0    
        \,;
        \,\,
        &\text{for all timelike} 
        \\
        &\text{ vectors} \,\,  \bm t
        \,,
    \end{aligned}
    \end{equation}
    with its geometric form $\textbf{Ric}(\bm t, \bm t) \geq 0$, by the Einstein equation.
\end{enumerate}
\begin{rem}\label{rem:ECeffective}
On top of the physical and geometric definitions of the energy conditions, there is an \emph{effective} definition which is based on the quantitative relations between energy density $\hat{\epsilon}$ and principal pressures $\hat{p}_{\hat \mu}$ in an orthonormal frame where the energy-momentum tensor takes special forms, most importantly, if it is of Hawking--Ellis type I, meaning the energy-momentum tensor is diagonal with respect to the orthonormal frame (see, e.g., \cite{1973_EllisHawking_LSS,2017_Curiel_PrimerEC,2021_Visser_genericwarp} for more details). In this special case, we have
\begin{enumerate}
    \item the NEC is equivalent to 
    \begin{equation*}
        \hat{\epsilon} + \hat{p}_{\hat \mu} \geq 0 \,; \,\, \forall \, \hat \mu \in \{1,2,3\}
        \,,
    \end{equation*}
    \item the WEC is equivalent to $\hat{\epsilon} \geq 0$ and
    \begin{equation*}
        \hat{\epsilon} + \hat{p}_{\hat \mu} \geq 0 \,; \,\, \forall \, \hat \mu \in \{1,2,3\}
        \,,
    \end{equation*}
    \item the DEC is equivalent to $\hat{\epsilon} \geq 0$ and
    \begin{equation*}
        \hat{\epsilon} \geq | \hat{p}_{\hat \mu} | \,; \,\, \forall \, \hat \mu \in \{1,2,3\}
        \,,
    \end{equation*}
    \item and the SEC is equivalent to $\hat{\epsilon} + \sum_{\hat \mu} \hat{p}_{\hat \mu}  \geq 0$ and
    \begin{equation*}
        \hat{\epsilon} + \hat{p}_{\hat \mu} \geq 0 \,; \,\, \forall \, \hat \mu \in \{1,2,3\}
        \,.
    \end{equation*}
\end{enumerate}
\end{rem}

Before continuing to the next section, we will summarize some general misunderstandings and mistakes in the following.
\begin{error}\label{err:misundEC}
    In many works,\footnote{See also \Cref{foot:VisserCritiqueEC}.} some fundamental misunderstandings about the notion of energy conditions, in general, and the WEC, in particular, can be found (cf.~\Cref{err:EulerianWEC}). For example, in a recent work by \citet{2024_Abellan_SphericalWarpbasedBubble} one reads: ``Energy conditions are expressions that can be used to relate the quantities that appear in the energy-momentum tensor $T_{\mu\nu}$ . As is well known, these expressions are used more as a rough guideline than as a constraint on these quantities,'' which evidently points at one of these misunderstandings (see also \Cref{er:AbellanIneq}), and another example of violation of \Cref{principle}.
\end{error}
\begin{error}\label{er:AbellanIneq}
    \citet{2024_Abellan_SphericalWarpbasedBubble} make a mistake when deriving the equivalent inequalities to the energy conditions. Those inequalities do not hold as they are written. We do not want to analyze them, since the work suffers from a bigger error, i.e., \Cref{err:AbellanECbig}. However, we highlight the inconsistency with the subsequent work by the same authors for the case of vanishing ``heat flow''  \cite{2025_Bolivar_piecewise}.
\end{error}
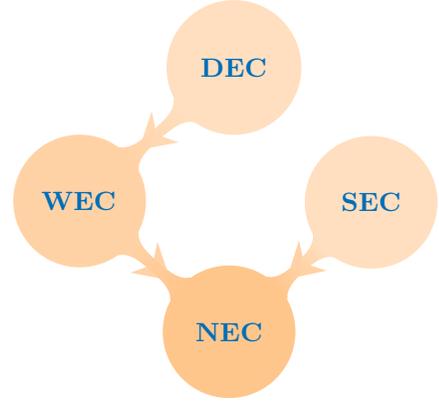
\begin{figure}[!t]
\centering
\begin{tikzpicture}[mindmap, outer sep=0pt, scale = .8]
\path[mindmap, concept color=orange!45,minimum size=1.5cm,  text=NavyBlue, every node/.append style={font=\normalsize}]
  node[concept, text width=5em] (NEC) {\textbf{\ref{NEC}}}
  child[grow=150,concept color=orange!35] {node[concept, minimum size=1.5cm, text width=5em, shift={(-10:1.5)}] (WEC) {\textbf{\ref{WEC}}}
  child[grow=50,concept color=orange!25, text width=5em, shift={(0:.7)}] {node[concept] (DEC) {\textbf{\ref{DEC}}}
  }
  }
  child[grow=30, concept color=orange!25, shift={(10:-2)}] {node[concept, minimum size=1.5cm, text width=5em] (SEC) {\textbf{\ref{SEC}}}};

\begin{pgfonlayer}{background}    
  \draw [concept connection,->,orange!35,shorten >=5pt,-{Stealth[angle=60:1pt 5]}] 
    (SEC) to (NEC);

  \draw [concept connection,->,orange!35,shorten >=5pt,-{Stealth[angle=60:1pt 5]}] 
    (WEC) to (NEC);

  \draw [concept connection,->,orange!30,shorten >=5pt,-{Stealth[angle=60:1pt 5]}] 
    (DEC) to (WEC);

\end{pgfonlayer}
\end{tikzpicture}
\caption{The hierarchy of the pointwise energy conditions with implications depicted by arrows}
\label{fig:ECs}
\end{figure}
\begin{error}\label{err:AbellanECbig}
    \citet{2024_Abellan_SphericalWarpbasedBubble} make another fundamental mistake: for their proposed model, they show that the NEC and the WEC are violated, but the DEC is not (see \Cref{fig:ECs})! This clearly points at a deep flaw not only in their analysis and numerics (\Cref{rem:numErrors}), but also in their conception (cf.~\Cref{err:misundEC}).\footnote{Even setting aside this fundamental issue (and other more important ones mentioned in the present work), the main point of \citet{2024_Abellan_SphericalWarpbasedBubble} remains unclear, as their model suffers from the violation of all energy conditions.}
\end{error}
Sometimes (pointwise) energy conditions are investigated in a wrong way where either they have no meaning, or they are not valid. We recognize at least two instances of such cases, which we formulate in Remarks~\ref{rem:distrEC} and \ref{rem:junctionEC} and bring the examples in Errors~\ref{err:distrECerror} and \ref{err:shellEC}, respectively.
\begin{rem}\label{rem:distrEC}
If the spacetime metric is below $C^2$, i.e., twice continuously differentiable, then the ``classical'' energy conditions lose their meaning as the Einstein tensor $\bm G$, hence the energy-momentum tensor $\bm T = (8 \pi G)^{-1} \bm G$ are not well-defined in this case. In such cases one needs a distributional version of the classical energy conditions (see \cite[Section~4.4]{2022_Steinbauer_lowregular} for an excellent example in the context of singularity theorems for metrics of low regularity). See \Cref{err:distrECerror}, where this fact is neglected.
\end{rem}
\begin{error}\label{err:distrECerror}
The metric of the model proposed by \citet{2021_Lentz_breakingbarrier} (and at least the first example by \citet{2021_FellHeisenberg_positive}) is clearly below $C^2$.  Therefore, even if the proposal by \citet{2021_Lentz_breakingbarrier} was physically sound (which is certainly not the case), by \Cref{rem:distrEC} the classical energy conditions have no meaning in this situation. 
\end{error}
\begin{rem}\label{rem:junctionEC}
    If a thin shell (or other similar configuration) is considered, then the energy conditions should be evaluated for the induced energy-momentum tensor on the junction surface (cf., e.g., \cite{1995_Goldwirth_CommentJunctionEnergy}, and for a recent account of this topic \cite{2023_Maeda_EnergyConditionsNontimelike}). See \Cref{err:shellEC}, where this fact is ignored. 
\end{rem}
\begin{error}\label{err:shellEC}
    \citet{2025_Bolivar_piecewise} consider a thin shell in a static setting,\footnote{It is not clear what the point of this setting is, as it is not the goal of a warp drive, see \Cref{err:timelikeKilling}.} and investigate the (pointwise) energy conditions of the energy-momentum tensor resulting from the Einstein tensor, thereby neglecting the fact that the induced energy-momentum tensor on the junction surface should be analyzed, see \Cref{rem:junctionEC}.
\end{error}

After this general consideration, we start with the WEC since almost all works focus  primarily on it due to its most intuitive character.

\subsubsection{Weak energy condition}
\label{sec:WEC}

Energy conditions are, in general, conditions required to ensure that ``energy should be positive'' \cite{2017_Curiel_PrimerEC}. Clearly, the simplest and the most intuitive energy condition that reflects this requirement is the \ref{WEC}. Therefore, the violation of the WEC is equivalent to the existence of negative energy and mass. Accordingly, the majority of works in the field of warp-drive spacetimes are focused on the validity of the \ref{WEC}, particularly those that aim to resolve its violation.

Before turning to Errors, we gather some general results in terms of theorems. We start by recalling that \Rwarp{} models generically violate either the \ref{WEC} or the \ref{SEC}.
\begin{theorem}[\natario{}, 2002]\label{thm:NatarioWSEC}
    \Rwarp{} models violate either the WEC or the SEC.
\end{theorem}
\begin{proof}
    See \citet[Theorem~1.7]{Natario_2002_warp} for the proof.
\end{proof}
\begin{theorem}\label{thm: natario WEC}
The \natario{} zero-expansion warp drive  violates the WEC.   
\end{theorem}
\begin{proof}
    If $\CK = \Ik = 0$, then $\IIk = - \hnorm{\bm \CK}^2 / 2 \leq 0$, which shows the violation of the WEC by the Hamilton constraint \eqref{eq: E for R1+R3}, unless ${\bm \CK} \equiv \bm 0$ identically, which yields exactly Minkowski spacetime that can be easily concluded from  \eqref{eq:evolEqRwarp} together with \eqref{eq:evol_K_Lie}.
\end{proof}
In fact, \Cref{thm: natario WEC} is a special corollary of the following theorem.
\begin{theorem}\label{thm: zero expansion DEC}
If a non-vacuum spacetime admits a vanishing mean-curvature foliation whose Ricci scalar satisfies $\CR < 2 \cc + \hnorm{\bm \CK}^2$, then  the WEC is violated.
\end{theorem}
\begin{proof}
If $\CK = 0$, then  the Hamilton constraint \eqref{eq: Hamilton constraint II} implies that $E < 0$, when $\CR < 2 \cc + \hnorm{\bm \CK}^2$. If $\CR = 0 = \cc$, one recovers \Cref{thm: natario WEC} for $N = 1$.
\end{proof}
\begin{error}
    Sometimes the fact that most of \Rwarp{} models, in particular Alcubierre warp drive, violate the WEC is associated to the superluminal motion. However, as it is clear from \eqref{eq: energy den in Alcubierre}, this fact is independent of the value of $V_s$, and it is rather a property of such spacetimes. For reasons of brevity, we do not list the works that made this mistake.
\end{error}
\begin{error}\label{er:truncationWEC}
\citet[Section~3.1]{2021_Bobrick_physicalwarp} claim that the fact that the WEC is violated in the \alcub{} and \natario{} spacetimes is due to the truncation of the warp field and spherical symmetry of these spacetimes. As mentioned in Errors~\ref{er:truncation} and \ref{er:sph_symm}, both the latter statements and the conclusion are flawed. This idea was repeated uncritically by \citet{2021_FellHeisenberg_positive}, as it is already pointed out correctly by \citet[Appendix~A]{2021_Visser_genericwarp}.
\end{error}
As explained and criticized duly by \citet{2021_Visser_genericwarp}, positivity of the Eulerian energy \eqref{eq:eulerianenergy} does not guarantee that the WEC is satisfied, since by definition \ref{WEC} is satisfied only when the inequality holds for all timelike vectors.
\begin{error}\label{err:EulerianWEC}
In many works,\footnote{\citet[Section~V]{2021_Visser_genericwarp} rightly raised this point concerning several of the works listed here. We mention them again, but we add to them the instances which were not published at the time.\label{foot:VisserCritiqueEC}} 
including those by \citet{2021_Bobrick_physicalwarp,2021_FellHeisenberg_positive,2021_Lentz_breakingbarrier, 2024_Garattini_BHWD}, \citet{Abellan_2023_spherical, Abellan_2023_anisotropic} (see \cite{2021_Visser_genericwarp} for more critiques on works on which \citet{Abellan_2023_spherical, Abellan_2023_anisotropic} are based) take the WEC to be equivalent to the positivity of the Eulerian energy \eqref{eq:eulerianenergy}, without proving that the WEC is in fact satisfied.
\end{error} 
\begin{error}\label{er:VisserShoshanyTotalMass}
Sometimes $\Etot$ from \eqref{eq:Etot} is used to speculate why the WEC is violated. For example,
\citet[Section~VII.C]{2021_Visser_genericwarp} argue that the WEC is violated since $\Etot = 0$ should hold. But, as we mentioned in \Cref{sec:notionofMassEnergy} (see also \Cref{rem:ADMmassNewton}), $\Etot$ is not related to the ADM mass, in general. Moreover, as mentioned in \Cref{er:inadequateADMVisser}, the given falloff for the metric components is not enough to conclude the asymptotic flatness.
This error is therefore a combination of two  mistakes regarding \Cref{def: asympt. flatness} and the notion of energy.    
\end{error}
\begin{error}\label{er:LentzargueVisser}
    \citet[Section~4]{2021_Lentz_hyperfast} attempts to argue that the argument by \citet{2021_Visser_genericwarp} illustrated in \Cref{er:VisserShoshanyTotalMass} (which is indeed incorrect) cannot be applied to the model proposed by \citet{2021_Lentz_breakingbarrier}, because the Eulerian energy is not smooth, hence one cannot apply Stokes' theorem as usual. This argument suffers from many mistakes: $(i)$ the Eulerian energy therein is not continuous, it is at most piecewise continuous or differentiable, $(ii)$ even if the arguing with $\Etot$ given in \eqref{eq:Etot} was in order, the negative term is absent in this model since the coordinate vorticity vanishes, $(iii)$ and finally, even if the Eulerian energy in that sense is positive, this, by \Cref{rem:distrEC}, has nothing to do with the validity of the classical WEC, since it is not well-defined (even assuming its validity, the author does not show it). Lastly, even if one neglects all these problems, and assumes that the derivation of equations are correct, as mentioned in \Cref{err:EulerianWEC}, \citet{2021_Lentz_breakingbarrier} offers no proof for the claim that the WEC is not violated.
\end{error}

\subsubsection{Null energy condition}
\label{sec:NEC}

Although the violation of the NEC is rather trivial for simple \Rwarp{} models, such as the Alcubierre and the zero-expansion \natario{} warp drives (see \cite{Alcubierre_2017_warp_basics, 2021_Visser_genericwarp}), a direct proof for the violation of the NEC for \Rwarp{} models was given by \citet{2021_Visser_genericwarp} only recently.
\begin{theorem}[\citet{2021_Visser_genericwarp}]
\label{thm:NEC}
    \Rwarp{} models violate the NEC.
\end{theorem}

\begin{rem}\label{rem:NECbubble}
    \Cref{def:WD}, i.e., the characteristic property of the warp bubble together with \ref{R_asym}, play a fundamental role in proving \Cref{thm:NEC}. Indeed, a generic metric of the form \eqref{eq:Rwarp_metric} does \emph{not}, in general, violate the WEC, let alone the NEC if these two assumptions, especially \Cref{def:WD} are not imposed, i.e., if the metric is not supposed to describe a warp-drive spacetime.
    An example of this kind was observed by \citet[Section~4.3]{2023_Shoshany_WarpCTC} where $E = 0 = \Omega_{ij}$, in which case the violation of the WEC cannot be concluded without further considerations.
\end{rem}
It should be noted that using global techniques and a definition for superluminal travel that requires a path, in a specific sense, to be faster than all neighboring ones, and relies on the assumption of the \emph{generic condition} (cf.~\cite[Section~4.4]{1973_EllisHawking_LSS} for a definition, see also \cite[Section~4.2]{2023_Shoshany_WarpCTC} for the related discussion), \citet{Olum_1998_superluminal} shows the violation of the NEC, although the result is stated as the violation of the WEC. Of course, by the hierarchy of the energy conditions (see, \Cref{fig:ECs}), if the NEC is violated, then all the other energy conditions, in particular, the WEC is violated. Nonetheless, \citet{Olum_1998_superluminal} actually showed the violation of the NEC, and not merely the violation of the WEC, as reported throughout the literature, with only a few exceptions.

However, as shown by \citet[Section~4.2]{2023_Shoshany_WarpCTC}, none of the conditions mentioned above is satisfied for a wide variety of \Rwarp{} models that have a shape function similar to \eqref{eq:AlcShape}, e.g., the Alcubierre and Van Den Broeck warp drives introduced in \Cref{sec: Alcubierre} and \Cref{sec: van den Broeck}, respectively. Thus, one \emph{cannot} apply the proof by \citet{Olum_1998_superluminal} to the \Rwarp{} models.

\begin{rem}\label{rem:SchusterNEC}
    As explained in \Cref{err:Schustermain}, the example given by \citet{Visser_2023_adm} does not provide an example for violation of ``generic'' \natario{} warp-drive spacetimes.
\end{rem}

\subsubsection{Dominant energy condition}

In the notation used in \Cref{sec: energy_momentum_tensor}, the \ref{DEC} is satisfied iff (see, e.g., \cite{Schoen_Yau_1981_positive}, \cite{Witten_1981_positive},  \cite[Section~8.3.4]{2012_Gourgoulhon_formalism}, and \cite[Section~3.7.3]{Straumann_2013_GR})
\begin{equation}\label{eq:DECineq}
    E \geq \hnorm{\bm J}
\quad
\text{and}
\quad
    E \geq 0
\,,
\end{equation}
where $\bm J = J^i \bm \p_i$ (see also \Cref{rem:ECeffective}).

We start with the following trivial corollary to \Cref{thm: positive mass theorem}, which says that \Rwarp{} models, as they are meant in the literature (i.e., being asymptotically flat), violate the DEC.
\begin{cor}\label{cor:asymflatDEC}
    \Rwarp{} models violate the DEC.
\end{cor}
\begin{proof}
    By the positive energy theorem \ref{thm: positive mass theorem}, we know that an \Rwarp{} model that satisfies the DEC (due to asymptotic flatness), is identically Minkowski. But, if one insists that it is not Minkowski, then this means that the DEC should be violated, by de Morgan's laws of negation of conjunction.
\end{proof}

The \ref{DEC}, in the form \eqref{eq:DECineq}, is defined by the requirement that the measured momentum density (or energy-momentum flux) must be causal and oriented in the direction of an observer's proper time, a principle traditionally interpreted as a prohibition against the superluminal propagation of energy (cf., e.g., \cite[Section~9.2]{1984_Wald_GR}, \cite[Section~2.1]{2017_Curiel_PrimerEC}, and \cite[Section~2.1]{2020_Kontou_EC}). 
This condition is foundational for deriving ``energy estimates,'' which is used to prove that a system has a well-posed IVF (or to prove the positive energy theorem; cf.~\Cref{thm: positive mass theorem}).  However, the DEC is not strictly necessary for the IVF results. We explain this in the following Remark.
\begin{rem}\label{rem:Geroch}
In an intriguing work, \citet{2011_Geroch_FTL} specifically proposes a \emph{democracy of causal cones}, arguing that different physical systems (modeled as symmetric hyperbolic equations) may possess their own unique causal cones that exist outside the standard light cone without invalidating the structural components of relativity. That is, the causal cones of some yet-unknown matter fields could be wider than that of the light cone. However, current experience dictates that the light cone wins the election of the democracy of causal cones mentioned above, largely because many common systems (such as electromagnetism and Dirac fields) happen to share that specific ``commonality'' (see \citet[Section~3]{2011_Geroch_FTL}).
Nonetheless, upon this work, \citet{2014_Earman_NoSuperluminal} argues that a well-posed IVF is a more appropriate method for expressing causal limits, particularly because physical systems can maintain a maximum propagation speed equal to that of light even when they violate the DEC (see \citet{2017_Curiel_PrimerEC} for more details).
\end{rem}

\begin{error}\label{er:ruleofthumb}
\citet{2024_Fuchs_constant} talk about a ``rule of thumb'' that 
``the Eulerian momentum flux and pressures should be less than the energy density to satisfy the energy conditions,'' and refer to their previous work \cite{2024_Helmerich_Bobrick_WarpFactory}, implying that they found this ``rule of thumb'' by numerical trail and error in a heuristic way (cf.~\Cref{err:syntax}). However, this is a well-known fact just explained above.
A similar statement was done by \citet{2021_Lentz_breakingbarrier}, where the author talks about the ``momentum conditions,'' and compares them with the WEC. However, ``momentum conditions'' are nothing but the momentum constraint equation \eqref{eq:moment_const}, and have different epistemic status than the WEC or any other energy conditions.
\end{error}

\subsubsection{Strong energy condition}

Unlike other energy conditions, the \ref{SEC} is routinely violated on cosmological scales by hypotheses in the standard model of cosmology  such as inflation and the dark energy currently driving the Universe's accelerated expansion.\footnote{The SEC is described as potentially ``gerrymandered'' by the relativity community because its physical significance is ``obscure at best'' and its primary impact is largely limited to proving singularity theorems, see \cite[Section~3.1]{2017_Curiel_PrimerEC}.} Therefore, the violation of the SEC might not be considered as bad as the violation of other energy conditions. 

In the case of \Rwarp{} models, \cref{thm:NatarioWSEC} offers a straightforward conclusion: if one insists that the WEC should hold, then the SEC is certainly violated for \Rwarp{} models.

\subsection{Miscellaneous}

In this section, we gather some Errors which we could not systematically classify above.
\begin{error}\label{er:Helmholtz}
\citet[Equation~(6)]{2021_FellHeisenberg_positive} apply the Helmholtz decomposition theorem to the shift vector to decompose it into a scalar and a solenoidal field on $\R^3$, without specifying how regular these fields should be; indeed, the Helmholtz decomposition theorem requires the scalar field to be $C^2$, as they have in mind a classical (and not a distributional) formulation (cf.~\Cref{rem:distrEC}). However, later in \cite[Equation~(10)]{2021_FellHeisenberg_positive} they introduce a scalar field that is not even $C^1$, as they explicitly mention this (see \cite[Appendix~C]{2021_Visser_genericwarp} for more critiques), though they introduce a $C^2$ function afterwards ``in the interest of constructing a physically interesting'' model  (on $\R^3 \setminus \{\bm 0\}$). But, this is independent of physicality, it is just self-inconsistent, thereby constituting another instance of a violation of \Cref{principle}.
\end{error}

\begin{error}\label{err:syntax}
Many works suffer from vague and wrong presentation including \citet{2021_Lentz_breakingbarrier}, \citet{2021_Bobrick_physicalwarp}, \citet{2024_Garattini_BHWD} (and other works by the same authors), \citet{2024_Helmerich_Bobrick_WarpFactory}, and \citet{2024_Fuchs_constant} phrase main ideas solely in a heuristic way in words without any formal description,
and make several mistakes which, at a first glance, might look like typographical ones, however, they reveal a deeper level of misunderstanding; for example, 
\begin{enumerate}
    \item taking a naive Newtonian point of view (as pointed out by \citet{2021_Visser_genericwarp}), in particular regarding the definitions of total energy and momentum,
    \item the covariant $4$-velocity (unit normal vector) is taken to be past-directed without any explanation, that is indeed inconsistent in their context,
    \item throughout the work by \citet{2024_Fuchs_constant} a blatant mistake is repeated, hinting at a basic misunderstanding, where a vector (or covector) is set to be equal to a scalar field (e.g., $u_i/u^0 = v_s$, in their notation, which would imply $v_s \equiv 0$ for the Alcubierre warp drive, see \Cref{er:ErrorsFuchsConstant}),
    \item the metric given by \citet[Equation~(16)]{2024_Fuchs_constant} is not spherically symmetric,  
    \item mixing the orthonormal and coordinate bases by \citet[Equation~(10)]{2024_Fuchs_constant},
    \item identifying the flatness with vanishing of spatial derivatives of the metric by \citet{2024_Fuchs_constant}.
\end{enumerate}
\end{error}
\begin{error}\label{er:White}
    In a recent work, \citet{2025_White_CylindricalNacelle} concerned themselves with the ``engineering-oriented design'' aspects of the warp bubble that lead to a different form of warp bubble than the Alcubierre one, thereby neglecting all real pathologies of such model discussed in the present work and elsewhere, as if this was the major problem. This points again (at best) at a Newtonian understanding of the problem, hence making several mistakes and creating some misunderstandings, e.g., $(i)$ giving a wrong expression for the extrinsic curvature of the Alcubierre warp drive written below \cite[Equation~(14)] {2025_White_CylindricalNacelle} (cf.~\eqref{eq:KijAlc}), $(ii)$ discussing the role of the shear tensor in a flat region  as ``a diagnostic of interior flatness,'' $(iii)$ presuming the existence or the ability of producing ``exotic energy,'' $(iv)$ vague formulation, e.g., ``the Lorentz factor between observers at rest in the ADM slicing and the physical flow of spacetime described by the shift field,'' and $(v)$ discussing the role of the ``York time'' and distinguishing it from the  mean curvature (trace of the extrinsic curvature), which might be important in the so-called constant mean curvature (CMC) slicing, that is irrelevant in this context anyway.
\end{error}
\begin{error}\label{er:cloughFallacy}
    Although the work by \citet{2024_Clough_GW} represents perhaps the most advanced numerical contribution (and the only sound numerical GR work) in the field of warp drive, it suffers from a fallacy: the ``mind projection fallacy'' \cite{1990_Jaynes_probabilitylogic}; indeed, the authors claim that the gravitational waves, resulting from the ``containment failure'' of an Alcubierre warp drive used by a ``post-warp civilization,'' could be potentially detected by future detectors targeting higher frequencies. But, if a ``pre-warp civilization'' (i.e., us, and in particular \citet{2024_Clough_GW}) can foresee the collapse of the Alcubierre warp drive, why a post-warp civilization, which presumably would possess a better theory of gravity, could not know this fact, and go beyond to construct it? Therefore, the idea that future gravitational-wave detectors should search for such signals, as a sign for ``extraterrestrial life,'' must not be taken seriously, even if one thinks that the Alcubierre warp drive would be physical despite all fundamental pathologies explained in this work and elsewhere.
\end{error}
%

\subsection{The warp bubble and its mechanism}
\label{sec:bubble_mech}

The warp bubble and its mechanism proposed in the recent claims about physical warp drives by \citet{2021_Lentz_breakingbarrier, 2021_Bobrick_physicalwarp, 2021_FellHeisenberg_positive} and others, are at best verbally formulated, and when they try to model it mathematically, the result sometimes even fails to be a solution to the Einstein equations (see, e.g., Errors~\ref{er:hyperbolicshitLentz} and \ref{err:TOV}).

In this section, we point out some of these misconceptions.
\begin{rem}
    In the majority of works, the whole warp-drive mechanism is associated to the expansion tensor (or the extrinsic curvature) and its rate. However, the zero-expansion \natario{} warp drive could serve as a hint that it is not clear what the mechanism is supposed to be; indeed, since if the expansion could be zero, then one can ask: ``what is driving the warp bubble?''
\end{rem}
\begin{error}\label{er:soliton}
\citet{2021_Lentz_breakingbarrier} and \citet{2021_FellHeisenberg_positive} identify warp drives with ``solitons.'' However, it is not clear what is exactly meant by a soliton; a ``classic'' solitary wave, or a gravitational soliton? In either case, a warp drive \emph{cannot} be, or at least, has \emph{not} been shown to be associated with solitons; indeed, from a classical perspective, they lack several fundamental properties, such as the general ability $(i)$ to maintain their shape if they are characterized by the expansion tensor and its rate, $(ii)$ hence, to be stable, $(iii)$ to have maximal set of conserved quantities, and from the gravitational point of view these solutions have not been generated by an \emph{inverse scattering method} applied to Einstein's equations, regardless of other (yet-to-be-shown) properties (cf., e.g., \cite{2001_Belinski_GravitationalSolitons} for more details on gravitational solitons), unless they are supposed to be static or stationary vacuum solutions, i.e., black hole solutions, in which case they cease to be a warp drive  (cf.~\Cref{err:timelikeKilling}). This provides yet another example of violation of \Cref{principle}.
\end{error}
\begin{error}\label{er:fuchsgeodesic}
    \citet{2024_Fuchs_constant} call a solution to the $3+1$-form of the geodesic equation warp drive, without exactly saying in what way it differs from the motion of a particle observed from a given reference frame, apart from a very heuristically and verbally phrased program (cf.~Errors \ref{er:ErrorsFuchsConstant}  and \ref{err:syntax}). 
\end{error}

Therefore, it is not clear what a warp bubble, hence a warp drive, is supposed to be (\Cref{def:WD} provides merely a general geometric description for a warp bubble) and how it is supposed to work, quite on the contrary to the statement done by \citet{2024_Clough_GW} in their abstract that ``despite originating in science fiction, warp drives have a concrete description in general relativity.''

\section{Discussion and conclusions}
\label{sec:conclusion}

The goal of the present work was to demystify some ambiguities based on fundamental misconceptions, misunderstandings, and mistakes in the field of warp-drive spacetimes, in terms of Errors~\ref{er:alcgh}--\ref{er:fuchsgeodesic}, to show that many pseudo-problems are caused by violation of \Cref{principle}, i.e., they have emerged from an inconsistent and vague formulation of the problem as well as improper application of GR.

These errors occurred especially in recent works where it is claimed that a physical warp drive is feasible. We summarize the most important ones in the following (one needs to add to them the errors reported by \citet{2021_Visser_genericwarp}):
\begin{enumerate}
    \item Two works by \citet{2021_Lentz_breakingbarrier,2021_Lentz_hyperfast}: these works initiated the idea of possibility of constructing physical warp drives that suffer from many errors which we listed in Errors \ref{er:imposs_GH}, \ref{er:LentzElliptic}, \ref{err:ADMtotMass},  \ref{err:contEq},  \ref{er:hyperbolicshitLentz}, \ref{err:hypshiftzeroJ}, \ref{er:numericalmagic},  \ref{err:distrECerror}, \ref{err:EulerianWEC},  \ref{er:LentzargueVisser}, \ref{er:ruleofthumb},  \ref{err:syntax}, \ref{er:soliton}.
    
    \item Series of works by \citet{2021_Bobrick_physicalwarp,2024_Helmerich_Bobrick_WarpFactory}, \citet{2024_Fuchs_constant}: these papers encompass the widest variety of errors which can be easily refuted by applying GR adequately. We listed them in Errors~\ref{er:imposs_GH}, \ref{er:truncation}, \ref{er:ADMviolEC}, \ref{err:ADMtotMass}, \ref{er:sph_symm}, \ref{err:timelikeKilling}, \ref{er:coordtrafo},  \ref{er:ErrorsFuchsConstant},  \ref{err:contEq}, \ref{err:TOV}, \ref{er:numericalmagic}, \ref{er:IVF}, \ref{er:truncationWEC}, \ref{err:EulerianWEC}, \ref{er:ruleofthumb},  \ref{err:syntax},  \ref{er:fuchsgeodesic}.
    
    \item Series of works by \citet{Abellan_2023_anisotropic,2024_Abellan_SphericalWarpbasedBubble,Abellan_2023_spherical}, and \citet{2025_Bolivar_piecewise}: a naive evaluation of the components of the Einstein tensor is performed in these works without addressing the main pathologies of such models while making various errors (see also critiques by \citet{2021_Visser_genericwarp} on the works on which this series of works is based): Errors~\ref{err:timelikeKilling}, \ref{er:numericalmagic}, \ref{er:IVF}, \ref{err:misundEC}, \ref{er:AbellanIneq}, \ref{err:AbellanECbig}, \ref{err:EulerianWEC}.
    
    \item The work by \citet{2021_FellHeisenberg_positive}: this work also contains several errors that we addressed in Errors~\ref{er:truncation}, \ref{err:ADMtotMass}, \ref{err:timelikeKilling}, \ref{er:coordtrafo}, \ref{er:numericalmagic}, \ref{err:distrECerror}, \ref{er:truncationWEC}, \ref{err:EulerianWEC}, \ref{er:Helmholtz}, \ref{er:soliton}.
\end{enumerate}
One has to add to this list works that appeared in the aerospace communities. We only addressed one written by members of one of these communities 
in \Cref{er:White}.

In all these works, the given descriptions are very vague, and mostly verbally, not even formally, prescribed (cf.~\Cref{err:syntax}), especially in works by \citet{2021_Lentz_breakingbarrier, 2021_Bobrick_physicalwarp}, \citet{2021_FellHeisenberg_positive}, \citet{2024_Helmerich_Bobrick_WarpFactory}, and \citet{2024_Fuchs_constant}, \citet{Abellan_2023_anisotropic,Abellan_2023_spherical,2025_Bolivar_piecewise}, which caused an unrealistic hype in both scientific and nonscientific communities despite all Errors mentioned in the present paper and reported elsewhere, mainly by \citet{2021_Visser_genericwarp}, where the authors proved, in particular, \Cref{thm:NEC} showing that \Rwarp{} models are fundamentally unphysical by violating the NEC. This leads us to recognize at least two issues on two different but related levels:
\begin{enumerate}
    \item These works have been embraced by some aerospace communities, in spite of the fact that we do not even know what a warp bubble and its mechanism are supposed to be (cf.~\Cref{sec:bubble_mech}), and the current proposals are in the most optimistic scenario consistent, but empty from any (known) physical content.

    \item As mentioned in \Cref{sec: intro}, we restricted ourselves to works that have been peer-reviewed and found their way into academic journals in the field of GR.

\end{enumerate}
Both these might spark potential research ideas in sociology of science, which certainly goes beyond the scope of the present work, and lies outside our expertise.

Notwithstanding, we do appreciate the curiosity to speculate and investigate warp-drive spacetimes as a  possibility for FTL travels, however with the following philosophy spelled out in, e.g., \cite{2017_book_Lobo,Visser_2004_fundamental}: warp-drive spacetimes are useful and interesting as Gedankenexperiments to understand and probe GR better because of their seemingly simple structure, but sometimes counterintuitive results.

Following this philosophy, to enrich such Gedankenexperiments and our understanding,  in \cite{2024_Barzegar_Buchert_Letter} we recently proposed to drop some of the restrictions listed in \Cref{sec: restrictions}, most importantly \ref{R_tilt}, that is allowing for a tilted warp bubble which in turn allows for nonzero covariant acceleration and vorticity. Accordingly, we proposed that \ref{R_gauge} could be relaxed as well, allowing for general kinematics. On top of this, \ref{R_curv} could also be dropped, i.e., allowing for spatial curvature in the given foliation. 

In this work, we saw that another (reasonable) restriction, i.e., \ref{R_asym}, and the assumption of global hyperbolicity can be lifted:
\begin{enumerate}
    \item Asymptotic flatness (\ref{R_asym}): this puts a strong constraint on warp-drive spacetimes as shown in \Cref{sec:asymptflatADM} (see also, e.g., \cite{1976_Tipler_CausalityAsymFlat}). While this assumption is physically reasonable (see, e.g., \cite[Setion~19.4]{1973_MTW_Gravitation}), an immediate consequence of its lift would be that the violation of the NEC might be avoided since this restriction plays a fundamental role in \Cref{thm:NEC}. An example of a reasonable scenario without \ref{R_asym} would be that of a warp drive in a cosmological context, in particular, with closed spatial hypersurfaces.
    \item Global hyperbolicity: as we showed in \Cref{thm:globhyp} this assumption forbids construction of any \Rwarp{} models. This, in a sense, supports the similar but different result by \citet{1998_Low_limits}. We also argued in \Cref{rem:GHdeter} that the omission of this assumption might be a necessary condition to allow for construction of a warp drive in terms of agency.
\end{enumerate}

We, however, warn that any attempt to construct such Gedankenexperiments should satisfy \Cref{principle}, and should not be taken for granted as fulfilling some physical requirements such as energy conditions, let alone being used for speculative technological applications.

\begin{acknowledgments}
The authors were supported in part by the European Research Council (ERC) under the European Union's Horizon 2020 research and innovation program (grant agreement ERC advanced grant 740021-ARTHUS, PI: T.B.). H.B. was also funded by the Austrian Science Fund (FWF) [Grant DOI: 10.55776/J4803]. For open access purposes, the authors have applied a CC BY public copyright license to any author-accepted manuscript version arising from this submission.
\end{acknowledgments}
%

\begin{appendix}
\section{Conformally flat metrics}
\label{sec: conform}

\numberwithin{equation}{section}

Let $\meth$ be a $3$-dimensional Riemannian metric on a smooth Riemannian manifold $\Sigma_t$ that foliates a Lorentzian manifold $\mathcal{M}$ and assume there exists a smooth function $\psi: \Sigma_t \rightarrow \R$ such that
\begin{eqnarray}
    {\meth} 
    =
    e^{2 \psi} \bm \gamma
\,,
\end{eqnarray}
where $\bm \gamma$ is a flat metric. 
Then, $\meth$ is said to be a conformal flat metric. Note that the function $\psi$ is defined on $\Sigma_t$ for each time $t$, hence a function on $\mathcal{M}$ in general. We gather some useful relations in this appendix (cf., e.g., \cite[Chapter 1, J.]{Besse_1987_Einstein}, \cite[Appendix VI]{2008_Choquet_GR}, \cite[Section 3.1]{Baumgarte_Shapiro_2010_numerical}, \cite[Chapter 7]{2012_Gourgoulhon_formalism}). The Levi-Civita connection coefficients associated to $\meth$ read
\begin{equation}
    \tensor{\Gamma[{\meth}]}{^i_{j k}}
    =
    \tensor{\Gamma[{\bm \gamma}]}{^i_{j k}}
    +
    2 \p_{(j} \psi \, \delta^i_{k)}
    -
    \gamma_{j k} \gamma^{i \ell} \p_\ell \psi 
\,,
\end{equation}
where $\tensor{\Gamma[{\bm \gamma}]}{^k_{i j}}$ are the connection coefficients of the flat metric $\bm \gamma$.
The Ricci tensor associated to $\meth$ takes the following form
\begin{equation}
\begin{aligned}
    \CR_{i j}
    &=
    -
    \left(
        \p_i \p_j \psi 
        -
        \tensor{\Gamma[{\bm \gamma}]}{^k_{i j}} \p_k \psi 
        -
        \p_i \psi \p_j \psi
    \right)
\\
&\quad
    -
    \left(
        \Delta_{{\bm \gamma}} \psi
        +
        \left| \bm \p \psi \right|^2_{\bm \gamma}
    \right)
    \gamma_{i j}
\,,
\end{aligned}
\end{equation}
where $\Delta_{\bm \gamma}$ is the Laplacian with respect to $\bm \gamma$ and  $\left| \bm \p \psi \right|^2_{\bm \gamma} = \dnorm{\bm \p \psi}^2 = \delta^{i j} \p_i \psi \p_j \psi$. Therefore, the Ricci scalar reads
\begin{equation}\label{eq: con_Ricci_sc}
    \CR
    =
    -2 e^{-2 \psi}
    \left(
        2 \Delta_{{\bm \gamma}} \psi
        +
        \left| \bm \p \psi \right|^2_{\bm \delta}
    \right)
\,.
\end{equation}
Moreover, the Laplacian $\Delta_{\meth}$ on functions $\Phi$ on $\Sigma_t$ has the following behavior
\begin{equation}\label{eq: conf Laplace}
    \Delta_{\meth} \Phi
    =
    e^{-2 \psi}
    \left(
        \Delta_{{\bm \gamma}} \Phi
        +
        \delta^{ij} \p_i \psi \p_j \Phi
    \right)
\,.
\end{equation}

\section{Solving the TOV equation}
\label{app:TOV}

In this appendix, we elaborate on \Cref{err:TOV}.
\citet[Section 3.1]{2024_Fuchs_constant} propose a process to construct a stable shell matter. We show that their process fails already in the first step as it is not a solution to the Einstein equations (the very same mistake was done by \citet{2021_Bobrick_physicalwarp} as pointed out correctly by \citet[Appendix~A]{2021_Visser_genericwarp}). 

We start with the general static, spherically symmetric metric (cf.~\cite[Section~5.8]{2004_Carroll_GR}, which is used by \citet{2024_Fuchs_constant}, see also \Cref{err:syntax})
\begin{equation}
    \bm g
    =
    - e^{2a(r)} \diffb t^2
    +
    e^{2b(r)} \diffb r^2
    +
    r^2 \bm g_{\mathbb{S}^2}
\,,
\end{equation}
where $\bm g_{\mathbb{S}^2}$ is the round metric on the unit 2-sphere. In the first step of their process, they consider a shell with inner and outer radii $R_1$ and $R_2$, respectively, with a constant density $\rho$ that gives rise to the total mass $M$, i.e.,
\begin{equation*}
    \rho
    =
    \begin{cases}
        0 \,; & 0 \leq r < R_1
        \,,\\
        3 M / [4 \pi (R_2^3 - R_1^3)]
        \,; & R_1 \leq r \leq R_2
        \,,\\
        0 \,; &   r > R_2
        \,.
    \end{cases}
\end{equation*}
Moreover, they assume that the pressure $P$ is positive and isotropic, and vanishes for $r < R_1$ and $r = R_2$. With this setup they continue their speculation and try to solve this problem numerically. However, these pieces of information are enough to conclude that their process is ill-defined; indeed, since this is not a solution to the Einstein equation. The reason is that they analyze the TOV equation without respecting the Einstein equation. To show this, we start with evaluating the time-time component of the Einstein equation which results in $e^{-2b} = 1 - 2 G m(r)/r$ with $m(r) = 4 \pi \int_{R_1}^r \rho(s) s^2 \diff s$ (cf.~\cite[Equation~(5.143)]{2004_Carroll_GR}). The first Darmois--Lanczos--Israel junction condition (cf., e.g., \cite[Section~3.7]{2004_Poisson_toolkit} and \cite{1996_Mansouri_shell}) implies that the induced metric on the shell surface should match that of Schwarzschild at $r = R_2$, hence
\begin{equation}
    e^{2a(R_2)}
    =
    1 - 2 \frac{G M}{R_2}
\,,
\end{equation}
since $M = m(R_2)$ is the Schwarzschild mass, and it should match the Minkowski metric at $r=R_1$, i.e., 
\begin{equation}
    a(R_1) = 0
\,.
\end{equation}
Now, integrating the TOV equation (\cite[Equation~(5.153)]{2004_Carroll_GR})
\begin{equation}
    \frac{\diff P}{\diff r}
    =
    -
    (\rho + P) \frac{\diff a}{\diff r}
\,,
\end{equation}
results in
\begin{equation}
    \ln \left[ \frac{\rho}{\rho + P(r)} \right]
    =
    a(r) - a(R_2)
\,,
\end{equation}
since $\rho$ is constant. Evaluating the above at $r = R_1$ we find
\begin{equation}
    \frac{P(R_1)}{\rho}
    =
    \sqrt{1 - 2 \frac{G M}{R_2}}
    -
    1
    <
    0
\,,
\end{equation}
which not only violates the SEC (and possibly other energy conditions), but also contradicts their assumption of positive pressure in the first step (see \cite[Figure~4]{2024_Fuchs_constant}, where there is a misprint for the initial quantities and the ``smoothed'' ones).
\end{appendix}

\bibliography{Warp_review_bib}

\end{document}